\documentclass[10pt]{article}

\usepackage{eqnarray,amsmath,amssymb,amsthm,color,epsfig,mathrsfs}
\usepackage{amsbsy,array,float,caption,mathtools}
\usepackage{multirow,hhline,float,soul}
\usepackage[bookmarksnumbered=true]{hyperref}

\renewcommand{\labelenumi}{\alph{enumi}.}

\newtheorem{theorem}{Theorem}

\newtheorem{lemma}[theorem]{Lemma}

\newtheorem{proposition}[theorem]{Proposition}

\newcounter{spslist}

\newcommand{\col}[3]{ \renewcommand{\arraystretch}{#1}
                \left[\!\! \begin{array}{c} #2 \\ #3 \end{array} \!\!\right] }

\newcommand{\mat}[5]{ \renewcommand{\arraystretch}{#1}
                    \left[\!\! \begin{array}{cc}
                            #2 & #3 \\
                            #4 & #5 \end{array} \!\!\right] }

\newcounter{geqncount}
    {\refstepcounter{equation}%
     \setcounter{geqncount}{\value{equation}}%
     \setcounter{equation}{0}%
  }%
    {\setcounter{equation}{\value{geqncount}}}

\newcommand{\lfe}{{-e}}
\newcommand{\lfp}{{-p}}
\newcommand{\lfl}{{-\ell}}
\newcommand{\rte}{{+e}}
\newcommand{\rtp}{{+p}}

\newcommand{\kplus}{k_+}
\newcommand{\kzero}{k_0}
\newcommand{\kone}{k_1}

\newcommand{\Pl}{P_{\tiny -}} 
\newcommand{\Ple}{P_{-e}} 
\newcommand{\vle}{v_{-e}} 
\newcommand{\vll}{v_{-\ell}} 
\newcommand{\kle}{k_{-e}} 
\newcommand{\wle}{w_{-e}} 
\newcommand{\Plp}{P_{-p}} 
\newcommand{\vlp}{v_{-p}} 
\newcommand{\wlp}{w_{-p}} 
\newcommand{\klp}{k_{-p}} 
\renewcommand{\Pr}{P_+} 
\newcommand{\Pre}{P_{+e}} 
\newcommand{\vre}{v_{+e}} 
\newcommand{\wre}{w_{+e}} 
\newcommand{\kre}{k_{+e}} 
\newcommand{\Prp}{P_{+p}} 
\newcommand{\vrp}{v_{+p}} 
\newcommand{\wrp}{w_{+p}} 
\newcommand{\krp}{k_{+p}} 

\newcommand{\Plrp}{P_{\mp p}} 

\newcommand{\Vo}{\mathring{V}}
\newcommand{\Vom}{\mathring{V}_-}
\newcommand{\Vop}{\mathring{V}_+}
\newcommand{\To}{\mathring{T}}
\newcommand{\Do}{\mathring{D}}

\newcommand{\rzero}{{\hspace{-0.7pt}+\hspace{-0.7pt}0}}
\newcommand{\rone}{{\hspace{-0.7pt}+\hspace{-0.7pt}1}}
\newcommand{\rtwo}{{\hspace{-0.7pt}+\hspace{-0.7pt}2}}
\newcommand{\lzero}{{\hspace{-0.7pt}-\hspace{-0.7pt}0}}
\newcommand{\lone}{{\hspace{-0.7pt}-\hspace{-0.7pt}1}}
\newcommand{\ltwo}{{\hspace{-0.7pt}-\hspace{-0.7pt}2}}

\renewcommand{\O}{{O}}

\newcommand{\slantfrac}[2]{{\hbox{\tiny$\raisebox{0.6pt}{#1}\!/_{\!#2}$}}}
\newcommand{\third}{\slantfrac{1}{3}}
\newcommand{\thirds}{\slantfrac{2}{3}}
\newcommand{\mthird}{\slantfrac{--1}{3}}
\newcommand{\mthirds}{\slantfrac{--2}{3}}

\newcommand{\psig}{\psi^\guided}

\newcommand{\inc}{\text{\scriptsize in}}

\newcommand{\out}{\text{\scriptsize out}}
\newcommand{\guided}{\text{\itshape g}}

\newcommand{\Range}{{\mathrm{Ran}}}

\newcommand{\sspan}{{\mathrm{span}}}

\newcommand{\fadj}{{\text{\tiny $[$}*{\text{\tiny $]$}}}}

\newcommand{\xx}{{\mathbf x}}

\newcommand{\rr}{{\mathbf r}}
\newcommand{\HH}{{\mathbf H}}
\newcommand{\EE}{{\mathbf E}}
\newcommand{\BB}{{\mathbf B}}
\newcommand{\DD}{{\mathbf D}}

\newcommand{\ZZ}{\mathbb{Z}}
\newcommand{\RR}{\mathbb{R}}
\newcommand{\CC}{\mathbb{C}}

\renewcommand{\Re}{\text{Re}\,}
\renewcommand{\Im}{\mathrm{Im}\,}

\newcommand{\kk}{{\boldsymbol{\kappa}}}
\newcommand{\kw}{(\kk,\omega)}

\newcommand{\kwz}{(\kk^0,\omega^0)}

\begin{document}

\begin{center}
{\bfseries \Large  Pathological scattering by a defect in a slow-light\\ \vspace{0.7ex} periodic layered medium}
\end{center}

\vspace{0ex}

\begin{center}
{\scshape \large Stephen P. Shipman${}^\dagger$ and Aaron T. Welters${}^\ddagger$} \\
\vspace{2ex}
{\itshape ${}^\dagger$Department of Mathematics\\
Louisiana State University\\
Baton Rouge, Louisiana 70803, USA\\
\vspace{1ex}
${}^\ddagger$Department of Mathematical Sciences\\
Florida Institute of Technology\\
Melbourne, Florida 32901, USA
}
\end{center}

\vspace{2ex}
\centerline{\parbox{0.9\textwidth}{
{\bf Abstract.}\
Scattering of electromagnetic fields by a defect layer embedded in a slow-light periodically layered ambient medium exhibits phenomena markedly different from typical scattering problems.  In a slow-light periodic medium, constructed by Figotin and Vitebskiy, the energy velocity of a propagating mode in one direction slows to zero, creating a ``frozen mode" at a single frequency within a pass band, where the dispersion relation possesses a flat inflection point.  The slow-light regime is characterized by a $3\!\times\!3$ Jordan block of the log of the $4\!\times\!4$ monodromy matrix for EM fields in a periodic medium at special frequency and parallel wavevector.  The scattering problem breaks down as the 2D rightward and leftward mode spaces intersect in the frozen mode and therefore span only a 3D subspace $\Vo$ of the 4D space of EM fields.\\
\indent
\hspace{1.2em} Analysis of pathological scattering near the slow-light frequency and wavevector is based on the interaction between the flux-unitary transfer matrix $T$ across the defect layer and the projections to the rightward and leftward spaces, which blow up as Laurent-Puiseux series.  Two distinct cases emerge: the generic, non-resonant case when $T$ does not map $\Vo$ to itself and the quadratically growing mode is excited; and the resonant case, when $\Vo$ is invariant under $T$ and a guided frozen mode is resonantly excited. 
}}

\vspace{3ex}
\noindent
\begin{mbox}
{\bf Revised:} \today
\end{mbox}
\vspace{3ex}

\hrule
\vspace{1.1ex}

\section{Introduction}

When a periodic layered medium carries electromagnetic fields near a frequency at which the energy velocity of a wave across the layers vanishes, it is said to be operating in the slow-light regime; or it is simply called a slow-light medium.
The typical slow-light regime for a layered medium occurs at the edge of a propagation frequency band, where the dispersion relation between frequency and transverse wavenumber (axial dispersion relation) is flat and the energy velocity of the electromagnetic (Bloch) mode therefore vanishes.  In this regime, transmission of energy from air into a slow-light medium is inefficient.  This inefficiency is overcome by nonreciprocal layered media that admit slow light at a frequency inside a propagation band.  Such media were devised by Figotin and Vitebskiy~\cite{FigotinVitebskiy2001,FigotinVitebskiy2003,FigotinVitebskiy2006} 
by using alternating layers of anisotropic and magnetic media.
Their work reveals the pathological nature of the scattering of a wave when it strikes the interface between a slow-light medium and another medium.  The modal components of scattered fields blow up near the frequency of vanishing energy velocity but cancel with one another to retain a bounded total field.  The reason lies in the coalescence of three (out of a total of four) Bloch electromagnetic modes into a single ``frozen mode" of zero flux across the layers, when the structure of rightward and leftward (or incoming and outgoing) modes, which is typical for scattering problems, breaks~down.

In the present investigation, a lossless slow-light medium plays the role of an ambient space, and electromagnetic fields in this medium are scattered by a defective layer, or slab (Fig.~\ref{fig:layered}).  The pathological scattering is particularly singular when the slab admits a ``guided frozen mode", the slow-light analogue of a guided slab mode that can be excited.  The modal coefficients of a scattering field admit expansions in negative and positive fractional powers---Laurent-Puiseux series---of the perturbation from the frequency and wavevector of a frozen mode.  This comes from the analytic perturbation of a $3\times3$ Jordan block associated with the $4\times4$ system of ordinary differential equations that the Maxwell equations reduce to for a layered medium.  The power laws associated with scattering of each mode from the right and left are summarized in Table~\ref{table:scattering} of sec.~\ref{sec:main} and detailed in Theorems~\ref{thm:leftscatteringreg}--\ref{thm:rightscatteringres}.  The proofs are delicate and utilize new results on the analytic perturbation of singular linear systems (Theorem~\ref{thm:localbandstrucandjordanstructure}, Proposition~\ref{prop:PuiseuxSeriesEigenvAsympt}, Lemma~\ref{lemma:Mtb}) and an interplay between the perturbed Jordan-form matrix, an indefinite inner product (the flux), and a flux-unitary transfer matrix across the slab.

Our results include an analysis of the combined effects of slow light and guided-mode resonance, which we call {\em resonant pathological scattering} (Theorems \ref{thm:leftscatteringres} and \ref{thm:rightscatteringres}).  In \cite{ShipmanWelters2012} we contrasted resonant scattering by a defect slab in two scenarios: when the ambient medium is anisotropic (but not in the slow-light regime) and when the ambient medium is operating in the slow-light regime.  Our work \cite{ShipmanWelters2013} gives a detailed analysis of the first scenario, particularly of the Fano-type resonance associated with a guided mode~\cite{FanJoannopoul2002,FanSuhJoannopoul2003,ShipmanVenakides2005,Shipman2010}.  This paper gives a detailed analysis of the second scenario.

A main result shows how the distinction between nonresonant and resonant pathological scattering is reflected cleanly in the linear-algebraic structure of the Maxwell equations at the parameters for which the ambient periodic medium admits a frozen mode (Theorem~\ref{thm:D}).  In the nonresonant case, the defect slab admits no guided frozen mode, and the quadratically growing mode associated with the $3\times3$ Jordan block can be excited.  In the resonant case, a guided frozen mode is excited, and the quadratically growing one cannot be excited by any incident field (Table~\ref{table:scattering}).

We emphasize that this article concerns a particular class of slow-light media.
Vanishing energy velocity can be due to various physical reasons, one of which is strong spatial dispersion produced by the periodicity of a composite structure.
The present study deals with periodic structures that are layered, or one-dimensional.
A slow-light regime is achieved by tuning not only the frequency $\omega$, but also the wavevector $\kk$ parallel to the layers.  At fixed values of these parameters, the Maxwell equations admit four electromagnetic modes (Floquet eigenmodes).  
In our study, at a certain value of $\kk$, a branch of the axial dispersion relation $\omega=\omega(k)$ between frequency $\omega$ and the wavenumber $k$ perpendicular to the layers has a stationary inflection point ($d\omega/dk = d^2\omega/dk^2 = 0, d^3\omega/dk^3 \not= 0$).  Thus the corresponding mode is frozen only in the sense that its energy velocity perpendicular to the layers is zero.  This is possible only for very special periodically layered structures.  A necessary condition is axial spectral asymmetry, or $\omega(k)\not=\omega(-k)$ \cite{FigotinVitebskiy2003a}.   Fabricating an axially asymmetric structure with a stationary inflection point is challenging; it is necessary, although not sufficient, to include anisotropic layers \cite{JBallatoABallato2005,JBallatoABallatoFigotinVitebskiy2005, FigotinVitebskiy2001, FigotinVitebskiy2003, FigotinVitebskiy2003a}.


\begin{center}
{\scalebox{0.55}{\includegraphics{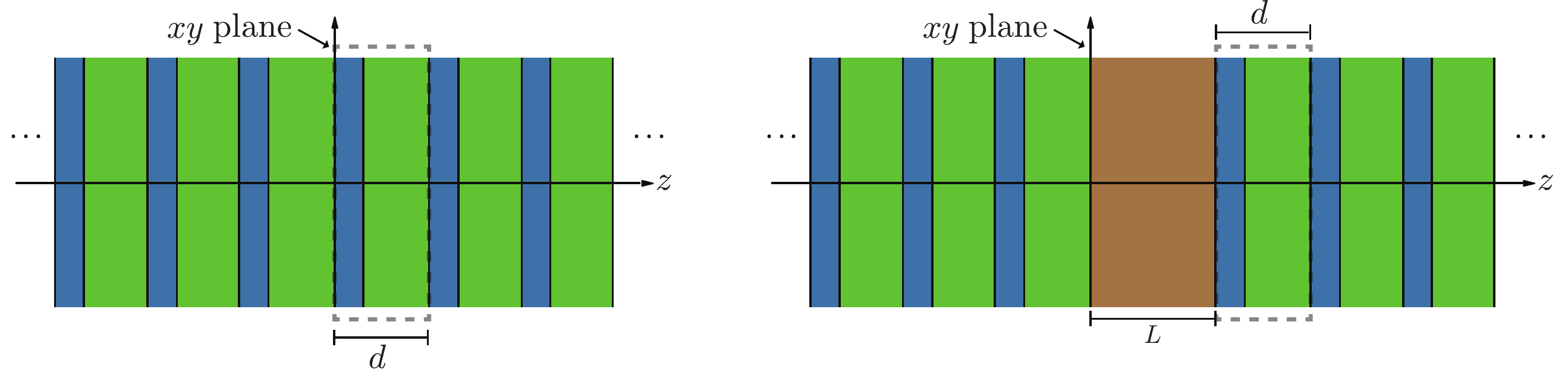}}}
\vspace{-1.8ex}
\captionof{figure}{\small
{\bfseries Left.}  A periodic layered medium.  Appropriate choices of anisotropic and magnetic layers render it a slow-light medium at specific frequency and wavevector parallel to the layers.
{\bfseries Right.}  A defective layer, or slab, embedded in an ambient periodic layered medium.
}
\label{fig:layered}
\end{center}

\subsection{Overview of applications of slow light}

Let us take a brief look at the scientific interest in slow light.
The activity in recent years is witnessed by a sizable and growing literature; browse, for example, \cite{Melloni12, FigotinVitebskiy2011, Cassan2011, WhiteKraussMelloni2010, BoydGauthier2011, Boyd2009, Baba2008, Krauss2008, KhurginTucker2008, Krauss2007, BabaMori2007, Ku2007, Boyd2006, FigotinVitebskiy2006} for an overview of the ideas.  Interest in slow light is due essentially to its application in producing two operations in optical devices.  First, it allows a time delay in an optical line to be tuned without physically changing the device length \cite{PovinelliJohnson2005}.  Second, it is used for the spatial compression of optical signals and energy together with increased dwell time, which can reduce device size and enhance light-matter interactions \cite{FigotinVitebskiy2005, FigotinVitebskiy2006, FigotinVitebskiy2007, SettleKrauss2007, Baba2008, BurrVitebskiy2013}. For the former, research into slow light is anticipated to be especially useful in the fields of optical telecommunication including all-optical data storage and processing, whereas the latter can dramatically enhance optical absorption, gain, phase shift \cite{SoljacicJohnson2002}, and nonlinearity \cite{SoljacicJoannopoulos2004, CorcoranKrauss2009, MonatKrauss2010, MonatEggleton2010, Boyd2011,ShramkovaSchuchinsk2014}, which could improve and miniaturize numerous optical devices such as amplifiers and lasers~(see \cite{FigotinVitebskiy2006, Baba2008}).  For instance, slow light has been investigated for absorption enhancement for potential applications in increasing solar cell efficiency \cite{Duche2008, John2008, ParkSeassal2009, Chen2011, DemesyJohn2012}, the enhancement of gain or spontaneous emission of lasers \cite{DowlingScalora1994, NohCao2008, NohCao2011, VitebskiyKottos2014}, the miniaturization of antennas or radio-frequency devices \cite{FigotinVolakis2005, MumcuVolakis2006, YargaVolakis2008, YargaVolakis2009, IrciVolakis2010, SertelVolakis2011}, \cite[Chap.~5]{Gross2011}, and the superamplification in high-power microwave amplifiers~\cite{Schamiloglu2012}.

Of particular interest recently has been the exploitation of enhanced light-matter interaction associated with slow light propagation for applications in gyroscopic or rotating systems and gyrotropic/magnetic systems, especially in magnetophotonic crystals that have spin-dependent photonic bandgap structure and localized modes of light with magnetic tunability \cite{Inoue2006, Inoue2013}. For instance, slow-light enhanced light-matter interactions have been proposed for increasing the sensitivity of optical gyroscopes \cite{Leonhardt2000, Shahriar2007, Peng2007, Fan2009, Schwartz2014}, enhancing rotary photon drag and image rotation based on a mechanical analog of the magnetic Faraday effect \cite{Faraday1846, Nienhuis1992, Padgett2006, ArnoldBoydPadgett2011, BarkerBoydPadgett2013, BarkerBoydPadgett2014a}, and enhancing magneto-optical (MO) effects, such as Faraday or Kerr rotation \cite{Zvezdin1997, Zvezdin2004, Zvezdin2009, Inoue2013b, Vinogradov2013}, which are important in applications using optical isolators, circulators, or other nonreciprocal devices \cite{Zvezdin1997, Pozar2011, Fan2011, FigotinVitebskiy2013, Alu2014a, Alu2014b}.  For instance, the enhancement of MO effects in multilayered structures such as one-dimensional magnetic photonic crystals (see~\cite{Vinogradov2013}) has been attributed to: (i) the localization of light near a defect and to those defect states (guided modes) with a high $Q$-factor (quality factor) associated with resonant transmission anomalies \cite{Inoue1998, Inoue1999, Inoue2000, SteelLevy2000, Levy2000, Inoue2006}; (ii) the enhanced light-matter interaction of slow light due to the low group velocity increasing interaction time \cite{Zvezdin2004, Zvezdin2009}; (iii) the Borrmann effect in photonic crystals, specifically relating to the frequency-dependent field redistribution and enhancement inside a photonic crystal unit cell \cite{Borrmann1941, VinogradovInoue2007, InoueRazdolskii2008, Inoue2008, Vinogradov2009}.

This study was carried out with these applications in mind---it incorporates scattering by a defect into a slow-light photonic crystal medium.  Such a system exhibits amplitude enhancement of fields scattered by the defect.  It also exhibits a new phenomenon, namely {\em the exitation of a guided frozen mode} which can be conceived as frozen-light analogy of a guided mode.

\subsection{Overview of analysis of scattering in the slow-light regime}

Although the analysis becomes quite technically involved, there is an underlying esthetic structure that we try to make clear in sec.~\ref{sec:pathological}:  The asymptotic nature of the scattering problem near the frozen-mode parameters $\kwz$ is reflected in the linear-algebraic relationships between two matrices defined at $\kwz$: a flux-hermitian matrix $\mathring{K}$ with a $3\times3$ Jordan block, and a flux-unitary matrix $\To$.  The next few paragraphs describe the framework in which the analysis takes place.

In a layered medium, the dielectric and magnetic tensors, $\epsilon$ and $\mu$, depend only on one spatial variable~$z$ and are independent of the other spatial variables $\xx=(x,y)$.  When considering EM fields of the form $(\EE(z),\HH(z))e^{i(\kk\cdot\xx-\omega t)}$, where the frequency $\omega$ and the wavevector $\kk$ parallel to the layers are fixed, the Maxwell equations reduce to a $4\!\times\!4$ system of ordinary differential equations (ODEs) for the components of the $E$ and $H$ fields directed parallel to the layers (Fig.~\ref{fig:layered}).  The matrix that transfers an EM field across one period of the layered structure is called the unit-cell transfer matrix or monodromy matrix, and it can be written as $e^{iKd}$, where $K$ is a matrix that is self-adjoint with respect to an indefinite energy-flux form coming from the electromagnetic Poynting vector.

Consider for a moment typical EM scattering in a periodically layered medium. 
At generic parameters $\kw$, the medium admits a two-dimensional space of rightward-directed modes and a two-dimensional space of leftward-directed modes, and these spaces together span the four-dimensional space of all fields $(\EE_\|,\HH_\|)$, identified with $\CC^4$.  Particularly interesting is the case in which the medium admits a pair of propagating modes, one rightward and one leftward, and a pair of evanescent modes, rightward and leftward; this is possible in anisotropic (but not isotropic) media.  By choosing a defect layer appropriately, one can construct guided modes that are exponentially confined to the layer and spectrally embedded in a propagation band of the ambient periodic medium.  The instability of these modes under perturbations of the system is associated with sharp resonance phenomena.  Detailed analysis of this resonance is carried out in our previous work~\cite{ShipmanWelters2013}.

But in a slow-light medium, the rightward/leftward mode structure breaks down at some specific pair $\kwz$.  The matrix $K=\mathring{K}$ admits a Jordan normal form with a one-dimensional block corresponding to a rightward positive-flux eigenmode; and a three-dimensional block corresponding to a zero-flux eigenmode, a mode of linear growth in $z$ and negative flux, and a quadratically growing mode.  The zero-flux mode is the {\em frozen mode}.

Under a generic perturbation of $\kw$ from $\kwz$, $K$ becomes diagonalizable, and the medium possesses the familiar eigenmode structure, where a 2D rightward space $V_+$ and a 2D leftward space $V_-$ each contains a propagating and an evanescent mode.  Take any one-parameter analytic perturbation described by an analytic function $(\kk(\eta),\omega(\eta))$ of a complex variable $\eta$ in a neighborhood of $\eta=0$, with $(\kk(0),\omega(0))=\kwz$.  In the limit $\eta\to0$, the rightward propagating mode in $V_+$ persists, but the other three modes---the rightward evanescent mode in $V_+$ and both leftward modes in $V_-$---all coalesce into the zero-flux mode of the three-dimensional Jordan block.  These three modes admit expansions that are Puiseux series in $\eta^\third$.

The key to pathological scattering by a defect layer lies in the limiting behavior of the rightward and leftward spaces $V_+$ and $V_-$ as $\eta\to0$.  While for $\eta\not=0$, $V_+$ and $V_-$ span $\CC^4$, their limiting spaces $\Vop$ and $\Vom$ do not.  The span $\Vo:=\Vop+\Vom$ is a three-dimensional subspace of $\CC^4$ characterized by the exclusion of quadratically growing fields.

Denote by $T(\eta)$ the analytic transfer matrix across the defect layer, and set $\To=T(0)$.  We show in this work that one of two distinctly different asymptotic scattering behaviors occurs as $\eta\to0$, depending on whether or not $\To$ transfers the limiting space $\Vo$ to itself or not.  In the generic case, the 3D spaces $T(\Vo)$ and $\Vo$ have a 2D intersection in $\CC^4$, and it turns out that necessarily the intersection of $\Vo$ with each the spaces $\To(\Vom)$ and $\To^{-1}(\Vop)$ is 1D and that this 1D intersection identifies in a precise way the limit as $\eta\to0$ of the field scattered by an energy-carrying wave incident upon the slab from the right or left.  The part of $\To(\Vom)$ or $\To^{-1}(\Vop)$ not in $\Vo$ produces excitation of the quadratic mode.

The special case that $\To(\Vo)=\Vo$ is the {\em resonant} scenario.  It turns out that this occurs exactly when $\To$ maps the zero-flux mode to itself.  This means that the structure supports a global field that has no energy flux but yet is not exponentially decaying.  This field is the slow-light analogue of a guided mode, which decays exponentially away from the defect in the typical scattering regime analyzed in~\cite{ShipmanWelters2013}; we call it the {\em guided frozen mode}.

The delicate nature of the asymptotics of the scattering problem as $\eta\to0$ is manifest in the Laurent-Puiseux expansions of the projections onto $V_+$ and $V_-$.  These projections blow up as $\eta^\mthirds$ (Propostion~\ref{prop:PuiseuxSeriesProjAsympt}).  The analysis reduces to a singular analytic perturbation problem of linear algebra in $\CC^4$; the players are

\smallskip
\noindent
\hspace*{1.5em}(1) an indefinite flux form $[\cdot\,,\,\cdot]$ with two-dimensional positive and negative spaces;\\
\hspace*{1.5em}(2) a flux-self-adjoint matrix $K(\eta)$ that is analytic in $\eta$;\\
\hspace*{1.5em}(3) a $3\times3$ Jordan block of $K(0)$;\\
\hspace*{1.5em}(4) an analytic flux-unitary matrix $T(\eta)$.
\smallskip

\noindent
At the center of the analysis is the collapse of the rightward and leftward spaces formed by the eigenvectors of $K(\eta)$ as $\eta\to0$ and the images of these spaces under $T(\eta)$.

\smallskip
Here is a brief overview of the exposition:

\medskip
{\bfseries Section~\ref{sec:slowlight}:  Slow light in layered media and scattering by a defect.}
The reduction of the equations of electrodynamics in a lossless periodic layered medium to an ODE depending on frequency and layer-parallel wavevector is reviewed.
Then the connection between the Jordan form of the Maxwell ODEs and a {frozen mode} within a spectral band is developed (Theorem~\ref{thm:localbandstrucandjordanstructure}).  We derive canonical Laurent-Puiseux series for the EM modes and the projections onto these modes when the frequency and wavevector are perturbed from those that admit a frozen mode (Propositions \ref{prop:PuiseuxSeriesEigenvAsympt} and \ref{prop:PuiseuxSeriesProjAsympt}).
Scattering of EM waves by a defect layer is reviewed, and we introduce the idea of a {guided frozen mode} of the defect.

\smallskip
{\bfseries Section~\ref{sec:pathological}: Pathological scattering in the slow-light limit.}
At frequency and wavevector that admit a frozen mode for a periodic ambient medium, we develop the
algebraic structure of the scattering problem, involving the Jordan form of the Maxwell ODEs, the transfer matrix across the defect slab, the indefinite flux form, and an algebraic characterization of a guided frozen mode of the defect (Theorem~\ref{thm:D}).
Sec.~\ref{sec:main} contains the main results, Table~\ref{table:scattering} and Theorems
\ref{thm:leftscatteringreg}--\ref{thm:rightscatteringres}.  They present a detailed analysis of delicate pathological scattering, both nonresonant and resonant.

\smallskip
{\bfseries Section~\ref{sec:analyticperturbationtheory}: Analytic perturbation for scattering in the slow-light regime.}
This section proves results used in the analysis that establish connections between the Jordan canonical structure of $K(0)$ and the band structure of a slow-light medium; and canonical Laurent-Puiseux series for the eigenvalues and eigenvectors of $K(\eta)$ and the associated projections.

\section{Slow light in layered media and scattering by a defect}\label{sec:slowlight} 

This section reviews electromagnetic propagation theory in layered media, and in the slow-light regime it further develops the mode structure and problem of scattering by a defect layer.  It also introduces the concept of a guided frozen mode of a defect layer.

\subsection{Electrodynamics of lossless layered media}\label{sec:CanonicalODEs}

The Maxwell equations for time-harmonic electromagnetic fields $\big(\EE(\rr),\HH(\rr),\DD(\rr),\BB(\rr)\big)e^{-i\omega t}$ ($\omega\not=0$) in linear anisotropic media without sources are
\begin{equation}\label{MaxwellEqsTimeHar}
\mat{1.1}{0}{\nabla\times}{-\nabla\times}{0}
\col{1.1}{\EE}{\HH}
=
-\frac{i\omega}{c}
\col{1.1}{\DD}{\BB},\qquad
\col{1.1}{\DD}{\BB}
=
\mat{1.1}{\epsilon}{0}{0}{\mu}
\col{1.1}{\EE}{\HH}
\end{equation}
(in Gaussian units), where $c$ denotes the speed of light in a vacuum. We consider non-dispersive and lossless media, which means that the dielectric permittivity $\epsilon$ and magnetic permeability $\mu$ are $3\times3$ Hermitian matrices that depend only on the spatial variable $\rr=(x,y,z)$.
A layered medium is one for which $\epsilon$ and $\mu$ depend only on~$z$.  Thus
\begin{align}\label{LosslessLayeredMedia}
\epsilon=\epsilon(z)=\epsilon(z)^*,\;\;\mu=\mu(z)=\mu(z)^*,
\end{align}
where $*$ denotes the Hermitian conjugate (adjoint) of a matrix. 
Typically, a layered medium consists of layers of different homogeneous materials.
We assume that the medium is passive.  This means that the frequency-independent and $z$-dependent material tensors $\epsilon, \mu$ are Hermitian matrix-valued functions which are bounded (measurable) and coercive, that is, for some constants $c_1, c_2>0$ 
\begin{align}\label{PassiveMedia}
c_1I\leq \epsilon(z) \leq c_2 I
\;\text{ and }\;
c_1I\leq \mu(z) \leq c_2 I
\end{align}
for all $z\in\RR$, where $I$ denotes the $3\times 3$ identity matrix.


\subsubsection{The canonical Maxwell ODEs}\label{sec:TheCanonicalMaxwellODEs}
Because of the translation invariance of layered media along the $xy$ plane, solutions of equation (\ref{MaxwellEqsTimeHar}) are sought in the form
\begin{equation}\label{TangentialPlaneWavesRepr}
\col{1.1}{\EE}{\HH}=
\col{1.1}{\EE(z)}{\HH(z)}
e^{i(k_1x+k_2y)}\,,
\end{equation}
in which
%
  $\kk=(k_1,k_2)$
%
is the wavevector parallel to the layers.
The harmonic Maxwell equations (\ref{MaxwellEqsTimeHar}) for this type of solution can be reduced to a system of ordinary differential equations for the tangential electric and magnetic field components (see \cite{Berreman1972} and \cite[Appendix]{ShipmanWelters2013}),
\begin{equation}\label{canonical}
  -iJ\, \frac{d}{dz} \psi(z) \,=\, A(z;\kk,\omega)\, \psi(z),
\end{equation}
in which
\[
\psi(z)=\left[E_1(z), E_2(z), H_1(z),H_2(z)\right]^T\,,
\]
\begin{equation}\label{J}
  J=
  \renewcommand{\arraystretch}{1.2}
\left[
  \begin{array}{cccc}
    0 & 0 & 0 & 1 \\
    0 & 0 & -1 & 0 \\
    0 & -1 & 0 & 0 \\
    1 & 0 & 0 & 0
  \end{array}
\right], \quad
J^*=J^{-1}=J\,.
\end{equation}
The $4\times 4$ matrix $A(z;\kk,\omega)$ is given in \cite[\S A5, A7]{ShipmanWelters2013} and its analytic properties are described in \cite[\S A8, A12, A13]{ShipmanWelters2013}, in particular, it is a Hermitian matrix for real $\kw$, $\omega\not=0$. We will refer to the ODEs in (\ref{canonical}) as the canonical \emph{Maxwell ODEs}.

Interface conditions for electromagnetic fields in layered media require that tangential electric and magnetic field components $\psi(z)$ be continuous across the layers \cite{Jackson1999}, which means $\psi$ is an absolutely continuous function of~$z$ satisfying the Maxwell ODEs (\ref{canonical}).
As was shown in \cite[Appendix]{ShipmanWelters2013}, every solution $\psi(z)$ of the Maxwell ODEs (\ref{canonical}) is the tangential EM field components of a unique electromagnetic field in the form (\ref{TangentialPlaneWavesRepr}) and vice versa.

\subsubsection{The transfer matrix} \label{sec:transfermatrix}

The initial-value problem
\begin{equation}
  -iJ\, \frac{d}{dz} \psi(z) \,=\, A(z;\kk,\omega)\, \psi(z),
  \qquad
  \psi(z_0)=\psi_0
\end{equation} 
for the Maxwell ODEs (\ref{canonical}) has a unique solution \cite[Appendix]{ShipmanWelters2013}
\begin{equation}
  \psi(z) \,=\, T(z_0,z)\, \psi(z_0)
\end{equation}
for each initial condition $\psi_0\in \CC^4$.   The $4\times 4$ matrix $T(z_0,z)$ is called the \emph{transfer matrix}.
It satisfies
\begin{align}
T(z_0,z)=T(z_1,z)T(z_0,z_1),\quad T(z_0,z_1)^{-1}=T(z_1,z_0),\quad T(z_0,z_0)=I,
\end{align}
for all $z_0,z_1,z\in \RR$.   As a function of $z$, it is absolutely continuous and, as a function of the wavevector-frequency pair $(\kk,\omega)\in \CC^2\times \CC\setminus\{0\}$, it is analytic~\cite{ShipmanWelters2013} (in the uniform norm as a bounded function on compact subsets of the $z$-axis).
Perturbation analysis of analytic matrix-valued functions and their spectrum is central to the study of scattering problems, particularly 
those involving guided modes, as in our previous work \cite{ShipmanWelters2013}, or slow light, as in this work or as discussed in \cite{FigotinVitebskiy2006,Welters2011,Welters2011a,ShipmanWelters2012}, for instance.

\subsubsection{Electromagnetic energy flux in layered media}

There is an indefinite inner product $[\cdot,\cdot]:\mathbb{C}^4\times \mathbb{C}^4\rightarrow\mathbb{C}$ associated
with the energy conservation law for the Maxwell ODEs (\ref{canonical}) coming from Poynting's theorem, for lossless layered media.  An important consequence, which we will discuss below, is that the transfer matrix is unitary with respect to this indefinite inner product for real frequencies and wavevectors. Moreover, as was shown in \cite{FigotinVitebskiy2006, ShipmanWelters2013} and as we shall see later in this paper, $[\cdot,\cdot]$ and the transfer matrix $T$ are fundamental to analyzing scattering problems in lossless layered media, especially problems involving guided modes and slow light.

\smallskip
{\bfseries The indefinite inner product.}\;
The indefinite sesquilinear energy-flux form associated with the canonical ODEs (\ref{canonical}) is
\begin{equation}\label{def:flux}
  [\psi_1,\psi_2] := \frac{c}{16\pi}(J\psi_1,\psi_2),\;\;\;\psi_1,\psi_2\in\mathbb{C}^4,
\end{equation}
where $(\cdot,\cdot)$ is the usual complex inner product in $\mathbb{C}^4$ with the convention of linearity in the second component and conjugate-linearity in the first, that is $(v,cw)=c(v,w)=(\bar cv,w)$.  The indefinite inner product $[\cdot,\cdot]$ will play a central role in the analysis of scattering, and its physical significance is that the corresponding quadratic form $\psi\mapsto[\psi,\psi]$ gives the time-averaged electromagnetic energy flux (the Poynting vector) in the normal direction of the layers \cite[\S III.B]{ShipmanWelters2013} as described below. The adjoint of a matrix $M$ with respect to $[\cdot,\cdot]$ is denoted by $M^\fadj$ and is called the {\em flux-adjoint} of $M$.  It is equal to $M^\fadj=J^{-1}M^*J$, where $M^*=\overline{M}^T$ is the adjoint of $M$ with respect to the standard inner product $(\cdot, \cdot)$. 

If $\omega$ is real and nonzero and $\kk$ is real, then the matrix $A$ is self-adjoint with respect to $(\cdot,\cdot)$, {\itshape i.e.}, $A^*=A$, and $JA$ is self-adjoint with respect to $[\cdot,\cdot]$, and $T$ is unitary with respect to $[\cdot,\cdot]$, {\itshape i.e.}, for any $\psi_1, \psi_2\in \mathbb{C}^4$,
\begin{equation*}
\renewcommand{\arraystretch}{1.1}
\left.
  \begin{array}{ll}
  {} [T\psi_1,\psi_2] = [\psi_1,T^{-1}\psi_2]\,, &\hspace*{1em} T^\fadj = T^{-1}\,.
  \end{array}
\right.
\end{equation*}
The flux-unitarity of $T$ follows from the energy conservation law \cite[Theorem 3.1]{ShipmanWelters2013} and expresses the conservation of energy in a $z$-interval $[z_0,z_1]$ through the principle of spatial energy-flux invariance for lossless media,
\begin{equation}\label{flux1}
  [\psi(z_1),\psi(z_1)]=[T(z_0,z_1)\psi(z_0),T(z_0,z_1)\psi(z_0)]=[\psi(z_0),\psi(z_0)]
  \quad \big((\kk,\omega) \,\text{ real}\big).
\end{equation}

\smallskip
{\bfseries The EM energy flux.}\;
For any time-harmonic EM field with spatial factor of the form (\ref{TangentialPlaneWavesRepr}) with real nonzero $\omega$ and real $\kk$, the time-averaged energy flux (as defined in \cite[\S6.9]{Jackson1999}) is given by the real part of the complex Poynting vector $\mathbf{S}=\frac{c}{8\pi}\EE\times\overline{\HH}$, namely,
\begin{align}
\operatorname{Re}\,\mathbf{S} = \frac{c}{8\pi}\operatorname{Re}\,\left(\EE\times\overline{\HH}\right),
\end{align}
which only depends on the spatial variable $z$, {\itshape i.e.},
\begin{align}
\operatorname{Re}\,\mathbf{S}(\rr)=\operatorname{Re}\,\mathbf{S}(z)=\frac{c}{8\pi}\operatorname{Re}\,\left(\EE(z)\times\overline{\HH(z)}\right).
\end{align}
The time-averaged energy flux in the normal direction ${\bf e}_3$ (pointing in the positive $z$-direction) of the layers is just the normal component of $\operatorname{Re}\,\mathbf{S}(\rr)$ which reduces to \cite[\S III.B]{ShipmanWelters2013}
\begin{align}
\operatorname{Re}\,\mathbf{S}(z)\cdot {\bf e}_3=[\psi(z),\psi(z)],
\end{align}
where $\psi(z)$ is the tangential EM field components of the EM field as described in sec. \ref{sec:TheCanonicalMaxwellODEs}, and this is a solution of the Maxwell ODE (\ref{canonical}). As such, the time-averaged energy flux in the normal direction ${\bf e}_3$ of the layers is constant, {\itshape i.e.},
\begin{align}
[\psi(z),\psi(z)]=[\psi(z_0),\psi(z_0)]
\end{align}
for all $z_0,z\in\mathbb{R}$. This is a statement of the energy conservation law for lossless layered media (cf.~\cite[Theorem 3.1]{ShipmanWelters2013}).

\subsection{A periodic ambient medium in the slow-light regime}\label{sec:periodic} 

Let the material coefficients $\epsilon$ and $\mu$ of the ambient space be periodic with period $d$, {\itshape i.e.},
\begin{align}
\epsilon(z+d)=\epsilon(z),\;\;\;\mu(z+d)=\mu(z).\label{periodicmaterials}
\end{align}
Then for the Maxwell ODEs (\ref{canonical}) the propagator $iJ^{-1}A(z;\kk,\omega)$ for the field $\psi(z)$ along the $z$-axis is $d$-periodic.  According to the Floquet theory (see, {\it e.g.}, \cite[Ch.\!~II]{YakubovichStarzhinsk1975}), the general solution of the Maxwell ODEs is pseudo-periodic, meaning that the transfer matrix $T(0,z)$ is the product of a periodic matrix and an exponential matrix,
\begin{eqnarray*}
 && T(0,z)=F(z)e^{iKz},
   \quad
  F(z+d) = F(z)\,,
  \quad
  F(0) = I\,,
 \\
  && \psi(z) \,=\, T(0,z) \psi(0) = F(z) e^{iKz} \psi(0)\,.
\end{eqnarray*}
For a real pair $(\kk,\omega)$ with $\omega\not=0$, $F(z)$ can be chosen to be flux-unitary and the constant-in-$z$ matrix $K$ to be flux-self-adjoint:
\begin{equation*}
  [F(z)\psi_1,\psi_2] = [\psi_1,F(z)^{-1}\psi_2]\,,
  \quad
  [K\psi_1,\psi_2] = [\psi_1,K\psi_2]\,.
\end{equation*}
In concise notation, $F(z)^\fadj=F(z)^{-1}$ and $K^\fadj=K$.  The matrix $T(0,d)=e^{iKd}$ is called the unit-cell transfer matrix or monodromy matrix for the sublattice $d\ZZ\subset\RR$. The flux-self-adjointness of $K$ implies that its eigenvalues come in conjugate pairs.  For any pair $\kwz\in\RR^2\times \RR/\{0\}$ the matrices $F$ and $K$ can be chosen such that they are analytic in a complex open neighborhood of $\kwz$.  This follows from the analytic properties of $A$ (see \cite[(A14), (A8)]{ShipmanWelters2013} and \cite[\S III.4.6]{YakubovichStarzhinsk1975}). We shall assume that this is the case in a complex open neighborhood $\mathcal{N}$ of a pair $\kwz$ and restrict ourselves to wavevector-frequency pairs $\kw$ in this neighborhood.

\subsubsection{Electromagnetic Bloch waves and the dispersion relation}\label{sec:EMBlochwaves}

Each solution to the eigenvalue problem
\begin{equation}
  K\psi_0=k\psi_0,\;\;\;k\in\mathbb{C},\;\;\;\psi_0\in\mathbb{C}^4,\;\;\;\psi_0\not=0
\end{equation}
corresponds to an electromagnetic Bloch wave (eigenmode) 
\begin{equation}
  \psi(z)=T(0,z)\psi_0,
\end{equation}
which is pseudo-periodic,
\begin{equation}
\psi(z+d)=e^{ikd}\psi(z),\;\;\;z\in \mathbb{R}.  
\end{equation}
The vector $\psi(0)=\psi_0$ is an eigenvector of the monodromy matrix $T(0,d)$ ({\itshape i.e.}, the transfer matrix over the periodic unit cell) with corresponding eigenvalue $e^{ikd}$ ({\itshape i.e.}, a Floquet multiplier) since
\begin{equation}
T(0,d)\psi(0)=e^{ikd}\psi(0).
\end{equation}

The dispersion relation $\omega=\omega(\kk,k)$ is a function (possibly multivalued) which is implicitly defined by the zero set of an analytic function $R$, specifically,
\begin{gather}
R(k,\kk,\omega)=\det\left(kI-K\kw\right),\;\;\;(k,\kk,\omega)\in\mathbb{C}\times\mathcal{N},\label{def:CharacteristicPolynomialIndicatorMatrix}\\
\omega=\omega(\kk,k) \Leftrightarrow R(k,\kk,\omega)=0.
\end{gather}
Moreover, for a given $\kw\in\mathcal{N}$, the set of (axial) wavenumbers modulo $2\pi/d$ are the set of $k\in\mathbb{C}$ such that $R(k,\kk,\omega)=0$.

The following theorem is important in the study of slow light since it gives the fundamental connection between the local band structure of the dispersion relation, particularly its stationary points, and the Jordan normal form of the indicator matrix $K$.  Although stated in the context of electrodynamics in periodic layered media for the $4\times 4$ indicator matrix $K$ (where $e^{iKd}=T(0,d)$),
this theorem is very general, essentially applying to any periodic differential-algebraic equations (DAEs) or ODEs that are of definite type (see \cite{Welters2011a}) and hence for indicator matrices of any finite dimension.  Parts (i)--(iii) were first proved for the unit-cell transfer matrix $T(\omega)=T(0,d;\omega, \kappa_0)$ near a stationary inflection point of the dispersion relation by Figotin and Vitebskiy \cite[Appendix B]{FigotinVitebskiy2003} under the three assumptions:
a) $T(\omega)$ is analytic near $\omega_0$;
b) $T(\omega_0)$ has a degenerate eigenvalue $e^{ik_0d}$ (Floquet multiplier) with algebraic multiplicity $3$ satisfying $\left|e^{ik_0d}\right|=1$;
c) the generic condition on the dispersion relation $\frac{\partial}{\partial \omega}\det(e^{ik_0d}-T(\omega))|_{\omega=\omega_0}\not = 0$ is satisfied.
Their analysis also proved that the geometric multiplicity of $T(\omega_0)$ corresponding to the eigenvalue $e^{ik_0d}$ must be $1$, that is, the Jordan normal form of $T(\omega_0)$ corresponding to this eigenvalue is a single $3\times 3$ Jordan block. A further generalization of this result for the monodromy matrix was proved by Welters \cite{Welters2011, Welters2011a} in the context of periodic DAEs and ODEs. The theorem below is an extension of these results and is crucial to the study of slow-light scattering problems, as will be seen later on.

\begin{theorem}[local band structure and Jordan structure]\label{thm:localbandstrucandjordanstructure}
Let $k^0$ be a real eigenvalue of the indicator matrix $K\kwz$, and let $g$ be the number of Jordan blocks (geometric multiplicity) corresponding to $k^0$ in the Jordan normal form of $K\kwz$.
Let $m_1\geq\cdots\geq m_g\geq 1$ be the dimensions of each of those Jordan blocks (partial multiplicities), and set $m=m_1+\cdots+m_g$ (algebraic multiplicity of $k^0$).  Then
\begin{enumerate}
\item[i.] The order of the zero of $R(k^0,\kk^0,\omega)$ at $\omega=\omega^0$ is $g$ and the order of the zero of $R(k,\kk^0,\omega^0)$ at $k=k^0$ is $m$, where $R(k,\kk^0,\omega)$ is the characteristic polynomial of the matrix $K(\kk^0,\omega)$ (see (\ref{def:CharacteristicPolynomialIndicatorMatrix}).
\item[ii.] All solutions $(k,\omega)$ to $R(k,\kk^0,\omega)=0$ in a complex neighborhood of $(k^0,\omega^0)$ are given by exactly $g$ (counting multiplicities) single-valued nonconstant real-analytic functions $\omega=\omega_1(\kk^0,k), \ldots, \omega=\omega_g(\kk^0,k)$ (the band functions) which satisfy $\omega_j(\kk^0,k^0)=\omega^0$.
\item[iii.] The band functions can be reordered such that the number $m_j$ is the order of the zero of the analytic function $\omega_j(\kk^0,k)-\omega^0$ at $k=k^0$, for $j=1,\ldots, g$.
\item[iv.] Let the ordered set $\{\sigma_j$, $j=1,\ldots,g\}$ denote the sign characteristic of the self-adjoint matrix $K\kwz$ with respect to the indefinite inner product $[\cdot,\cdot]$ corresponding to the eigenvalue $k^0$ and let $\omega_j^{(m_j)}(\kk^0,k^0)$ denote the $m_j$-th derivative of the $j$-th band function with respect to $k$ at $k=k^0$. Then the sign characteristic can be reordered such that $\sigma_j=\frac{\omega_j^{(m_j)}(\kk^0,k^0)}{\left\vert\omega_j^{(m_j)}(\kk^0,k^0)\right\vert}$ for $j=1,\ldots, g$. In particular, the sign characteristic is uniquely determined by the limits of the signs of the group velocities, {\itshape i.e.},
\begin{equation}
\sigma_j=\lim\limits_{k\downarrow k^0}\frac{\omega_j^{(1)}(\kk^0,k)}{\left\vert\omega_j^{(1)}(\kk^0,k)\right\vert},
\end{equation}
for $j=1,\ldots,g$.
\end{enumerate}
\end{theorem}

\subsubsection{Stationary points of the dispersion relation and the slow-light regime}

The slow-light regime can be understood as a small complex open neighborhood of a real pair $\kwz\in \mathbb{R}^2\times \mathbb{R}\!\setminus\!\{0\}$ such that the axial dispersion relation $\omega=\omega(\kk^0,k)$ has a stationary point at $(k^0,\omega^0)$; in other words, there exists a band function $\omega=\omega_j(\kk^0,k)$ with $\omega^0=\omega_j(\kk^0,k^0)$, where $k^0$ is an real eigenvalue of $K\kwz$, such that the axial group velocity is zero at $k=k^0$, {\itshape i.e.}, $\frac{\partial \omega_j}{\partial k}(\kk^0,k^0)=0$. Any electromagnetic Bloch wave $\psi(z)$ with $\psi(0)$ an eigenvector of $K(\kk^0,\omega_j(\kk^0,k))$ corresponding to the real eigenvalue $k$ with $0<|k-k_0|\ll 1$ is called a slow wave or slow light field with frequency $\omega_j(\kk^0,k)$, tangential wavevector $\kk^0$, and axial wavenumber $k$.

From Theorem \ref{thm:localbandstrucandjordanstructure}, we see that the slow-light regime corresponds to a real pair $\kwz\in \mathbb{R}^2\times \mathbb{R}\!\setminus\!\{0\}$ at which the indicator matrix $K$ is non-diagonalizable, having a real eigenvalue with a corresponding Jordan block of dimension at least $2$.
Our study concerns a $3\times3$ Jordan block, although the type of anisotropic media assumed in this paper allows for any Jordan normal form and sign characteristics that are possible for a $4\times 4$ matrix that is flux-self-adjoint (see \cite[Corollary 5.2.1; Table, p.~90]{GohbergLancasterRodman2005}).

\smallskip
{\bfseries Assumptions on the unperturbed Jordan structure.}\;
At a real pair $\kwz$, assume that $K$ has a real eigenvalue $k_0$ whose corresponding Jordan normal form contains a $3\times 3$ Jordan block.  Set $\mathring{K}:=K\kwz$, and assume that the eigenvalues $\kplus$ and $\kzero$ of $\mathring{K}$ are distinct and that the energy-flux form $[\cdot,\cdot]$ is positive-definite on the eigenspace of $\mathring{K}$ corresponding to $\kplus$. 

We introduce a one-parameter perturbation $(\kk(\eta),\omega(\eta))$ of $\kwz$ that is analytic in a complex neighborhood of $\eta=0$, is real for real $\eta$, and satisfies $(\kk(0),\omega(0))=\kwz$.  For example, this may be a small perturbation of the frequency alone, {\itshape i.e.}, $\omega(\eta)=\omega^0+\eta$ and $\kk(\eta)=\kk^0$. We use the notational convention that if $f\kw$ is any function of $\kw$ defined at $(\kk(\eta),\omega(\eta))$, then $f(\eta):=f(\kk(\eta),\omega(\eta))$.

It follows from these assumptions that $K(\eta)$ is an analytic matrix-valued function of $\eta$ in a complex neighborhood of $\eta=0$ and is flux-self-adjoint for real $\eta$ near $\eta=0$ with $K(0)=\mathring{K}$. This allows one to apply analytic perturbation theory for matrices to the study of how the spectrum and Jordan structure of $K(\eta)$ changes as a function of a single perturbation parameter $\eta$.

\subsubsection{Algebraic structure of the slow-light ambient medium}\label{sec:UnidirAmbMed}

There are three distinguished bases of $K(\eta)$ that are natural for the study of the scattering problem: a Jordan basis at $\eta=0$; a basis of eigenvectors at $\eta\not=0$; a modification of the latter that incorporates the normal mode of $\mathring K$ that is linear in $z$.

\smallskip
{\bfseries Jordan basis.}\;
At $\eta=0$, the theory of indefinite linear algebra \cite{GohbergLancasterRodman2005} provides a Jordan basis
$\{e_+,e_0,e_1,e_2\}\subset\CC^4$ with respect to which the propagator $\mathring{K}=K(0)$ and the energy-flux form $[\cdot,\cdot]$ have the form
\begin{equation}\label{jordanform}
  \widetilde {\mathring{K}} \,=\,
  \renewcommand{\arraystretch}{1.1}
\left[
  \begin{array}{cccc}
    \kplus & 0 & 0 & 0 \\
    0 & \kzero & 1 & 0 \\
    0 & 0 & \kzero & 1 \\
    0 & 0 & 0 & \kzero
  \end{array}
\right]\,,
\quad
  \big([e_i,e_j]\big)_{i,j\in\{+,0,1,2\}} \,=\,
  \renewcommand{\arraystretch}{1.1}
\left[
  \begin{array}{cccc}
    1 & 0 & 0 & 0 \\
    0 & 0 & 0 & -1 \\
    0 & 0 & -1 & 0 \\
    0 & -1 & 0 & 0
  \end{array}
\right]\,
\end{equation}
with $\kplus$ and $\kzero$ distinct real numbers.
Hence, it follows that the dual basis to $\{e_+,e_0,e_1,e_2\}$ is
\begin{equation}\label{dualbasis}
  \left\{ \varepsilon^+,\, \varepsilon^0,\, \varepsilon^1,\, \varepsilon^2  \right\}
  \,=\, \left\{ [e_+,\cdot],\,-[e_2,\cdot],\,-[e_1,\cdot],\,-[e_0,\cdot] \right\}\,.
\end{equation}
The columns of the matrix $e^{i\widetilde {\mathring{K}}z}$ express the generalized eigenmodes as combinations of the $e_i$\,:
\begin{equation}\label{jordanpropagation}
  e^{i\widetilde {\mathring{K}}z} =
  \renewcommand{\arraystretch}{1.0}
\left[
  \begin{array}{cccc}
    1 & 0 & 0 & 0 \\
    0 & 1 & iz & -z^2/2 \\
    0 & 0 & 1 & iz \\
    0 & 0 & 0 & 1
  \end{array}
\right]
\renewcommand{\arraystretch}{1.0}
\left[
  \begin{array}{cccc}
    e^{i\kplus z} & 0 & 0 & 0 \\
    0 & e^{i\kzero z} & 0 & 0 \\
    0 & 0 & e^{i\kzero z} & 0 \\
    0 & 0 & 0 & e^{i\kzero z}
  \end{array}
\right]\,.
\end{equation}
The first column of $e^{i\widetilde {\mathring{K}}z}$, or $e_+e^{ik_+z}$, is the {\em rightward propagating eigenmode}; it carries positive energy flux $[e_+,e_+]=1$, and the second column, $e_0e^{ik_0z}$, is the {\em zero-flux eigenmode} with $[e_0,e_0]=0$.  The third column, $(e_1+ize_0)e^{ik_0z}$, is called the {\em linear mode}, which, as seen by the interaction matrix in~(\ref{jordanform}), carries negative energy flux $-1$; and the fourth column $(e_2+ize_1-z^2e_0/2)e^{ik_0z}$ is the {\em quadratic mode}, which carries zero energy~flux.

The Maxwell field corresponding to the zero-flux eigenmode is called a {\em frozen mode} of the periodic medium~\cite{FigotinVitebskiy2003a}:
\begin{equation}
  F(z) e_0 e^{ik_0z} \qquad \text{(frozen mode of periodic medium).}
\end{equation}
A  medium that admits a frozen mode is commonly called ``unidirectional" because it has only one propagating eigenmode (corresponding to $e_+$), which carries energy across the layers in one direction~\cite{FigotinVitebskiy2013,FigotinVitebskiy2003}.  However, one should keep in mind that a generalized eigenmode, namely the linear mode $e_1+ize_0$, carries energy in the other direction.  Thus from the point of view of energy flux, the medium is not truly unidirectional.

\smallskip
{\bfseries Basis of rightward and leftward eigenvectors.}\;
Under a generic analytic perturbation about $\eta=0$, where by generic we mean
\begin{align}
[e_0,K^{\prime}(0)e_0]\not=0,\label{GenericCondition}
\end{align}
the matrix $K(\eta)$ admits an eigenvector basis that determines the rightward and leftward modes that define the problem of scattering by a slab, discussed below in sec.~\ref{sec:scatteringdefectlayer}.
Propositions~\ref{prop:PuiseuxSeriesEigenvAsympt} and~\ref{prop:PuiseuxSeriesProjAsympt} are proved in sec.~\ref{sec:analyticperturbationtheory}.

\begin{proposition}\label{prop:PuiseuxSeriesEigenvAsympt}
The matrix $K(\eta)$ has four distinct eigenvalues $\krp(\eta),\klp(\eta),\kre(\eta),\kle(\eta)$ for $0<|\eta|\ll~\!1$. The eigenvalue $\krp$ is an analytic function of $\eta$ which is real for real $\eta$, and the eigenvalues $\klp,\kre,\kle$ are the branches of a convergent Puiseux series in $\eta^\third$,
\begin{equation}\label{eigenvalues}
  \renewcommand{\arraystretch}{1.2}
\left.
  \begin{array}{l}
  \krp = \kplus + \O(\eta)\,, \\
  \klp = \kzero + \eta^\third \kone + \O(\eta^\thirds)\,, \\
   \kre = \kzero + \eta^\third \zeta\kone + \O(\eta^\thirds)\,, \\
   \kle = \kzero + \eta^\third \zeta^2\kone + \O(\eta^\thirds)\,,
  \end{array}
\right.
\end{equation}
as $\eta\rightarrow 0$, where $\zeta=e^{2\pi i/3}$ and $\kone$ is the number
\begin{align}\label{k1}
\kone=-\sqrt[3]{[e_0,K^{\prime}(0)e_0]}\not=0.
\quad\text{(for the real cube root)}
\end{align}
(The cube root can be taken to be real because the flux-self-adjointness of $K(\eta)$ makes $[e_0,K^{\prime}(0)e_0]$ real.)
The eigenvalue $\klp(\eta)$ is real for real $\eta^\third$ near $\eta=0$.  An eigenvector $\vrp$ of $\krp$ can be chosen to be analytic in $\eta$, and for the eigenvalues $\klp,\kre,\kle$, eigenvectors can be chosen to be the branches of a convergent Puiseux series in $\eta^\third$ such that
\begin{equation}\label{eigenvectors}
  \renewcommand{\arraystretch}{1.1}
\left.
  \begin{array}{l}
    \vrp \,=\, e_+ + \O(\eta)\,, \\ 
    \vre \,=\, e_0 + \eta^\third\zeta w_1 + \eta^\thirds\zeta^2w_2 + \O(\eta)\,, \\
    \vlp \,=\, e_0 + \eta^\third w_1 + \eta^\thirds w_2 + \O(\eta)\,, \\
    \vle \,=\, e_0 + \eta^\third\zeta^2 w_1 + \eta^\thirds \zeta w_2 + \O(\eta)\,,   
  \end{array}
\right.
\end{equation}
as $\eta\rightarrow 0$, for some $w_1,w_2\in \mathbb{C}^4$. These coefficients have the following properties
\begin{equation}\label{LowOrderExpanEigenvecSpanRel}
  \renewcommand{\arraystretch}{1.1}
\left.
  \begin{array}{lll}
     w_1\in \sspan\{e_0,e_1\}, & w_1\not\in \sspan\{e_0\}, & [e_1,w_1]=-\kone,\\
     w_2\in \sspan\{e_0,e_1,e_2\}, & w_2\not\in \sspan\{e_0,e_1\}, & [e_0,w_2]=-\kone^2.
  \end{array}
\right.
\end{equation}
By (\ref{dualbasis}), one has $w_1=k_1e_1+ce_0$ and $w_2=k_1^2e_2+ae_1+be_0$ for some constants $a$, $b$, and $c$.
The eigenvectors in (\ref{eigenvectors}) can also be chosen to have the additional property that for real $\eta^{\third}$ near $\eta=0$, the flux interactions between the eigenvectors $\{\vrp,\vre,\vlp,\vle\}$ are given by the matrix
\begin{equation}\label{fluxmatrix2}
  \left[v_i(\eta),v_j(\eta)\right]_{i,j\in\{+p,+e,-p,-e\}} =
    \renewcommand{\arraystretch}{1.1}
\left[
  \begin{array}{cccc}
    1 & 0 & 0 & 0 \\
    0 & 0 & 0 & -C\zeta\eta^\thirds \\
    0 & 0 & -C\eta^\thirds & 0 \\
    0 & -C\zeta^2\eta^\thirds & 0 & 0
  \end{array}
\right]
\,,
\end{equation}
in which
\begin{equation}
C=3\kone^2>0.\label{UniqueConstantInNormalizedFluxFormPerturbed}
\end{equation}
\end{proposition}

Notice in Proposition~\ref{prop:PuiseuxSeriesEigenvAsympt} that the formula (\ref{fluxmatrix2}) for the flux interactions between the eigenvectors is {\em exact}---there are no  corrections of higher order in~$\eta$.

The rightward and leftward spaces are defined by
\begin{eqnarray}
  && V_+ = V_+(\eta) = \sspan\{\vrp,\vre\}\,,  \qquad \text{(rightward space)} \label{rightwardspace} \\
  && V_- = V_-(\eta) = \sspan\{\vlp,\vle\}\,.  \qquad \text{(leftward space)} \label{leftwardspace}  
\end{eqnarray}

\begin{proposition}\label{prop:PuiseuxSeriesProjAsympt}
In a complex neighborhood of $\eta=0$, the eigenprojection $\Prp(\eta)$ for $K(\eta)$ corresponding to the eigenvalue $\krp(\eta)$ is analytic and the eigenprojections $\Plp(\eta)$, $\Pre(\eta)$, $\Ple(\eta)$ of $K(\eta)$ corresponding to the eigenvalues $\klp(\eta)$, $\kre(\eta)$, $\kle(\eta)$ are the branches of a convergent Laurent-Puiseux series in $\eta^{\third}$, which for real $\eta^{\third}$ near $\eta=0$ are
\begin{eqnarray}
\Prp &=& \frac{[\vrp,\cdot\;]}{[\vrp,\vrp]}\vrp=[e_+,\cdot]e_++O(\eta) \label{Prp}\\
\Plp &=& \frac{[\vlp,\cdot\;]}{[\vlp,\vlp]}\vlp \label{Plp}\\
       &=& \frac{-1}{C}\left(\eta^{-\thirds}[e_0,\cdot]e_0+\eta^{-\third}\left([w_1,\cdot]e_0+[e_0,\cdot]w_1\right)+\left([w_2,\cdot]e_0+[w_1,\cdot]w_1+[e_0,\cdot]w_2\right)\right)+O(\eta^\third) \notag\\
\Pre &=& \frac{[\vle,\cdot\;]}{[\vle,\vre]} \vre \label{Pre}\\
       &=& \frac{-1}{C}\left(\zeta\eta^{-\thirds}[e_0,\cdot]e_0+\zeta^{2}\eta^{-\third}\left([w_1,\cdot]e_0+[e_0,\cdot]w_1\right)+\left([w_2,\cdot]e_0+[w_1,\cdot]w_1+[e_0,\cdot]w_2\right)\right)+O(\eta^\third) \notag\\
\Ple &=& \frac{[\vre,\cdot\;]}{[\vre,\vle]} \vle \label{Ple}\\
       &=& \frac{-1}{C}\left(\zeta^{2}\eta^{-\thirds}[e_0,\cdot]e_0+\zeta\eta^{-\third}\left([w_1,\cdot]e_0+[e_0,\cdot]w_1\right)+\left([w_2,\cdot]e_0+[w_1,\cdot]w_1+[e_0,\cdot]w_2\right)\right)+O(\eta^\third) \notag
\end{eqnarray}
as $\eta\rightarrow 0$.
\end{proposition}

\smallskip
{\bfseries Nondegenerate basis of modes.}\;
The basis \,$\{\vlp,\vle\}$\, for the leftward space is degenerate in the limit $\eta\to 0$ because $\vlp-\vle\to0$.  By scaling this difference by $\eta^\mthird$, one obtains a vector $\vll$ 
whose limit is a multiple of $w_1$,
\begin{eqnarray}\label{def_vll}
  \vll(\eta) &:=& \eta^\mthird \big(\vlp(\eta)-\vle(\eta)\big) \,=\, (1-\zeta^2)w_1 + \eta^\third(1-\zeta)w_2 + \O(\eta^\thirds) \,, \\
  \lim_{\eta\to0} v_\lfl(\eta) &=& (1-\zeta^2) w_1.   \notag
\end{eqnarray}
Since $w_1=k_1e_1+\varepsilon^0(w_1)e_0$ with $k_1\not=0$, the vector $\vll(0)$ carries negative energy flux $[\vll,\vll]=-C$ (see (\ref{jordanform})) and excites the linear mode $(e_1+ize_0)e^{ik_0z}$ of the Maxwell ODEs.
The basis \,$\{\vll,\vle\}$\, for the leftward space is well behaved as $\eta\to0$ and has limit $\{(1-\zeta^2)w_1,e_0\}$.
The flux interactions between the modes $\{\vrp,\vre,\vll,\vle\}$ are given by the matrix
\begin{equation}\label{fluxmatrix3}
  \left[v_i,v_j\right]_{i,j\in\{+p,+e,-\ell,-e\}} =
    \renewcommand{\arraystretch}{1.1}
\left[
  \begin{array}{cccc}
    1 & 0 & 0 & 0 \\
    0 & 0 & C\zeta\eta^\third & -C\zeta\eta^\thirds \\
    0 & C\zeta^2\eta^\third & -C & 0 \\
    0 & -C\zeta^2\eta^\thirds & 0 & 0
  \end{array}
\right]\qquad
  \text{($\eta^{\third}$ real)}.
\end{equation}

\subsubsection{EM mode structure in the slow-light regime}
The algebraic structure in the slow-light regime $|\eta|\ll 1$ describe in sec.~\ref{sec:UnidirAmbMed} allows us to interpret the EM eigenmodes as leftward and rightward (for $\eta\not=0$) based on the sign of the energy flux for the propagating modes or the sign of the imaginary part of the wavenumbers for the evanescent modes.  It also allows a convenient analysis of the limit of the spaces spanned by these leftward and rightward modes as $\eta\rightarrow 0$.

From Proposition \ref{prop:PuiseuxSeriesEigenvAsympt}, the general solution to the Maxwell ODE (\ref{canonical}) in the periodic medium (\ref{periodicmaterials})~is
\begin{equation}\label{MaxwellODEsPeriodicGeneralSolution}
  \psi(z) \,=\, F(z)\left( a\,\vrp e^{i\krp z} + b\,\vlp e^{i\klp z} + c\,\vre e^{i\kre z} + d\,\vle e^{i\kle z} \right),
  \quad
  a,b,c,d\in\CC.
\end{equation}
The matrix (\ref{fluxmatrix2}) gives the energy-flux interactions among the four modes, which is independent of $z$, and implies
\begin{equation}\label{conservation}
  [\psi(z),\psi(z)] \,=\,  |a|^2 -C\eta^\thirds \left( |b|^2 + 2\Re(\zeta d\bar c) \right).\qquad
  \text{($\eta^{\third}$ real)}
\end{equation}
\smallskip
{\bfseries Rightward and leftward mode spaces.}\;
In the following discussion, we will assume that $\eta^\third>0$ near $\eta=0$ and $k_1>0$ (see (\ref{k1})).  This implies by Proposition \ref{prop:PuiseuxSeriesEigenvAsympt} and equations (\ref{MaxwellODEsPeriodicGeneralSolution}) and (\ref{conservation}) that the EM eigenmodes with corresponding wavenumbers $\klp$, $\krp$ are propagating with negative and positive energy flux, respectively, whereas the EM eigenmodes with corresponding wavenumbers $\kle$, $\kre$ are evanescent and decay exponentially to zero as $z\rightarrow-\infty$ and $z\rightarrow\infty$, respectively, since $\operatorname{Im}\kle<0$ and $\operatorname{Im}\kre>0$.

The modes are designated as follows:
\begin{equation}\label{modes}
\renewcommand{\arraystretch}{1}
\left.
  \begin{array}{rl}
    \text{leftward propagating:} & \wlp(z) = F(z)\,\vlp e^{i\klp z} \\
    \text{rightward propagating:} & \wrp(z) = F(z)\,\vrp e^{i\krp z} \\
    \text{leftward evanescent:} & \wle(z) = F(z)\,\vle e^{i\kle z} \\
    \text{rightward evanescent:} & \wre(z) = F(z)\,\vre e^{i\kre z}\,.
  \end{array}
\right.
\end{equation}
Denote by $\Pl$ and $\Pr$ the rank-2 complementary projections onto the {leftward and rightward spaces},
\begin{eqnarray}
  && \Pl = \Plp + \Ple\,, \quad \Range(\Pl) = V_- = \sspan\{\vlp,\vle\} \label{leftwardprojection}\\
  && \Pr = \Prp + \Pre\,, \quad \Range(\Pr) = V_+ = \sspan\{\vrp,\vre\} \label{rightwardprojection}\,.
\end{eqnarray}

The form of the flux interaction matrix (\ref{fluxmatrix2}) plays an important role in the way fields are scattered by an obstacle.  The critical fact is that each oscillatory mode carries energy in isolation while the evanescent modes induce energy flux only when superimposed with one another.  This idea is manifest in the flux-adjoints of the projection operators: 
Both projections $\Plp$ and $\Prp$ onto the ``propagating subspaces" are flux-self-adjoint, and the projections $\Ple$ and $\Pre$ onto the ``evanescent subspaces" are flux-adjoints of each~other:
\begin{eqnarray}
  && [\Plp\psi_1,\psi_2] = [\psi_1,\Plp\psi_2]\,,\\
  && [\Prp\psi_1,\psi_2] = [\psi_1,\Prp\psi_2]\,,\\
  && [\Ple\psi_1,\psi_2] = [\psi_1,\Pre\psi_2]\,.
\end{eqnarray}
In concise notation, $\Plrp^\fadj=\Plrp$ and $\Ple^\fadj=\Pre$.

\smallskip
{\bfseries Limits of rightward and leftward spaces.}\;
For each (small enough) nonzero value of~$\eta$, the spaces $V_-$ and $V_+$ span $\CC^4$.  But their limits do not, as discussed by Figotin and Vitebskiy \cite{FigotinVitebskiy2006}:
\begin{eqnarray}
  && \Vo_+ \,:=\, \lim_{\eta\to0} V_+(\eta) = \sspan\{e_+,e_0\}\,,\label{Voplus} \\
  && \Vo_- \,:=\, \lim_{\eta\to0} V_-(\eta) = \sspan\{e_0,e_1\}\,. \label{Vominus}
\end{eqnarray}
The span of the limiting spaces $\Vo_-$ and $\Vo_+$ is only a three-dimensional subspace of $\CC^4$, and the intersection of $\Vo_-$ and $\Vo_+$ is the zero-flux eigenspace:
\begin{eqnarray}
  \Vo &:=& \Vo_+ + \Vo_- \,=\, \sspan\{ e_+,e_0,e_1 \}\,, \label{Vo}\\
  && \Vo_+ \cap \Vo_- \,=\, \sspan\{ e_0 \}\,.
\end{eqnarray}
The space $\Vo$ excludes all vectors that excite the quadratic mode, that is, all vectors with a nonzero $e_2$ component.
As a consequence of the fact $\Vo_+ + \Vo_- \not= \CC^4$, the projections $\Pr$ and $\Pl$ must be singular functions of $\eta$, ceasing to exist at $\eta=0$; indeed Proposition \ref{prop:PuiseuxSeriesProjAsympt} reveals that $\Pr(\eta)$ and $\Pl(\eta)$ are Laurent-Puiseux series in $\eta^\third$ whose lowest-order term is $\eta^{-\thirds}$.

The limit (\ref{Voplus}) is understood as the image of the norm-limit as $\eta\to0$ of the operator
\begin{equation}
  Q_+(\eta)\,v \,=\, \frac{(\vrp(\eta),v)}{(\vrp(\eta),\vrp(\eta))}\vrp(\eta) + \frac{(\vre(\eta),v)}{(\vre(\eta),\vre(\eta))}\vre(\eta)
\end{equation}
(which is not a projection), which maps onto the space $V_+(\eta)$.
The limit of $Q_+(\eta)$, denoted by $\mathring Q_+$, maps onto the subspace of $\CC^4$ spanned by the limits of the vectors $\vrp(\eta)$ and $\vre(\eta)$, which are $e_+$ and $e_0$.  Thus the image of $\mathring Q_+$ is $\sspan\{e_+,e_0\}$.  The limit (\ref{Vominus}) is the image as $\eta\to0$ of the norm-limit of the operator
\begin{equation}
  Q_-(\eta)v = \frac{(\vll(\eta),v)}{(\vll(\eta),\vll(\eta))}\vll(\eta) + \frac{(\vle(\eta),v)}{(\vle(\eta),\vle(\eta))}\vle(\eta)
\end{equation}
onto the space $V_-(\eta)$.
The limiting operator $\mathring Q_-$ has image spanned by the limits of $\vll$ and $\vle$, which are $(1-\zeta^2)w_1$ and $e_0$.  Since $w_1=k_1e_1+ce_0$ by Proposition~\ref{prop:PuiseuxSeriesEigenvAsympt} and $k_1\not=0$, the image of $\mathring Q_-$ is $\sspan\{e_0,e_1\}$.

\medskip
{\bfseries Note on notation ($X\to\mathring X$).}\, Many of the quantities we are dealing with depend (usually tacitly) on~$\eta$ and may have a meaning at $\eta=0$ (such as the matrices $T, K$, the spaces $V_\pm$, and the vectors $v_{\pm p,e}$).  We will always denote the value of an $\eta$-dependent quantity $X$ at $\eta=0$ by $\mathring{X}$.

\subsection{Scattering by a defect layer}\label{sec:scatteringdefectlayer}

We now introduce a single layer, or slab, of a lossless medium extending from $z=0$ to $z=L$ with permittivity $\epsilon(z)$ and permeability $\mu(z)$ satisfying (\ref{LosslessLayeredMedia}) and (\ref{PassiveMedia}).  This layer acts as a scatterer of electromagnetic waves originating in the ambient medium, but it can act simultaneously as a waveguide.
Scattering of waves by a defect slab is developed in detail in our previous work~\cite{ShipmanWelters2013}.  The theory there applies in the present study for $0<\eta\ll1$.  In the limit $\eta\to0$ the solution of the scattering problem becomes pathological
because the concept of rightward and leftward modes breaks down, as discussed in the previous section.

\subsubsection{The scattering problem (for $\eta\not=0$)}\label{sec:scattering}

It is convenient to choose the point $z=L$ so that the electromagnetic coefficients $\epsilon(z)$ and $\mu(z)$ in the period $[-d,0]$ are identical to those in the period $[L,L+d]$ ({\itshape i.e.}, $\epsilon(z+L+d)=\epsilon(z)$ and $\mu(z+L+d)=\mu(z)$ for all $z\in[-d,0]$) as in Fig.~\ref{fig:layered}. Thus, in the ambient (periodic) space, any solution $\psi$ of the Maxwell ODEs (\ref{canonical}) has the form
\begin{equation}\label{ambientfield}
\renewcommand{\arraystretch}{1.0}
\left.
  \begin{array}{rl}
  z<0: & \psi(z) \;=\; F(z)e^{iKz}\psi(0) \,, \\
  \vspace{-1.5ex}\\
  z>L: & \psi(z) \;=\; F(z-L)e^{iK(z-L)}\psi(L)\,,    
  \end{array}
\right.
\end{equation}
with $\psi(0)$ and $\psi(L)$ related through the flux-unitary transfer matrix across the slab $T=T(0,L)$,
\begin{equation}\label{SlabTransferMatrix}
  T\psi(0)=\psi(L),
  \qquad
  [T\psi_1,\psi_2] = [\psi_1,T^{-1}\psi_2].
\end{equation}
%


A solution $\psi(z)$ to the scattering problem, {\itshape i.e.}, \emph{a scattering field}, is a solution of the Maxwell ODEs (\ref{canonical}) that is decomposed outside the slab into incoming and outgoing parts.  This decomposition expands the solution~(\ref{ambientfield}) into physically meaningful modes,
\begin{equation}\label{totalfield}
\renewcommand{\arraystretch}{1.0}
\left.
  \begin{array}{rrl}
  z<0: & \psi(z) &=\; \psi_-^\out(z) + \psi_+^\inc(z)\,, \\
  \vspace{-1ex}\\
  z>L: & \psi(z) &=\; \psi_-^\inc(z) + \psi_+^\out(z)\,.    
  \end{array}
\right.
\end{equation}
The scattering problem is illustrated in Fig.~\ref{fig:scattering}.

\begin{figure}
\centerline{\scalebox{0.4}{\includegraphics{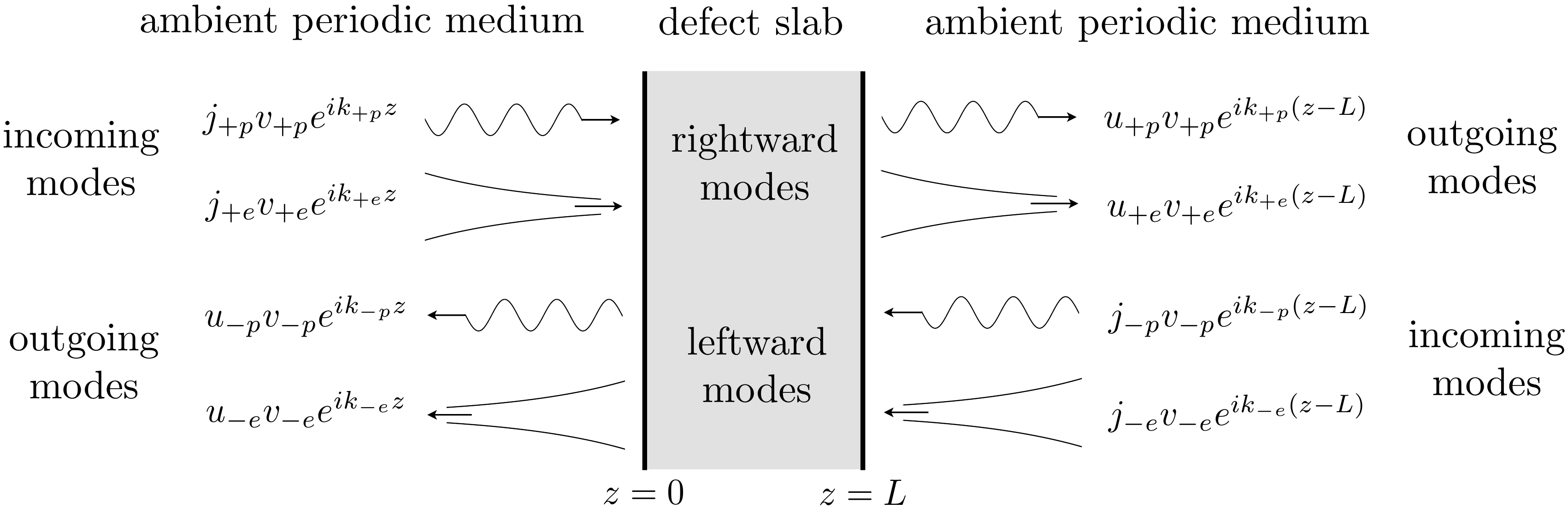}}}
\caption{\small  Scattering of harmonic source fields in a periodic anisotropic ambient medium by a defect layer (slab), at frequency $\omega$ and wavevector $\kk$ parallel to the slab for which there is a pair of rightward and leftward propagating modes and a pair of rightward and leftward evanescent modes.
At such $\kw$, the medium does not admit a frozen mode, that is, $\kw\not=\kwz$, or $0<\eta\ll1$.  Sources on each side of the slab emit fields, which arrive at the slab as {\em incoming} propagating and evanescent modes.  The modes of the {\em outgoing} field are directed in the opposite direction to the incoming modes.  The periodic prefactors $F(z)$ (left) and $F(z-L)$ (right) are~omitted in the figure.}
\label{fig:scattering}
\end{figure}

The field $\psi(z)$ is determined by $\psi(0)$ or $\psi(L)$ alone, and these boundary values are related through the slab transfer matrix $T$ by $T\psi(0)=\psi(L)$.
By defining incoming and outgoing vectors in $\CC^4$ as
\begin{equation}
  \Psi^\inc = \psi_-^\inc(L) + \psi_+^\inc(0)\,,
  \qquad
  \Psi^\out = \psi_-^\out(0) + \psi_+^\out(L)\,,
\end{equation}
the scattering problem may be written concisely as
\begin{equation}\label{scattering2}
  (T\Pl - \Pr) \Psi^\out + (T\Pr - \Pl)\Psi^\inc \,=\, 0\,.
  \qquad
  \text{(scattering problem)}
\end{equation}
This formulation reduces the scattering problem to finite-dimensional linear algebra involving the analytic slab transfer matrix $T=T(\eta)$ and the eigenspaces of $K=K(\eta)$.

The scattering problem is uniquely solvable for $\Psi^\out$ whenever $T\Pl-\Pr$ is invertible.   When it is not invertible, there is a nonzero vector $\Psi^\out$ that solves (\ref{scattering2}) with $\Psi^\inc=0$.
In this case, the projections of the vector $\Psi^\out$ onto the leftward and rightward subspaces are the traces of a solution $\psig(z)$ of the Maxwell ODEs that is outgoing as $|z|\to\infty$.  Because it has no incoming component on either side, conservation of energy (\ref{conservation}) requires that it be exponentially decaying as $|z|\to\infty$ and have zero energy flux for all $z$ (as described in detail in \cite{ShipmanWelters2013}).  Such a field is a {\em guided mode} of the slab, and its vector $\Psi^\out=\Psi^\guided$ of traces at $z=0$ and $z=L$ satisfies
\begin{equation}\label{guidedmode1}
  (T\Pl - \Pr) \Psi^\guided  \,=\, 0\,.
  \qquad
  \text{(guided-mode equation)}
\end{equation}


\subsubsection{Guided frozen modes (at $\eta=0$)}

Returning to the perturbation of a periodic medium that admits a frozen mode at $\kwz$, we assume that (\ref{guidedmode1}) admits no nonzero solution for $\kw$ near $\kwz$, or equivalently that {\em no guided mode $\psi^\guided$ exists for sufficiently small $\eta>0$}.
If fact, if one considers a perturbation of the frequency alone ($\kk(\eta)=\kk^0$) with $\frac{d\omega}{d\eta}(0)\not=0$, this assumption necessarily holds, as stated in the following theorem proved in sec.~\ref{sec:analyticperturbationtheory}.

\begin{theorem}[nondegeneracy of guided-mode condition]\label{thm:noguidemodesundernondegeneracycondition}
If $\kk(\eta)=\kk^0$ for $0<|\eta|\ll 1$ and $\frac{d\omega}{d\eta}(0)\not=0$, then the generic condition (\ref{GenericCondition}) is satisfied and
\begin{equation}
\det(T(\eta)\Pl(\eta)-\Pr(\eta))\not=0,
\qquad 0<|\eta|\ll 1.
\end{equation} 
\end{theorem}

The nature of the unidirectional limit depends on whether the system at $\eta\!=\!0$ admits a source-free field, analogous to the guided mode $\psi^\guided$ above.
The condition for $\psi^\guided$ is that the eigenmode $\vle$ match the eigenmode $\vre$ across the slab, that is, $T\vle = \ell\vre$ for some nonzero multiple of~$\ell$~\cite[sec. 2B and Theorem 4.1]{ShipmanWelters2013}.
Both $\vle$ and $\vre$ tend to the zero-flux eigenmode $e_0$ as $\eta\to0$; thus the guided-mode condition becomes
\begin{equation*}
  \To e_0 = \ell e_0 \quad \text{for some $\ell\in\CC$}\,.
  \qquad \text{(guided frozen mode condition)}
\end{equation*}
The corresponding solution of the Maxwell ODE, denoted by $\psi^0(z) = T(0,z)e_0$, is a {\em guided frozen mode}, and in the ambient medium it has the form
\begin{equation}\label{guidedfrozenmode}
  \psi^0(z) =
  \renewcommand{\arraystretch}{1}
\left\{
  \begin{array}{ll}
    F(z) e_0 e^{i\kzero z} & z\leq0\,, \\
    \ell\, F(z-L) e_0 e^{i\kzero(z-L)} & z\geq L\,.
  \end{array}
\right.
\qquad \text{(guided frozen mode)}
\end{equation}
We emphasize our assumption that no guided mode exists for sufficiently small $\eta>0$.
{Despite the oscillations of $\psi^0(z)$ due to a possibly nonzero $\kzero$, the energy flux of $\psi^0$ in the $z$ direction vanishes: $[\psi^0,\psi^0]=0$.}  The full electromagnetic field has an oscillatory factor of $e^{ik_1x+k_2y}$, imparting energy flux in the direction of $\kk=(k_1,k_2)$ parallel to the layers.  One can conceive of $\psi^0$ as a guided mode of the slab, even though it does not decay as $|z|\to\infty$.

\subsubsection{Slow-light limit of the scattering problem}

In sec.~\ref{sec:pathological}, we separate the analysis of pathological scattering as $\eta\to0$ into scattering from the left and scattering from the right; both problems are illustrated in Fig.~\ref{fig:leftrightscattering}.

\emph{Scattering from the left} means that the incident field emanates from a source to the left of the slab.  The incoming field consists of the incident field on the left of the slab, and the outgoing field consists of the reflected field on the left and the transmitted field on the right.  Thus the incoming and outgoing fields vectors are
%
\begin{equation}
\begin{aligned}
  \Psi^\inc &= j_\rtp \vrp + j_\rte \vre\,, \\
  \Psi^\out &= t_\rtp \vrp + t_\rte \vre + r_\lfp \vlp + r_\lfe \vle \\
                 &= t_\rtp \vrp + t_\rte \vre + \rho_\lfl \vll + \rho_\lfe \vle \,,
\end{aligned}
\qquad  \text{(scattering from the left)}
\end{equation}
with the change of coordinate
\begin{equation}\label{coordinatechange}
  \rho_\lfe=r_\lfe+r_\lfp
  \quad\text{and}\quad
  \rho_\lfl=\eta^\third r_\lfp\,.
\end{equation}
The scattering fields at $z=0$ and $z=L$ are
\begin{eqnarray*}
  \psi(0) &=& j_\rtp v_\rtp + j_\rte v_\rte + r_\lfp v_\lfp + r_\lfe v_\lfe \\
             &=& j_\rtp v_\rtp + j_\rte v_\rte + \rho_\lfl v_\lfl + \rho_\lfe v_\lfe, \\
  \psi(L) &=& t_\rtp v_\rtp + t_\rte v_\rte.
\end{eqnarray*}
Invariance of the energy flux $[\psi(z),\psi(z)]$ together with the flux relations~(\ref{fluxmatrix2},\ref{fluxmatrix3}) provide the conservation laws
\begin{equation}\label{leftconservation}
\begin{aligned}
  |t_\rtp|^2 &= |j_\rtp|^2 - C\eta^\thirds |r_\lfp|^2 - 2C\eta^\thirds \Re(\zeta\,\bar j_\rte r_\lfe)\,, \\
  |t_\rtp|^2 &= |j_\rtp|^2 - C |\rho_\lfl|^2 - 2C\eta^\thirds \Re(\zeta\,\bar j_\rte \rho_\lfe)
                         + 2C\eta^\third \Re(\zeta\,\bar j_\rte \rho_\lfl)\,.
\end{aligned}                         
\end{equation}

\begin{figure}
\centerline{\scalebox{0.35}{\includegraphics{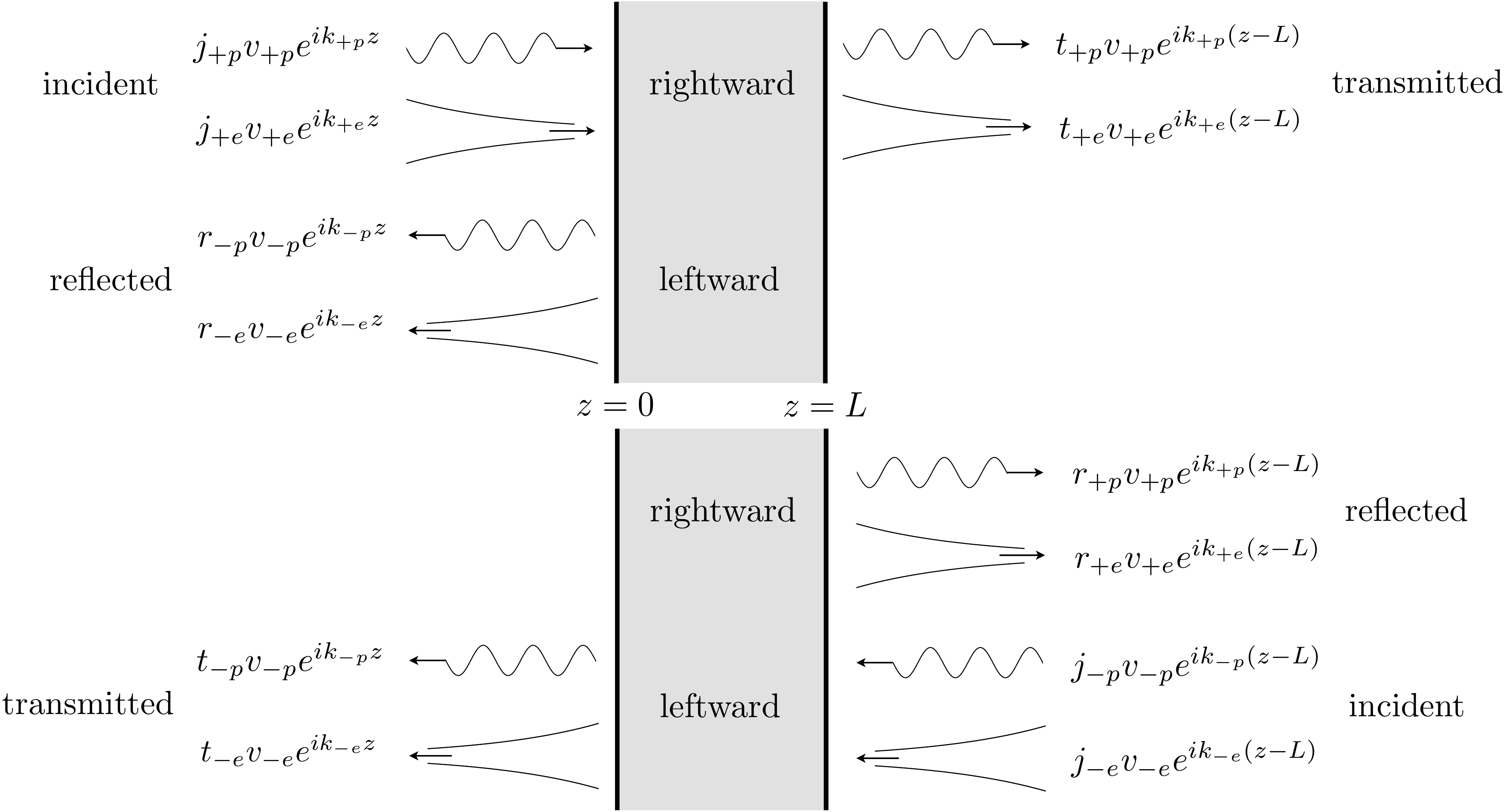}}}
\caption{\small Diagrams of scattering from the left (above) and scattering from the right (below).}
\label{fig:leftrightscattering}
\end{figure}

\noindent
By applying the projections $\Pr$ and $\Pl$ to the scattering problem (\ref{scattering2}) and then taking the flux with each of the eigenvectors yields two equations that are solved successively for the coefficients of the transmitted and the reflected fields:
\begin{equation}\label{leftscattering}
\begin{aligned}
 \mat{1.2}{\left[T\vrp,\vrp\right]}{\left[T\vrp,\vre\right]}{\left[T\vle,\vrp\right]}{\left[T\vle,\vre\right]}
\col{1.2}{t_\rtp}{t_\rte}
= \col{1.2}{j_\rtp[\vrp,\vrp]}{j_\rte[\vle,\vre]}
= \col{1.2}{j_\rtp}{-C\zeta^2\eta^\thirds j_\rte} \,, \\
 \mat{1.2}{\left[T\vlp,\vrp\right]}{\left[T\vlp,\vre\right]}{\left[T\vre,\vrp\right]}{\left[T\vre,\vre\right]}
\col{1.2}{t_\rtp}{t_\rte}
= \col{1.2}{r_\lfp[\vlp,\vlp]}{r_\lfe[\vre,\vle]}
= -C\eta^\thirds\col{1.2}{r_\lfp}{\zeta\, r_\lfe} \,. \\
\end{aligned}
\qquad
\text{(scattering from the left)}
\end{equation}

\emph{Scattering from the right} means
\begin{equation}
\begin{aligned}
  \Psi^\inc &= j_\lfp \vlp + j_\lfe \vle \\
                 &= i_\lfl \vll + i_\lfe \vle\,, \\
  \Psi^\out &= t_\lfp \vlp + t_\lfe \vle + r_\rtp \vrp + r_\rte \vre \\
                 &= \tau_\lfl \vll + \tau_\lfe \vle + r_\rtp \vrp + r_\rte \vre\,,
\end{aligned}
\qquad  \text{(scattering from the right)}
\end{equation}
with the two representations being related by a coordinate change analogous to (\ref{coordinatechange}).
The conservation laws are
\begin{equation}\label{rightconservation}
\begin{aligned}
  -C\eta^\thirds |t_\lfp|^2 &= -C\,\eta^\thirds|j_\lfp|^2 + |r_\rtp|^2 - 2C\,\eta^\thirds \Re(\zeta\,\bar r_\rte j_\lfe)\,, \\
  -C|\tau_\lfl|^2 &= -C|i_\lfl|^2 + |r_\rtp|^2 - 2C\eta^\thirds \Re(\zeta\,\bar r_\rte i_\lfe)
                        + 2C\,\eta^\third \Re(\zeta\,\bar r_\rte i_\lfl)\,,
\end{aligned}                         
\end{equation}
and the equations for the coefficients of the transmitted and reflected coefficients are
\begin{equation}\label{rightscattering}
\begin{aligned}
\mat{1.2}{\left[\vlp,T\vlp\right]}{\left[\vlp,T\vle\right]}{\left[\vre,T\vlp\right]}{\left[\vre,T\vle\right]}
\col{1.2}{t_\lfp}{t_\lfe}
= \col{1.2}{j_\lfp[\vlp,\vlp]}{j_\lfe[\vre,\vle]}
= -C\eta^\thirds \col{1.2}{j_\lfp}{\zeta\, j_\lfe} \,, \\
\mat{1.2}{\left[\vrp,T\vlp\right]}{\left[\vrp,T\vle\right]}{\left[\vle,T\vlp\right]}{\left[\vle,T\vle\right]}
\col{1.2}{t_\lfp}{t_\lfe}
= \col{1.2}{r_\rtp[\vrp,\vrp]}{r_\rte[\vle,\vre]}
= \col{1.2}{r_\rtp}{-C\zeta^2\eta^\thirds r_\rte} \,. \\
\end{aligned}
\qquad
\text{(scattering from the right)}
\end{equation}

A crucial quantity is the determinant of the matrix in~(\ref{leftscattering}),
\begin{eqnarray}
  && D=D(\eta) := [T\vrp,\vrp][T\vle,\vre] - [T\vrp,\vre][T\vle,\vrp]\,, \label{D}\\
  && \Do := D(0) = [\To e_+,e_+][\To e_0,e_0]-[\To e_+,e_0][\To e_0,e_+] \,. \label{Do}
\end{eqnarray}
$D(\eta)$ has a Puiseux series in $\eta$.
The vanishing of $\Do$ is equivalent to the existence of $\psi^0$, as we will see in Theorem~\ref{thm:D}.
In this case, the transfer matrix maps $\Vo$ to itself.
If $\Do\not=0$, then $\To[\Vo]$ has a two-dimensional intersection with~$\Vo$.
The following related quantity will also be relevant:
\begin{equation}
  \mathring D_* := [e_1,\To e_1][e_0,\To e_0]-[e_1,\To e_0][e_0,\To e_1]. \label{Doprime}
\end{equation}
%


\section{Pathological scattering in the slow-light limit}\label{sec:pathological}

We ask the central question, how is the asymptotic behavior, as $\eta\to0$, of the scattering problem reflected in the limiting algebraic structure at $\eta=0$, that is, at the wavevector-frequency pair $\kwz$ of the frozen mode?  There are two cases:
\begin{eqnarray}
  &&\To[\Vo]\not=\Vo \iff \mathring D\not=0 \iff \text{no guided frozen mode} \quad \text{(nonresonant case)} \notag\\
  &&\To[\Vo]=\Vo \iff \mathring D=0 \iff \text{$\exists$ guided frozen mode} \quad\;\, \text{(resonant case)}. \notag
\end{eqnarray}
(Recall the definition (\ref{Vo}) of $\Vo$ and that $\To=T(\eta=0)$.)
These two cases cleanly differentiate between two distinctly different types of scattering behavior, which we call the nonresonant and the resonant.  The resonant case is characterized by the existence of a guided frozen mode when $\kw=\kwz$ ({\itshape i.e.}, $\eta=0$).

\subsection{Algebraic structure at the frozen-mode parameters}

The clear distinction between the nonresonant and resonant cases is manifest in the following theorem, which narrows down what can happen algebraically at the $\eta=0$ limit.

\begin{theorem}[Algebraic structure at frozen-mode parameters]\label{thm:D}   
Let $D(\eta)$ and $\Do$ be defined by (\ref{D},\ref{Do}), $\To=T(0)$ be equal to the transfer matrix across the slab at $\eta=0$ ({\itshape i.e., $\kw=\kwz$}), and $\Vo$, $\Vo_\pm$ be defined by (\ref{Voplus},\ref{Vominus},\ref{Vo}).

If $D(\eta)\not=0$ in a punctured neighborhood of $\eta=0$, then
\begin{equation*}
  \mathrm{dim}\big( \To[\Vo]\cap\Vo \big) = \mathrm{dim}\big( \To^{-1}[\Vo]\cap\Vo \big) =
  \renewcommand{\arraystretch}{1}
\left\{
  \begin{array}{ll}
    2 & \text{if $\Do\not=0$ \; (nonresonant case),} \\
    3 & \text{if $\Do=0$ \; (resonant case),}
  \end{array}
\right.
\end{equation*}
and, with $\mathring W$ equal to any of the four spaces $\To[\Vo_\pm]$ or $\To^{-1}[\Vo_\pm]$,
\begin{equation*}
  \mathrm{dim}\big( \mathring W \cap\Vo \big) = 
  \renewcommand{\arraystretch}{1}
\left\{
  \begin{array}{ll}
    1 & \text{if $\Do\not=0$ \; (nonresonant case),} \\
    2 & \text{if $\Do=0$ \; (resonant case).}
  \end{array}
\right.
\end{equation*}
The following statements are equivalent characterizations of the resonant case.
\begin{enumerate}
 \item $\Do=0$.
 \item $\To e_0 = \ell e_0$ for some nonzero $\ell\in\CC$, that is, the Maxwell ODE admits a guided frozen mode at $\eta=0$.
 \item The space $\Vo$ is invariant under $\To$.
 \item $\To$ or $\To^{-1}$ takes $\Vo_+$ or $\Vo_-$ into $\Vo$.
 \item $\To$ or $\To^{-1}$ takes $e_0$ into $\Vo_-$ or $\Vo_+$.
 \item The transfer matrix $\To$, with respect to the Jordan basis $\{e_+,e_0,e_1,e_2\}$ of the matrix $\widetilde{\mathring{K}}$ in (\ref{jordanform}), has the form
\begin{equation}\label{To}
    \widetilde{\To} =
    \renewcommand{\arraystretch}{1.2}
\left[
  \begin{array}{cccc}
    a_{++} & 0 & a_{+1} & a_{+2} \\
    a_{0+} & \ell & a_{01} & a_{02} \\
    a_{1+} & 0 & a_{11} & a_{12} \\
    0 & 0 & 0 & \bar\ell^{-1}
  \end{array}
\right]
\end{equation}
with \,$|a_{++}|^2 - |a_{1+}|^2 = 1$\,, \,$|a_{11}|^2 - |a_{+1}|^2 = 1$\,, and $\bar a_{++}a_{+1}=\bar a_{1+}a_{11}$\,.  In particular, $|a_{++}|\not=0$ and $|a_{11}|\not=0$.
\item
The matrix $\To^{-1}$, with respect to the Jordan basis, has the form
\begin{equation}\label{Toinv}
    \widetilde{\To^{-1}} =
    \renewcommand{\arraystretch}{1.2}
\left[
  \begin{array}{cccc}
    b_{++} & 0 & b_{+1} & b_{+1} \\
    b_{0+} & \ell^{-1} & b_{01} & b_{02} \\
    b_{1+} & 0 & b_{11} & b_{12} \\
    0 & 0 & 0 & \bar\ell
  \end{array}
\right]
\end{equation}
with \,$|b_{++}|^2 - |b_{1+}|^2 = 1$\,, \,$|b_{11}|^2 - |b_{+1}|^2 = 1$\,, and $\bar b_{++}b_{+1}=\bar b_{1+}b_{11}$\,.  In particular, $|b_{++}|\not=0$ and $|b_{11}|\not=0$.
\end{enumerate}
\end{theorem}

\begin{proof}
First we prove that (b), (c), and (d) are equivalent.  Statement (c) trivially implies (d).  To prove that (d)$\implies$(b), let us assume $\To[\Vo_+]\subset\Vo$; similar arguments apply if $\To$ is replaced by $\To^{-1}$ or $\Vo_+$ by $\Vo_-$.  For some numbers $a_i$ and $b_i$, one has
\begin{eqnarray}
  && \To e_+=a_+e_++a_0e_0+a_1e_1\,,\\
  && \To e_0=b_+e_++b_0e_0+b_1e_1\,.
\end{eqnarray}
The flux relations given in~(\ref{jordanform}) and the flux-unitarity of $T$ imply
\begin{eqnarray}
  && 0 = [e_+,e_0] = [\To e_+,\To e_0] = \bar a_+b_+ - \bar a_1b_1\,,\\
  && 1 = [e_+,e_+] = [\To e_+,\To e_+] = |a_+|^2 - |a_1|^2\,, \\
  && 0 = [e_0,e_0] = [\To e_0,\To e_0] = |b_+|^2 - |b_1|^2\,.
\end{eqnarray}
These three equations are simultaneously tenable only if $b_+=b_1=0$.  Thus $\To e_0=b_0e_0$, and $b_0\not=0$ since $\To$ is invertible.  To show that (b)$\implies$(c), let $\To e_0=\ell e_0$.  Then
\begin{eqnarray*}
   && -\varepsilon^2(\To e_+) = [e_0,\To e_+] = [\To^{-1} e_0,e_+] = \bar\ell^{-1}[e_0,e_+] = 0\,,\\
   && -\varepsilon^2(\To e_1) = [e_0,\To e_1] = [\To^{-1} e_0,e_1] = \bar\ell^{-1}[e_0,e_1] = 0\,,
\end{eqnarray*}
which imply that $\To e_+$ and $\To e_1$ are contained in $\Vo$.  Thus $\To [\Vo]=\Vo$.

We now prove the equivalence of (a) and (b).  Suppose that $\To e_0=\ell e_0$.  Then
\begin{eqnarray*}
   \Do = \big[\To e_+,e_+\big]\big[\To e_0,e_0\big]-\big[\To e_+,e_0\big]\big[\To e_0,e_+\big] = \overline{\ell}\big[\To e_+,e_+\big]\big[e_0,e_0\big]-\overline{\ell}\big[\To e_+,e_0\big]\big[e_0,e_+\big] = 0.
\end{eqnarray*}
and thus $\Do=D(\eta=0)=0$.  Conversely, suppose that $\Do=0$ but that $D(\eta)\not=0$ for nonzero small $\eta$.  The first equation of (\ref{leftscattering}), with $j_\rtp=1$ and $j_\rte=0$ yields
\begin{equation}\label{hi1}
  D(\eta)\col{1.2}{t_\rtp}{t_\rte} = \col{1.2}{\big[T\vle,\vre\big]}{-\big[T\vle,\vrp\big]}
\end{equation}
The conservation law~(\ref{leftconservation}) gives $|t_\rtp|^2=1-3\eta^\thirds|r_\lfp|^2$, so that $t_\rtp$ is bounded.  Since $D(0)=0$, (\ref{hi1}) yields $\big[T\vle,\vre\big]\to0$, or $[e_0,\To^{-1}e_0]=[\To e_0,e_0]=0$.  This implies, by the flux relations in~(\ref{jordanform}), that
\begin{equation*}
  \To e_0 \in \sspan\{e_+,e_0,e_1\}\,
  \quad \text{and} \quad
  \To^{-1}e_0 \in \sspan\{e_+,e_0,e_1\}\,.
\end{equation*}
Since the first term in the defintion~(\ref{D}) of $D$ vanishes at $\eta=0$, one obtains, from the limit of the second term,
\begin{equation*}
  [\To e_+,e_0][\To e_0,e_+]=0.
\end{equation*}
In case $[\To e_+,e_0]=0$, one has \,$\To e_+\in\sspan\{e_+,e_0,e_1\}$.  In case $[e_0,\To^{-1}e_+]=[\To e_0,e_+]=0$, one has 
\,$\To^{-1}e_+\in\sspan\{e_+,e_0,e_1\}$.  Either way, statement (d) is true, and this implies (b), as already proved.

Statement (b) implies (e) because $e_0$ is contained in $\Vo_-$ and in $\Vo_+$.  To prove that (e) implies (b), suppose that $\To e_0$ or $\To^{-1} e_0$ is equal to $v=a e_+ + b e_0 + c e_1$.  By the flux-unitarity of $\To$, $0=[e_0,e_0]=[v,v]=|a|^2-|c|^2$.  Thus if $v\in\Vo_+$, or $c=0$, then $a=0$; and if $v\in\Vo_-$, or $a=0$, then $c=0$.  In either case, $v=b e_0$ thus verifying statement (b).

The form of $\To$ in statement (f) is equivalent to (b) and (c) together.  The lower right entry follows from
\begin{equation*}
  -1 = [e_0,e_2] = [\To e_0,\To e_2] = [\ell e_0,a_{+2}e_++a_{02}e_0+a_{12}e_1+a_{22}e_2] = -\bar\ell a_{22}.
\end{equation*}
The relations among the coefficients $a_{ij}$ come from the flux-unitarity of $\To$ and the flux matrix~(\ref{jordanform}).  Statement (g) is proved similarly.

The statements on $\mathrm{dim}\big( \mathring W \cap\Vo \big)$ and $\mathrm{dim}\big( \To^{\pm 1}[\mathring V] \cap\Vo \big)$ follow from the equivalence of (a), (c), and (d).
\end{proof}


\subsection{Behavior of scattering in the slow-light limit}\label{sec:main}

In this section, we analyze the pathological scattering characteristic of the slow-light limit $\eta\to0$ by considering solutions to the scattering problem, {\itshape i.e.}, \emph{scattering fields}, described in sec.~\ref{sec:scatteringdefectlayer} that are Puiseux series in $\eta^\third$.  Although such fields have limits as $\eta\to0$, their modal components can blow up as negative powers of $\eta^\third$.  Moreover, the limit at $\eta=0$ of a scattering field may exhibit linear or quadratic growth in $z$ on the left or right sides of the defect layer.  This behavior is summarized in Table \ref{table:scattering}.  The strategy we take to analyze pathological scattering is summarized as follows.

\vspace{.5cm}

\noindent
{\bfseries Strategy for analysis of pathological scattering} as $\eta\to0$: Identify all scattering fields, {\itshape i.e.}, solutions of the scattering problem, that are Puiseux series (bounded as $\eta\to0$) and find a distinguished, meaningful basis for them.  This is how it is done:
\renewcommand{\labelenumi}{\arabic{enumi}.}
\begin{enumerate}
\item Separate the full scattering problem into scattering from the left and scattering from the right.  Let us consider scattering from the left to be concrete.
\item Scattering-from-the-left fields have values at $z=L$ that are in $V_+(\eta)$ and values at $z=0$ that are in $T(\eta)^{-1}[V_+(\eta)]$.
\item The limiting values, at $\eta=0$, of these scattering fields, evaluated at $z=L$, are in $\Vo_+$, and evaluated at $z=0$ are in $\To^{-1}[\Vo_+]$.  But the limiting values are not solutions to scattering problems in the sense of sec.~\ref{sec:scattering} because the rightward and leftward spaces become dependent at $\eta=0$ (their intersection is $e_0$).  This is manifest in the blowing up of the modal projections as $\eta\to0$.
\item Consider two distinguished vectors, $v_1$ and $v_2$ in $\To^{-1}[\Vo_+]$ and the corresponding values $\To(v_1)$ and $\To(v_2)$ in $\Vo_+$.
\item Perturb $v_1$ and $v_2$ by general Puiseux series into $T(\eta)^{-1}[V_+(\eta)]$ and project onto the eigenmodes.  This reveals the asymptotic incident and reflected fields, which in general blow up as the mode basis degenerates.
\item Identify special Puiseux perturbations of $v_1$ and $v_2$ in $T(\eta)^{-1}[V_+(\eta)]$ which are scattering fields distinguished by the choice of the incident field.  This will require proving that the given incident field does indeed produce a bounded scattering field as $\eta\to0$.  Proving this may be subject to a generic condition relating the perturbation of the Jordan form structure of the ambient space at $\eta=0$ to the transfer matrix $\To$.
\end{enumerate}

\begin{table}[H]
\centerline{\small
\renewcommand{\arraystretch}{1.2}
\begin{tabular}{| l | c | c || rcl |}
\hline
\multicolumn{3}{|c||}{\parbox{0.35\textwidth}{\centering limiting field ($\eta=0$) \\ $\Vo_++\Vo_-=\Vo$\, is $3$-dimensional}}
& \multicolumn{3}{c|}{\parbox{0.38\textwidth}{\centering scattering field on incident edge ($0<\eta\ll 1$) \\ $V_+(\eta)+V_-(\eta)=\CC^4$\, is $4$-dimensional}}
\\\hline
value at incident edge
& \!\parbox{0.098\textwidth}{\centering growth in $z$\\left ($z<0$)}\!
& \!\parbox{0.112\textwidth}{\centering growth in $z$\\right ($z>L$)}\!
& incident field &$\!+\!$& reflected field
\\\hline\hline
\multicolumn{6}{|l|}{{\bfseries\slshape nonresonant scattering from the left}: source field incident at $z=0$}
\\\hline
\multirow{2}{*}{$\mathring v_\rone= e_+ + c e_0 + c e_1\in\Vo$} 
& \multirow{2}{*}{linear} 
& \multirow{2}{*}{bounded}
& $v_\rtp $ &$\!+\!$& $r v_\lfl + r v_\lfe$
\\
& &
& \multicolumn{3}{c|}{(scattering of rightward oscillatory mode)}
\\\hline
\multirow{2}{*}{$\mathring v_\rtwo = e_2 + ce_0 + ce_1  \not\in\Vo$}
& \multirow{2}{*}{quadratic}
& \multirow{2}{*}{bounded}
& $\eta^\mthirds c\big( v_\rte$ &$\!+\!$& $ \zeta^2 v_\lfp + \zeta v_\lfe \big) + O(1)$
\\
& &
& \multicolumn{3}{c|}{(cancellation of high-amplitude modes)}
\\\hline\hline
\multicolumn{6}{|l|}{{\bfseries\slshape nonresonant scattering from the right}: source field incident at $z=L$}
\\\hline
\multirow{2}{*}{$\mathring v_\lone= e_1 + c e_0 + c e_+\in\Vo$}
& \multirow{2}{*}{linear}
& \multirow{2}{*}{linear} 
& $c v_\lfl $ &$\!+\!$& $r v_\rtp + r v_\rte$
\\
& &
& \multicolumn{3}{c|}{(scattering of leftward generalized eigenmode)}
\\\hline
\multirow{2}{*}{$\mathring v_\ltwo = e_2 + ce_0 + ce_+ \not\in\Vo$}
& \multirow{2}{*}{linear} 
& \multirow{2}{*}{quadratic}
& $\eta^\mthirds c\big( v_\lfp +\zeta^2 v_\lfe$ &$\!+\!$& $\zeta v_\rte \big) + O(1)$
\\
& &
& \multicolumn{3}{c|}{(cancellation of high-amplitude modes)}
\\\hline\hline
\multicolumn{6}{|l|}{{\bfseries\slshape resonant scattering from the left}: source field incident at $z=0$; guided frozen mode $\psi^\guided$ exists}
\\\hline
$e_0$
& \multirow{4}{*}{bounded}
& \multirow{4}{*}{bounded}
& $\eta^\third c v_\rtp$&$\!+\!$& $r v_\lfp + r v_\lfe$
\\
(resonant excitation of $\psi^\guided$)
& & & 
\multicolumn{3}{c|}{(inc. field $\to 0$;\; reflected field $\to e_0$)}
\\\hhline{-~~---}
$e_0$
& &
& $c v_\rte$ &$\!+\!$& $r v_\lfp + r v_\lfe$
\\
(excitation of $\psi^\guided$)
& & &
\multicolumn{3}{c|}{(total field $\to e_0$)}
\\\hline
\multirow{2}{*}{$\mathring w_+= e_+ + c e_1\in\Vo$}
& \multirow{2}{*}{linear}
& \multirow{2}{*}{bounded}
& $v_\rtp - \eta^\mthird c v_\rte$ &$\!+\!$& $\eta^\mthird(r v_\lfp + r v_\lfe)$
\\
& & &
\multicolumn{3}{c|}{(cancellation of high-amplitude modes)}
\\\hline\hline
\multicolumn{6}{|l|}{{\bfseries\slshape resonant scattering from the right}: source field incident at $z=L$; guided frozen mode $\psi^\guided$ exists}
\\\hline
$e_0$
& \multirow{4}{*}{bounded}
& \multirow{4}{*}{bounded}
& $\eta^\third c v_\lfl$ &$\!+\!$& $r v_\rte + \eta^\third r v_\rtp$
\\
(resonant excitation of $\psi^\guided$)
& & & 
\multicolumn{3}{c|}{(inc. field $\to 0$;\; reflected field $\to e_0$)}
\\\hhline{-~~---}
$e_0$
& &
& $c v_\lfp + c v_\lfe$ &$\!+\!$& $r v_\rte + \eta^\third r v_\rtp$
\\
(excitation of $\psi^\guided$)
& & &
\multicolumn{3}{c|}{(generic constants $c$:\; total field $\to e_0$)}
\\\hline
\multirow{2}{*}{$\mathring w_- = e_1 + ce_+\in\Vo$}
& \multirow{2}{*}{linear}
& \multirow{2}{*}{linear}
&  $\eta^\mthird(c v_\lfp + c v_\lfe)$ &$\!+\!$& $\eta^\mthird r v_\rte + r v_\rtp$ 
\\
& & &
\multicolumn{3}{c|}{\!\!\!(special $c$'s: cancellation of high-amplitude modes)\!\!\!}
\\\hline
\end{tabular}}
\caption{\small This table is a summary of Theorems~\ref{thm:leftscatteringreg}--\ref{thm:rightscatteringres}.
(The symbol $c$ stands for a constant, which may be different in every occurrence, and $r=r(\eta)$ stands for a Puiseux series that may be different in each occurrence.) 
Near but not at the stationary inflection point of the dispersion relation for the ambient periodic medium, which corresponds to $0<\eta\ll1$, the scattering fields resulting from incident fields coming from the left of the defect slab form a two-dimensional space, and so do the fields resulting from incident fields from the right.  A well chosen basis for each of these spaces consists of two scattering fields (rightmost column) whose limits as $\eta\to0$ have values at the incident edge that form a distinguished basis for the limiting scattering fields (left).
Although the scattering fields themselves are (convergent) Puiseux series in $\eta$, their incident and reflected components at the incident edge may blow up as $\eta\to0$.  Since $\zeta$ is a primitive cube root of $1$, one has $1+\zeta+\zeta^2=0$.}
\label{table:scattering}
\end{table}

\subsubsection{Nonresonant scattering}\label{nonresonant}

{\bfseries Nonresonant scattering from the left.}\;
In scattering from the left for $\eta\not=0$, the space of all scattering fields $\psi(z)$, evaluated at the right edge of the slab ($z=L$), comprise the two-dimensional space $V_+(\eta)$ of transmitted rightward-directed fields.  The values of the fields at the left edge of the slab ($z=0$) comprise the space $T^{-1}(\eta)[V_+(\eta)]$.  Application of the projections (\ref{Prp}--\ref{Ple}) to the field $\psi(0)$ reveals the source and reflected field mode components.  All of these projections except for $\Prp$ blow up to order $\eta^\mthirds$.  Thus, {\em even if the field $\psi(0)$ remains bounded as $\eta\to0$, the modal components making up the incident and reflected fields may become unbounded}, but their large amplitudes mutually cancel to produce a bounded field.

As $\eta\to0$, the space of transmitted fields $V_+(\eta)$ at the right edge of the slab and the space of corresponding field values at the left edge have limits $\Vo_+=\sspan\{e_+,e_0\}$ (at $z=L$) and $\To^{-1}[\Vo_+]$ (at $z=0$), respectively.  {\em In the nonresonant case, Theorem~\ref{thm:D} guarantees that $\To^{-1}[\Vo_+]$ is not contained in $\Vo=\sspan\{e_+,e_0,e_1\}$.}
Recall that $\Vo$ is the three-dimensional span of the limits of the rightward and leftward spaces $V_+$ and $V_-$.  The fact that limits of certain scattering fields at the left edge of the slab do not lie in $\Vo$ is a manifestation of the pathological behavior of wave scattering as one approaches the parameters $\kwz$ ($\eta=0$) of the frozen mode.

The two-dimensional limiting space $\To^{-1}[\Vo_+]$ of fields at $z=0$ retains essential information about the asymptotics of the scattering fields as $\eta\to0$.  This space is spanned by two distinguished vectors, which are treated fully in Theorem~\ref{thm:leftscatteringreg} below.  One of them spans the intersection $\To^{-1}[\Vo_+] \cap \Vo$, which is one-dimensional by Theorem~\ref{thm:D}.  We denote it by $\mathring v_\rone$, and it is the limiting value as $\eta\to0$ of the scattering field produced when the incident field is the rightward propagating mode $F(z)\vrp e^{i\krp z}$.
The other, denoted by $\mathring v_\rtwo$, has a nonzero $e_2$-component and thus does not lie in $\Vo$.  It is the limiting value as $\eta\to0$ of a scattering field that exhibits quadratic growth in $z$ to the left of the slab and is produced by an incident rightward evanescent field with amplitude that blows up like $\eta^\mthirds$.
{\em All fields resulting from scattering from the left are linear combinations of these two distinguished scattering problems}.

\smallskip
Before stating the theorem, let us take a look at scattering fields whose limits as $\eta\to0$ are equal to $\mathring v_\rone$ or $\mathring v_\rtwo$ at $z=0$.  We consider Puiseux-series perturbations in $\eta^\third$ of these two vectors into the space $T(\eta)^{-1}[V_+(\eta)]$ of scattering fields at $z\!=\!0$ and the projections onto their rightward (incoming) and leftward (reflected) components.  Denote these perturbations by
\begin{eqnarray}
  && v_\rone(\eta) = \mathring v_\rone + \eta^\third v_\rone^1 + O(\eta^\thirds), \label{rone} \\
  && v_\rtwo(\eta) = \mathring v_\rtwo + \eta^\third v_\rtwo^1 + O(\eta^\thirds). \label{rtwo}
\end{eqnarray}
Each family of fields $\psi(z;\eta)$ produced by scattering from the left, and whose value at $z=0$ admits a Puiseux series, satisfies
\begin{equation}\label{rcombination}
  \psi(0;\eta) = \alpha(\eta)v_\rone(\eta) + \beta(\eta)v_\rtwo(\eta)
\end{equation}
for some Puiseux series $\alpha(\eta)$ and $\beta(\eta)$.

The projections of $v_\rone(\eta)$ onto the four modes are
\begin{equation}\label{roneprojections}
  \renewcommand{\arraystretch}{1.3}
\left.
  \begin{array}{rl}
     \text{incident} &
     \renewcommand{\arraystretch}{1.3}
      \left\{
      \begin{array}{lcl}
         \Prp v_\rone &=& e_+ + \eta^\third [e_+,v_\rone^1]e_+ + O(\eta^\thirds) \\
         \Pre v_\rone &=& \eta^\mthird\left( \zeta^2 a + \zeta b \right)e_0 + O(1)
      \end{array}
      \right.
      \\
     \text{reflected} &
     \renewcommand{\arraystretch}{1.3}
      \left\{
      \begin{array}{lcl}
         \Plp v_\rone &=& \eta^\mthird\left( a + b \right)e_0 + O(1) \\
         \Ple v_\rone &=& \eta^\mthird\left( \zeta a + \zeta^2 b \right)e_0 + O(1)
      \end{array}
      \right.
  \end{array}
\right.
\;
\text{(bounded left scattering fields near $\mathring v_\rone$)},
\end{equation}
in which
\begin{equation}
  \renewcommand{\arraystretch}{1.3}
\left\{
  \begin{array}{lcl}
    a &=& - \frac{1}{C} [w_1,\mathring v_\rone] = \frac{1}{C} [w_1,e_1][e_1,\mathring v_\rone]=-\frac{1}{3k_1}[e_1,\mathring v_\rone]\,, \\
    b &=& -\frac{1}{C} [e_0,v_\rone^1]\,.
  \end{array}
\right.
\end{equation}
Although the incident and reflected fields are generically unbounded as $\eta\to0$, the total scattering field is bounded (recall that $1+\zeta+\zeta^2=0$).


\smallskip
The projections of $v_\rtwo(\eta)$ onto the four modes are 
\begin{equation}\label{rtwoprojections}
  \renewcommand{\arraystretch}{1.3}
\left.
  \begin{array}{rl}
     \text{incident} &
     \renewcommand{\arraystretch}{1.3}
      \left\{
      \begin{array}{lcl}
         \Prp v_\rtwo &=& O(\eta^\third) \\
         \Pre v_\rtwo &=& \frac{\zeta}{C}\eta^\mthirds e_0 + O(\eta^\mthird)
      \end{array}
      \right.
      \\
     \text{reflected} &
     \renewcommand{\arraystretch}{1.3}
      \left\{
      \begin{array}{lcl}
         \Plp v_\rtwo &=& \frac{1}{C}\eta^\mthirds e_0 + O(\eta^\mthird) \\
         \Ple v_\rtwo &=& \frac{\zeta^2}{C}\eta^\mthirds e_0 + O(\eta^\mthird)
      \end{array}
      \right.
  \end{array}
\right.
\;
\text{(bounded left scattering fields near $\mathring v_\rtwo$)}.
\end{equation}
Observe that the incident field and the reflected field, although in $V_+(\eta)$ and $V_-(\eta)$, are always of strict order $\eta^\mthirds$ and thus blow up.  This agrees with the fact that the limit of the scattering field is not in the space $\Vo$.

\begin{theorem}[Nonresonant pathological scattering from the left]\label{thm:leftscatteringreg}
Assume that $D(\eta)\not=0$ for all sufficiently small $\eta$, including $\eta=0$ so that $\mathring D\not=0$.
\renewcommand{\labelenumi}{\arabic{enumi}.}
\renewcommand{\labelenumii}{\alph{enumii}.}
\begin{enumerate}
\item
($\eta\!=\!0$) The two-dimensional space $\To^{-1}[\Vo_+]$ is spanned by two distinguished vectors $\mathring v_\rone$ and $\mathring v_\rtwo$, characterized by the conditions
\begin{eqnarray}
  && \varepsilon^+(\mathring v_\rone)=1, \quad \varepsilon^2(\mathring v_\rone)=0,
      \quad \text{\itshape i.e.}\;\; \mathring v_\rone = e_+ + \varepsilon^0(\mathring v_\rone) e_0 + \varepsilon^1(\mathring v_\rone) e_1, \\
  && \varepsilon^+(\mathring v_\rtwo)=0, \quad \varepsilon^2(\mathring v_\rtwo)=1,
      \quad \text{\itshape i.e.}\;\; \mathring v_\rtwo = \varepsilon^0(\mathring v_\rtwo) e_0 + \varepsilon^1(\mathring v_\rtwo) e_1 + e_2,
\end{eqnarray}
and $\mathring v_\rone$ is also characterized by the conditions
\begin{equation}\label{vorone}
  \sspan\{\mathring v_\rone\} \,=\, \To^{-1}[\Vo_+] \cap \Vo \,,
  \qquad \varepsilon^+(\mathring v_\rone)=1.
\end{equation}
Denote by\, $\mathring\psi_\rone(z)$\, and \,$\mathring\psi_\rtwo(z)$\, the solutions to the Maxwell ODEs at $\eta=0$ in the presence of the defect slab with values at $z=0$ equal to $\mathring v_\rone$ and $\mathring v_\rtwo$:
\begin{eqnarray}
  && \frac{d}{dz} \mathring\psi_\rone(z) = iJ\mathring A(z) \mathring\psi_\rone(z),
        \quad \mathring\psi_\rone(0) = \mathring v_\rone, \\
  && \frac{d}{dz} \mathring\psi_\rtwo(z) = iJ\mathring A(z) \mathring\psi_\rtwo(z),
        \quad \mathring\psi_\rtwo(0) = \mathring v_\rtwo.
\end{eqnarray}
Both fields are bounded for $z>L$, and
\begin{enumerate}
\item $\mathring\psi_\rone(z)$ grows linearly for $z<0$ unless $\varepsilon^1(\mathring v_\rone)=0$, in which case it is bounded;
\item $\mathring\psi_\rtwo(z)$ grows quadratically for $z<0$.
\end{enumerate}
\item  The fields $\mathring\psi_\rone(z)$ and $\mathring\psi_\rtwo(z)$ are limits, locally uniform in $z$ as $\eta\to0$, of two fields $\psi_\rone(z)$ and $\psi_\rtwo(z)$ that are Puiseux series in $\eta$ and are solutions to problems of scattering from the left with the following properties:
\begin{enumerate}
\item The incident field for $\psi_\rone(z)$ is the rightward propagating mode
\begin{equation}
  \psi_\rone^\inc(z) \,=\, F(z;\eta)e^{ik_\rtp(\eta)z}\vrp(\eta),  
\end{equation}
and at the left edge of the slab, 
\begin{eqnarray}
  \psi_\rone(0) &=& \vrp + \rho_\lfl(\eta)\,\vll + \rho_\lfe(\eta)\,\vle \label{psirone1} \\
                       &=& \vrp + \eta^\mthird \left( \rho_\lfl(0) + O(\eta^\third) \right)\vlp + \eta^\mthird \left( -\rho_\lfl(0) + O(\eta^\third) \right)\vle, \label{psirone2}
\end{eqnarray}
in which the reflection coefficients \,$\rho_\lfl$ and $\rho_\lfe$\, and the terms $O(\eta^\third)$ are Puiseux series.
\item The incident field for $\psi_\rtwo(z)$ is the scaled rightward evanescent mode
\begin{equation}
   \psi_\rtwo^\inc(z) \,=\, \textstyle{\frac{\zeta}{C}}\, \eta^\mthirds\,F(z;\eta)e^{ik_\rte(\eta)z}\vre(\eta),  
\end{equation}
and at the left edge of the slab, 
\begin{equation}\label{label8465}
  \textstyle{\frac{C}{\zeta}}\, \psi_\rtwo(0) =
    \eta^\mthirds\vre + \big( \eta^\mthirds\zeta^2 + \eta^\mthird b + O(1) \big)\vlp + \big( \eta^\mthirds\zeta - \eta^\mthird b + O(1) \big)\vle,
\end{equation}
in which $b$ is a constant.
\end{enumerate}
\item For $0<\eta\ll1$, the scattering fields $\psi_\rone(z)$ and $\psi_\rtwo(z)$ span the two-dimensional space of fields produced by scattering from the left, that is, where the incident field is of the form
\begin{equation}
  \psi^\inc(z) = c\, F(z;\eta)e^{ik_\rtp z}\vrp + d\, F(z;\eta)e^{ik_\rtp z}\vre\,.
\end{equation}
%
\item  If for $0<\eta\ll1$, $\psi(z;\eta)$ is a left scattering field such that $\psi(0;\eta)$ is a Puiseux series (in $\eta^\third$) with limit in $\To^{-1}[\Vo_+]\cap \Vo$ as $\eta\to0$, then the incident field at $z=0$ has the form
\begin{equation}\label{label3987}
  \psi^\inc(0;\eta) = \alpha(\eta) v_\rtp(\eta) + \eta^\mthird \beta(\eta) v_\rte(\eta)\,,
\end{equation}
in which $\alpha$ and $\beta$ are Puiseux series.  Conversely, if $\psi^\inc(z;\eta)$ is an incident field whose value at $z=0$ is of the form (\ref{label3987}), then the corresponding scattering field $\psi(z;\eta)$ satifties $\psi(0;0)\in\To^{-1}[\Vo_+]\cap \Vo$.
\end{enumerate}
\end{theorem}

{\bfseries Nonresonant scattering from the right.}\;
Consider now a leftward field from the space $V_-(\eta)$, incident upon the slab from the right.  The resulting total field, evaluated at $z=0$ lies in the space $V_-(\eta)$, and, evaluated at $z=L$, lies in the space $T(\eta)[V_-(\eta)]$.  The asymptotic behavior of these scattering fields is captured by two distinguished scattering problems whose limits at $z=L$ are two distinguished vectors $\mathring v_\lone$ and $\mathring v_\ltwo$ that form a basis for the space $\To[\Vo_-]$.  The vector $\mathring v_\lone$ spans the 1D space $\To[\Vo_-]\cap\Vo$, and is the limit of the scattering field produced by the incident field that is a multiple of the leftward vector $v_\lfl$ at $z=L$.  The vector $\mathring v_\ltwo$ lies outside $\Vo$ and is the limiting value of the field produced by an incident field that blows up like $\eta^\mthirds$.

Consider, as before, Puiseux-series perturbations of $\mathring v_\lone$ and $\mathring v_\ltwo$ into the space $T(\eta)[V_-(\eta)]$ of scattering fields at $z\!=\!L$ and look at the projections onto their rightward (incoming) and leftward (reflected) components.  Denote these perturbations by
\begin{eqnarray}
  && v_\lone(\eta) = \mathring v_\lone + \eta^\third v_\lone^1 + O(\eta^\thirds), \label{lone} \\
  && v_\ltwo(\eta) = \mathring v_\ltwo + \eta^\third v_\ltwo^1 + O(\eta^\thirds). \label{ltwo}
\end{eqnarray}

The projections of $v_\lone(\eta)$ onto the four modes are
\begin{equation}\label{loneprojections}
  \renewcommand{\arraystretch}{1.3}
\left.
  \begin{array}{rl}
     \text{reflected} &
     \renewcommand{\arraystretch}{1.3}
      \left\{
      \begin{array}{lcl}
         \Prp v_\lone &=& [e_+,\mathring v_\lone]e_+ + \eta^\third [e_+,v_\lone^1]e_+ + O(\eta^\thirds) \\
         \Pre v_\lone &=& \eta^\mthird\left( \zeta^2 a + \zeta b \right)e_0 + O(1)
      \end{array}
      \right.
      \\
     \text{incident} &
     \renewcommand{\arraystretch}{1.3}
      \left\{
      \begin{array}{lcl}
         \Plp v_\lone &=& \eta^\mthird\left( a + b \right)e_0 + O(1) \\
         \Ple v_\lone &=& \eta^\mthird\left( \zeta a + \zeta^2 b \right)e_0 + O(1)
      \end{array}
      \right.
  \end{array}
\right.
\;
\hspace{-.75cm}\text{(bounded right scattering fields near $\mathring v_\lone$)},
\end{equation}
in which $a = -\frac{1}{C} [w_1,e_1]=1/(3k_1)$ and $b = \frac{1}{C} [e_0,v_\lone^1]$.
Again, the incident and reflected fields are generically unbounded as $\eta\to0$ but the total scattering field is bounded.

Because of the unique solution of the scattering problem, all four projections are determined by the incoming components $\Plp v_\lone$ and $\Ple v_\lone$ alone.  Both the incident and reflected field are bounded when $\zeta^2 a + \zeta b$ vanishes, and, generically in the $\eta\to0$ limit, the incident field contains the linearly growing mode to the right of the slab.

\smallskip
The projections of $v_\ltwo(\eta)$ onto the four modes are 
\begin{equation}\label{ltwoprojections}
  \renewcommand{\arraystretch}{1.3}
\left.
  \begin{array}{rl}
     \text{reflected} &
     \renewcommand{\arraystretch}{1.3}
      \left\{
      \begin{array}{lcl}
         \Prp v_\ltwo &=& [e_+,\mathring v_\ltwo]e_+ + \eta^\third[e_+,v_\ltwo^1]e_+ + O(\eta^\thirds) \\
         \Pre v_\ltwo &=& \frac{\zeta}{C}\eta^\mthirds e_0 + O(\eta^\mthird)
      \end{array}
      \right.
      \\
     \text{incident} &
     \renewcommand{\arraystretch}{1.3}
      \left\{
      \begin{array}{lcl}
         \Plp v_\ltwo &=& \frac{1}{C}\eta^\mthirds e_0 + O(\eta^\mthird) \\
         \Ple v_\ltwo &=& \frac{\zeta^2}{C}\eta^\mthirds e_0 + O(\eta^\mthird)
      \end{array}
      \right.
  \end{array}
\right.
\hspace{-.75cm}\text{(bounded right scattering fields near $\mathring v_\rtwo$)}.
\end{equation}
In this case again, the incident field and the reflected field, although in $V_-(\eta)$ and $V_+(\eta)$, are always of strict order $\eta^\mthirds$, and the limit of the scattering field is not in the space $\Vo=\Vo_-+\Vo_+$.

\begin{theorem}[Nonresonant pathological scattering from the right]\label{thm:rightscatteringreg}
Assume that $D(\eta)\not=0$ for all sufficiently small $\eta$, including $\eta=0$ so that $\mathring D\not=0$.
Items 2--4 hold under the conditions
\begin{equation}\label{generic0}
  [e_0,\To e_0]\not=0\,,
  \quad\text{\itshape i.e.,}\quad
  \To e_0 \not\in \Vo
\end{equation}
and
%
%
%
\begin{equation}\label{generic1}
  \mathring D_* \not=0,
\end{equation}
where $\mathring D_*$ is defined in~\ref{Doprime}.
\vspace{-2ex}
\renewcommand{\labelenumi}{\arabic{enumi}.}
\renewcommand{\labelenumii}{\alph{enumii}.}
\begin{enumerate}
\item
($\eta=0$) The two-dimensional space $\To[\Vo_-]$ is spanned by two distinguished vectors $\mathring v_\lone$ and $\mathring v_\ltwo$, characterized by the conditions
\begin{eqnarray}
  && \varepsilon^1(\mathring v_\lone)=1, \quad \varepsilon^2(\mathring v_\lone)=0,
      \quad \text{\itshape i.e.}\;\; \mathring v_\lone = \varepsilon^+(\mathring v_\lone)e_+ + \varepsilon^0(\mathring v_\lone) e_0 + e_1, \\
  && \varepsilon^1(\mathring v_\ltwo)=0, \quad \varepsilon^2(\mathring v_\ltwo)=1,
      \quad \text{\itshape i.e.}\;\; \mathring v_\ltwo = \varepsilon^+(\mathring v_\ltwo)e_+ + \varepsilon^0(\mathring v_\ltwo) e_0 + e_2,
\end{eqnarray}
and $\mathring v_\lone$ is also characterized by the conditions
\begin{equation}
  \sspan\{\mathring v_\lone\} \,=\, \To[\Vo_-] \cap \Vo \,,
  \qquad \varepsilon^1(\mathring v_\lone)=1.
\end{equation}
Denote by\, $\mathring\psi_\lone(z)$\, and \,$\mathring\psi_\ltwo(z)$\, the solutions to the Maxwell ODEs at $\eta\!=\!0$ in the presence of the defect slab with values at $z=L$ equal to $\mathring v_\lone$ and $\mathring v_\ltwo$:
\begin{eqnarray}
  && \frac{d}{dz} \mathring\psi_\lone(z) = iJ\mathring A(z) \mathring\psi_\lone(z),
        \quad \mathring\psi_\lone(L) = \mathring v_\lone, \label{IVPm1} \\
  && \frac{d}{dz} \mathring\psi_\ltwo(z) = iJ\mathring A(z) \mathring\psi_\ltwo(z),
        \quad \mathring\psi_\ltwo(L) = \mathring v_\ltwo. \label{IVPm2}
\end{eqnarray}
\begin{enumerate}
\item $\mathring\psi_\lone(z)$ grows linearly for $z>L$, and for $z<0$, it grows linearly unless $\varepsilon^1(\To^{-1}(\mathring v_\lone))=0$;
\item $\mathring\psi_\ltwo(z)$ grows quadratically for $z>L$, and for $z<0$, it grows linearly unless $\varepsilon^1(\To^{-1}(\mathring v_\ltwo))=0$.
\end{enumerate}
\item  The fields $\mathring\psi_\lone(z)$ and $\mathring\psi_\ltwo(z)$ are limits, locally uniform in $z$ as $\eta\to0$, of two fields $\psi_\lone(z)$ and $\psi_\ltwo(z)$ that are Puiseux series in $\eta$ and are solutions to problems of scattering from the left with the following properties:
\begin{enumerate}
\item The incident field for $\psi_\lone(z)$ is the scaled leftward mode
\begin{equation}
  \psi_\lone^\inc(z) \,=\, \eta^\mthird \frac{1}{k_1(1-\zeta^2)} \left( F(z-L)v_\lfp e^{ik_\lfp (z-L)} - F(z-L)v_\lfe e^{i k_\lfe (z-L)} \right), 
\end{equation}
and at the right edge of the slab,
\begin{equation}
  \psi_\lone(L) \,=\, {\textstyle\frac{1}{k_1(1-\zeta^2)}}\,\vll + r_\rtp(\eta)\,\vrp + r_\rte(\eta)\,\vre\,,
\end{equation}
in which the reflection coefficients \,$r_\rtp$\, and \,$r_\rte$\, are Puiseux series.
\item The incident field for $\psi_\ltwo(z)$ is
\begin{equation}
  \psi_\ltwo^\inc(z) \,=\, \textstyle{\frac{1}{C}}\, \eta^\mthirds \left( F(z-L)\vlp e^{ik_\lfp (z-L)} + \zeta^2 F(z-L)\vle e^{ik_\lfe (z-L)} \right),
\end{equation}
and at the right edge of the slab,
\begin{equation}
  \psi_\ltwo(L) \,=\,
   \textstyle{\frac{1}{C}} \left( \eta^\mthirds \vlp + \eta^\mthirds\zeta^2 \vle + \eta^\mthirds\zeta \vre \right) + v_+(\eta)\,,
\end{equation}
in which the vector $v_+(\eta)$ is a Puiseux series in the rightward space $V_+(\eta)$ with limit in $\Vo_+$.
\end{enumerate}
\item For $0<\eta\ll1$, the scattering fields $\psi_\lone(z)$ and $\psi_\ltwo(z)$ span the two-dimensional space of fields produced by scattering from the left, that is, where the incident field is of the form
\begin{equation}
  \psi^\inc(z) = c\, F(z-L)\vlp e^{ik_\lfp z} + d\, F(z-L)\vle e^{ik_\lfe z}\,.
\end{equation}
\item  If for $0<\eta\ll1$, $\psi(z;\eta)$ is a right scattering field such that $\psi(L;\eta)$ is a Puiseux series (in $\eta^\third$) with limit in $\To[\Vo_-]\cap \Vo$ as $\eta\to0$, then the incident field at $z=L$ has the form
\begin{equation}\label{label8325}
  \psi^\inc(L;\eta) = \eta^\mthird \big( \alpha(\eta) v_\lfp(\eta) + \beta(\eta) v_\lfe(\eta) \big) \,,
\end{equation}
in which $\alpha$ and $\beta$ are complex Puiseux series.
Conversely, if $\psi^\inc(z;\eta)$ is an incident field whose value at $z=L$ is of the form (\ref{label8325}), then the corresponding scattering field $\psi(z;\eta)$ satisfies $\psi(L;0)\in\To[\Vo_-]\cap \Vo$.
\end{enumerate}
\end{theorem}

\subsubsection{Resonant scattering}\label{resonant}

The resonant regime is characterized by the guided frozen mode that occurs when $\To e_0 = \ell e_0$ at $\eta=0$.  This mode is resonantly excited by the rightward propagating mode $F(z)e^{iKz}v_\rtp$ and by the leftward mode $F(z-L)e^{iK(z-L)}v_\lfl$.  
According to Theorem~\ref{thm:D}, the resonant case is also characterized by the non-generic algebraic condition that $\To[\Vo]=\Vo$, that is, the degenerate 3D span of the limiting rightward and leftward spaces (as $\eta\rightarrow 0$) transfers into itself across the slab.

Resonant excitation of the guided frozen mode $\psi^0(z)$ means that it is the limiting value of a scattering field as $\eta\to0$.  Such a scattering field is a (bounded) Puiseux series perturbation of $\psi^0(z)$, which, on either side of the slab is (a multiple of) $e_0$.  Thus, consider a field whose value at $z=0$ or $z=L$ is
\begin{equation}\label{resproj1}
  v_0(\eta) = e_0 + \eta^\third v_0^1 + \eta^\thirds v_0^2 + O(\eta)\,.
\end{equation}
The incoming and outgoing fields at the left and right edges of the slab are obtained by applying the projections onto the eigenmodes,
\begin{equation}\label{resproj2}
  \renewcommand{\arraystretch}{1.3}
\left.
  \begin{array}{rl}
     \text{rightward}\hspace*{-5pt} &
     \renewcommand{\arraystretch}{1.3}
      \left\{\hspace*{-3pt}
      \begin{array}{lcl}
         \Prp v_0 &\!=\!& \eta^\third[e_+,v_0^1]e_+ + \eta^\thirds [e_+,v_0^2]e_+ + O(\eta) \\
         \Pre v_0 &\!=\!& \frac{-1}{C} \left[ \eta^\mthird\zeta[e_0,v_0^1]e_0
                               + \left( [w_2,e_0]+\zeta^2[w_1,v_0^1]+\zeta[e_0,v_0^2] \right)e_0 + \zeta^2[e_0,v_0^1]w_1 + O(\eta^\third) \right]
      \end{array}
      \right.
      \\
     \text{leftward}\hspace*{-5pt} &
     \renewcommand{\arraystretch}{1.3}
      \left\{\hspace*{-3pt}
      \begin{array}{lcl}
         \Plp v_0 &\!=\!& \frac{-1}{C} \left[ \eta^\mthird[e_0,v_0^1]e_0
                               + \left( [w_2,e_0]+[w_1,v_0^1]+[e_0,v_0^2] \right)e_0 + [e_0,v_0^1]w_1 + O(\eta^\third) \right] \\
         \Ple v_0 &\!=\!& \frac{-1}{C} \left[ \eta^\mthird\zeta^2[e_0,v_0^1]e_0
                               + \left( [w_2,e_0]+\zeta[w_1,v_0^1]+\zeta^2[e_0,v_0^2] \right)e_0 + \zeta[e_0,v_0^1]w_1 + O(\eta^\third) \right]
      \end{array}
      \right.
  \end{array}
\right.
\end{equation}
Since $1+\zeta+\zeta^2=0$, the $O(\eta^\mthird)$ terms cancel upon summing the projections, as do all the terms comprising the $O(1)$ part except for $[w_2,e_0]=\overline{\varepsilon^2(w_2)}$, which is nonzero by (\ref{LowOrderExpanEigenvecSpanRel}).

It turns out that, if $v_0(\eta)$ is the solution of a scattering problem, then necessarily
\begin{equation}
  \varepsilon^2(v_0^1) = -[e_0,v_0^1] = 0,
\end{equation}
so that the $\eta^\mthird$ terms vanish and {\em the projections onto all modes are bounded}.

\medskip


\begin{theorem}[Resonant pathological scattering from the left]\label{thm:leftscatteringres}
Suppose that $\To e_0 = \ell e_0$ for some nonzero $\ell\in\CC$, or, equivalently, $\mathring D=0$.  Items 2--4 hold under the condition
\begin{equation}\label{generic2}
[e_+, \To e_1] \not= 0.
\end{equation}
This condition is equivalent in the resonant case $\mathring D = 0$ to any one of the following conditions: $[e_+, \To^{-1} e_1] \not= 0$, $[e_+, \To w_1][e_+, \To^{-1} w_1]\not= 0$, $\To[\mathring V_+]\not= \mathring V_+$, or $\To[\mathring V_-]\not= \mathring V_-$.

\renewcommand{\labelenumi}{\arabic{enumi}.}
\renewcommand{\labelenumii}{\alph{enumii}.}
\begin{enumerate}
  \item The two-dimensional space $\To^{-1}[\Vo_+]\subset\Vo$ is spanned by $e_0$ and a vector $\mathring w_+$  characterized by
\begin{equation}
  \varepsilon^+(\mathring w_+)=1,
  \quad
  \varepsilon^0(\mathring w_+)=0,
  \quad
  \text{\itshape i.e.}\;\; \mathring w_+ = e_+ + \varepsilon^1(\mathring w_+) e_1\,.
\end{equation}
Denote by $\mathring\phi_\rzero(z)$ and $\mathring\phi_+(z)$ the solution to the Maxwell ODEs in the presence of the defect slab with values at $z=0$ equal to $e_0$ and $\mathring w_+$:
\begin{eqnarray}
  && \frac{d}{dz} \mathring\phi_\rzero(z) = iJ\mathring A(z) \mathring\phi_\rzero(z),
        \quad \mathring\phi_\rzero(0) = e_0\,, \\
  && \frac{d}{dz} \mathring\phi_+(z) = iJ\mathring A(z) \mathring\phi_+(z),
        \quad \mathring\phi_+(0) = \mathring w_+\,.
\end{eqnarray}
\begin{enumerate}
\item $\mathring\phi_\rzero(z)$ is (a nonzero multiple of) the guided frozen mode;
\item $\mathring\phi_+(z)$ is bounded for all $z>L$ and grows linearly for $z<0$. In particular, $\varepsilon^1(\mathring w_+)\not=0$.
\end{enumerate}
\item The guided frozen mode $\mathring\phi_\rzero(z)$ is the limit, locally uniform in $z$ as $\eta\to0$, of two fields $\phi_\rone(z)$ and $\phi_\rtwo(z)$ that are Puiseux series in $\eta$ and are solutions of problems of scattering from the left with the following properties:
\begin{enumerate}
  \item ({\bfseries Resonant excitation of the guided frozen mode.}) The incident field for $\phi_\rone(z)$ is the scaled rightward propagation mode 
\begin{equation}
  \phi_\rone^\inc = \eta^\third \zeta k_1[\To e_+, e_1] F(z)\vrp(\eta) e^{ik_\rtp(\eta) z},
\end{equation}
and at the left edge of the slab,
\begin{equation}
  \phi_\rone(0) = \eta^\third \ell\zeta k_1[\To e_+, e_1]\vrp
            + a(\eta) \vlp + b(\eta) \vle\,,
\end{equation}
in which $a(\eta)$ and $b(\eta)$ are Puiseux series with $a(0)+b(0)=1$.  (The excitation is resonant because the incident field vanishes as $\eta\to0$.)
\item The incident field for $\phi_\rtwo(z)$ is the rightward evanescent mode
\begin{equation}
   \phi_\rtwo^\inc = c\,F(z)\vre e^{ik_\rte(\eta) z},
   \qquad
   c = -\frac{\,\ell\, k_1^2\,[\To e_+,e_1][\To e_1, e_+]}{\zeta[\To e_+,e_+]}\,,
\end{equation}
and at the left of the slab,
\begin{equation}
  \phi_\rtwo(0) = c\,\vre
                     + a(\eta) \vlp
                     + b(\eta)\vle\,.
\end{equation}
in which $a(\eta)$ and $b(\eta)$ are Puiseux series.
\end{enumerate}
\item The field $\mathring\phi_+(z)+c\,\mathring\phi_\rzero(z)$, for some constant $c$, is the limit, locally uniform in $z$ as $\eta\to0$, of the singular linear combination
\begin{equation}
  \phi_+(z) \,=\, \eta^\mthird \frac{1}{\,\ell\zeta k_1[\To e_+, e_1]} \left( \phi_\rone(z) - \phi_\rtwo(z) \right)\,,
\end{equation}
which is the scattering field resulting from the incident field
\begin{equation}
  \phi_+^\inc(z) \,=\, F(z)v_\rtp e^{ik_\rtp z} \,+\, \eta^\mthird \zeta k_1\frac{[\To e_1,e_+]}{[\To e_+, e_+]} F(z)v_\rte e^{ik_\rte z}\,.
\end{equation}
\item For $0<\eta\ll1$, any two of the fields $\phi_+(z)$, $\phi_\rone(z)$, $\phi_\rtwo(z)$ span the two-dimensional space of fields produced by scattering from the left, where the incident field is of the form 
\begin{equation}
  \psi^\inc(z) = c\, F(z)\vrp e^{ik_\rtp z} + d\, F(z)\vre e^{ik_\rte z}\,.
\end{equation}
\end{enumerate}
\end{theorem}


\begin{theorem}[Resonant pathological scattering from the right]\label{thm:rightscatteringres}
Assume the resonant condition, that $\To e_0 = \ell e_0$ for some $\ell\in\CC$, or, equivalently, $\mathring D=0$.  Items 2--4 hold under the condition
\begin{equation}\label{delta}
  \delta \,:=\, \ell - \left( |\ell|^2+1 \right) [e_1,\To e_1]  \,\not=\,  0\,.
\end{equation}
\vspace{-2ex}
\renewcommand{\labelenumi}{\arabic{enumi}.}
\renewcommand{\labelenumii}{\alph{enumii}.}
\begin{enumerate}
  \item The two-dimensional space $\To[\Vo_-]\subset\Vo$ is spanned by $e_0$ and a vector $\mathring w_-$ characterized by
\begin{equation}
  \varepsilon^1(\mathring w_-)=1,
  \quad
  \varepsilon^0(\mathring w_-)=0,
  \quad
  \text{\itshape i.e.}\;\; \mathring w_- = \varepsilon^+(\mathring w_-) e_+ + e_1\,.
\end{equation}
Denote by $\mathring\phi_\lzero(z)$ and $\mathring\phi_-(z)$ the solution to the Maxwell ODEs in the presence of the defect slab with values at $z=L$ equal to $e_0$ and $\mathring w_+$:
\begin{eqnarray}
  && \frac{d}{dz} \mathring\phi_\lzero(z) = iJ\mathring A(z) \mathring\phi_\lzero(z),
        \quad \mathring\phi_\lzero(L) = e_0\,, \\
  && \frac{d}{dz} \mathring\phi_-(z) = iJ\mathring A(z) \mathring\phi_-(z),
        \quad \mathring\phi_-(L) = \mathring w_-\,.
\end{eqnarray}
\begin{enumerate}
\item $\mathring\phi_\lzero(z)$ is (a nonzero multiple of) the guided frozen mode;
\item $\mathring\phi_-(z)$ grows linearly for $z<0$ and for $z>L$.
\end{enumerate}
\item The guided frozen mode $\mathring\phi_\lzero(z)$ is the limit, locally uniform in $z$ as $\eta\to0$, of a field $\phi_\lone(z)$ that is a Puiseux series in $\eta$ and satisfies the problem of scattering from the left with incident field equal to the scaled leftward mode
\begin{equation*}
  \phi^\inc_\lone(z) \,=\, c\, \left( F(z-L)v_\lfp e^{ik_\lfp(z-L)} - F(z-L)v_\lfe e^{ik_\lfe(z-L)} \right)
\end{equation*}
with value at the right edge of the slab
\begin{equation}
  \phi^\inc_\lone(L) \,=\, c\,\eta^\third\vll\,,
\end{equation}
in which
$c = \bar\ell/[3\ell\,k_1^4(1-\zeta^2)(2+|\ell|^2)]^{-1}$.
\item The field $\mathring\phi_-(z) + c\mathring\phi_\lzero(z)$, for some constant $c$, is the limit, locally uniform in $z$ as $\eta\to0$, of a field $\phi_-(z)$ that is a Puiseux series in~$\eta$ and satisfies the problem of scattering from the left with incident field that blows up as $\eta\to0$,
\begin{equation}
  \phi_-^\inc(z) \,=\, \eta^\mthird \left( j_\lfp F(z-L)v_\lfp e^{ik_\lfp (z-L)} + j_\lfe F(z-L)v_\lfe e^{ik_\lfe (z-L)} \right) \,,
\end{equation}
and whose value at $z=L$ is equal to
\begin{equation*}
  \phi_-^\inc(L) \,=\, \eta^\mthird \left( j_\lfp v_\lfp + j_\lfe v_\lfe \right)
                  \,=\, j_\lfp v_\lfl + \eta^\mthird (j_\lfp+j_\lfe)v_\lfe\,.
\end{equation*}
in which
\begin{equation}\label{jpje}
  \col{1.2}{j_\lfp}{j_\lfe} \,=\, \frac{\bar\ell}{9k_1^5(\zeta-1)\,\delta}
                                            \col{1.2}{2\zeta\hat\ell - c - \zeta^2\ell}{-2\zeta^2\hat\ell + c + \zeta\ell}\,.
\end{equation}
The vector $[j_\lfp\;\;j_\lfe]$ is not a multiple of $[1\;-\!1]$, that is, $j_\lfp+j_\lfe\not=0$\,.

\item For $0<\eta\ll1$, the fields $\phi_\lone(z)$ and $\phi_-(z)$ span the two-dimensional space of fields produced by scattering from the right, where the incident field is of the form 
\begin{equation}
  \psi^\inc(z) = c\, F(z-L)\vlp e^{ik_\lfp z} + d\, F(z-L)\vle e^{ik_\lfe z}\,.
\end{equation}
\end{enumerate}
\end{theorem}

\subsubsection{Proofs of the theorems}

Proofs of these theorems involve finding the solution of a system of the form
\begin{equation}
  M(\eta)t(\eta) = b(\eta) 
\end{equation}
as a Puiseux series $t(\eta)$ (see equations (\ref{leftscattering}) and (\ref{rightscattering})).  In the case of nonresonant scattering from the right and both resonant cases, $M(0)$ is noninvertible and the analysis is delicate.  The following lemma treats the solution of this problem under the conditions required by the theorems.

\begin{lemma}\label{lemma:Mtb}
Consider the equation
\begin{equation}\label{Mtb}
  M(z) t(z) = b(z)\,,
\end{equation}
in which $M(z)$ is a $2\times2$ matrix function and $b(z)$ is a column-vector function with $b(0)=0$, both given by power series
\begin{equation}
  M(z) = \sum_{n=0}^\infty M_n z^n\,,
  \qquad
  b(z) = \sum_{n=1}^\infty b_nz^n\,,
\end{equation}
with nonzero radii of convergence.
Suppose that $M_0$ is nonzero and noninvertible with
\begin{eqnarray}
  \rm{Null}(M_0) &=& \sspan\{s\}\,, \\
  \rm{Ran}(M_0) &=& \sspan\{r\}\,.
\end{eqnarray}
For integers $n\geq0$ define the matrix
\begin{equation}
  \mathcal{M}_0 = \mat{1.1}{s^TM_{2}s}{s^TM_{1}r}{r^TM_{1}s}{r^TM_{0}r}\,.
\end{equation}
Suppose that
\renewcommand{\labelenumi}{\alph{enumi}.}
\begin{enumerate}
  \item $s^Tr=0$,
  \item $s^T b_1=0$, {\itshape i.e.}, $b_1\in\mathrm{Ran}(M_0)$,
  \item $s^T M_1\,s = 0$,
  \item $\mathcal{M}_0$ is invertible.
\end{enumerate}
The equation (\ref{Mtb}) admits a power series solution $t(z) = \sum_{n=0}^\infty t_nz^n$ with a nonzero radius of convergence and such that
\begin{equation}\label{label9965}
  t_0 \in\sspan\{s\}\,.
\end{equation}
If either $b_1\not=0$ or $s^T b_2\not=0$, then 
\begin{equation}\label{label3524}
  t_0=0 \implies t_1\not\in\sspan\{s\}\,.  
\end{equation}
In addition, $t_0\not=0$ in each of the following situations:
\renewcommand{\labelenumi}{\roman{enumi}.}
\begin{enumerate}
\item if $s^T M_2 s = 0$, then
\begin{equation}
   M(z)t(z) = z\,r
   \quad \implies \quad
   t_0 = \frac{r^Tr}{r^TM_1s}\, s
\end{equation}
or, alternatively,
\item if $s^T M_2 s = 0$ and $r^T M_0 r \not=0$ then
\begin{equation}
   M(z)t(z) = z^2(s + \alpha r)
   \quad \implies \quad
   t_0 = -\frac{\big( s^Ts\big) \big(r^TM_0r\big)}{\big(s^TM_1r\big) \big(r^TM_1s\big)}\, s \,.
\end{equation}
\end{enumerate}
\end{lemma}

\begin{proof}
The equation $M(z) t(z) = b(z)$ is equivalent to a system of equations for the coefficients,
\begin{equation}\label{coeffeqn}
  \sum_{\ell=0}^n M_\ell\, t_{n-\ell} = b_n
  \qquad
  \text{for }\, n\geq0\,,
\end{equation}
in which $b_0=0$.
We will first prove that this system admits a unique solution $\{t_n\}_{n=0}^\infty$ and then prove that the power series $\sum_{n=0}^\infty t_nz^n$ has a nonzero radius of convergence.  Set
\begin{equation}
  t_n = c_n s + d_n r\,.
\end{equation}
For $n=0$, (\ref{coeffeqn}) is just $M_0t_0 = 0$, which implies that $d_0=0$, or $t_0=c_0s$.
For $n=1$, the equation
\begin{equation}
  M_0t_1 + M_1t_0 = b_1
\end{equation}
has a solution $t_1$ if and only if $s^TM_1t_0=s^Tb_1$, which, using $t_0=c_0s$ becomes
\begin{equation}
  c_0\, s^T M_1 s = s^T b_1\,.
\end{equation}
Both sides of this equation vanish by assumptions (b) and (c) above.  For $n\geq0$, there are two linear relations between the coefficients $c_n$ and $d_{n+1}$.  First, the solvability condition for $t_{n+2}$ in (\ref{coeffeqn}) is
\begin{equation}\label{one}
  \sum_{\ell=1}^{n+2}\left( c_{n+2-\ell}\, s^T M_\ell\, s + d_{n+2-\ell}\, s^T M_\ell\, r \right) = s^T b_{n+2}
  \qquad (n\geq0)\,.
\end{equation}
Second, multiplying (\ref{coeffeqn}) on the left by $r^T$ with $n$ replaced by $n+1$ gives
\begin{equation}\label{two}
  \sum_{\ell=0}^{n+1} \left( c_{n+1-\ell}\, r^T M_\ell\, s + d_{n+1-\ell}\, r^T M_\ell\, r \right) = r^T b_{n+1}
  \qquad (n\geq0)\,.
\end{equation}
Rewrite (\ref{one}) by using $d_0=0$ and $s^T M_1\,s=0$ and shifting the index of the $d$-terms by $1$ and the index of the $c$-terms by $2$ to obtain
\begin{equation}\label{oneprime}
  \sum_{\ell=0}^n \left( c_{n-\ell}\, s^T M_{\ell+2}\, s + d_{n+1-\ell}\, s^T M_{\ell+1}\, r \right) = s^T b_{n+2}
  \qquad (n\geq0)\,.
\end{equation}
Rewrite (\ref{two}) by using $d_0=0$ and $M_0s=0$ and shifting the index of the $c$-terms by $1$:
\begin{equation}\label{twoprime}
  \sum_{\ell=0}^n \left( c_{n-\ell}\, r^T M_{\ell+1}\, s + d_{n+1-\ell}\, r^T M_\ell\, r \right) = r^T b_{n+1}
  \qquad (n\geq0)\,.
\end{equation}
Equations (\ref{oneprime}) and (\ref{twoprime}) together yield the system
\begin{equation}\label{Mvb}
  \sum_{\ell=0}^n \mathcal{M}_\ell \col{1.1}{c_{n-\ell}}{d_{n+1-\ell}} = \col{1.1}{s^T b_{n+2}}{r^T b_{n+1}}
  \qquad (n\geq0)\,,
\end{equation}
in which
\begin{equation}
  \mathcal{M}_n = \mat{1.1}{s^TM_{n+2}s}{s^TM_{n+1}r}{r^TM_{n+1}s}{r^TM_{n}r}\,.
\end{equation}
Set $v_n=[c_n\;d_{n+1}]^T$ and $w_n=[s^T b_{n+2},r^T b_{n+1}]$, and define the power series
\begin{equation}
  \mathcal{M}(z) = \sum_{n=0}^\infty \mathcal{M}_n z^n\,,
  \qquad
  v(z) = \sum_{n=0}^\infty v_n z^n\,,
  \qquad
  w(z) = \sum_{n=0}^\infty w_n z^n\,.
\end{equation}
Equation (\ref{Mvb}) is equivalent to the formal equation
\begin{equation}
  \mathcal{M}(z) v(z) = w(z)\,.
\end{equation}
The series for $\mathcal{M}(z)$ has a nonzero radius of convergence because the series for $M(z)$ does, and since $\mathcal{M}_0$ is invertible by assumption (d), $\mathcal{M}^{-1}(z)$ exists and is analytic in a neighborhood of $z=0$.  Thus the coefficients $v_n$ are given by those of a convergent power series,
\begin{equation}
  v(z) = \mathcal{M}(z)^{-1} w(z)
\end{equation}
which proves that the power series $\sum_{n=0}^\infty c_nz^n$ and $\sum_{n=0}^\infty d_nz^n$ have nonzero radii of convergence and therefore so does $t(z)=\sum_{n=0}^\infty t_nz^n$ since $t_n=c_ns+d_nr$.

For $n=0$, we have
\begin{equation}
  \mathcal{M}_0 \col{1.1}{c_0}{d_1} = \col{1.1}{s^T b_2}{r^T b_1}
  = \renewcommand{\arraystretch}{1.1}
\left\{
  \begin{array}{ll}
    \col{1.1}{0}{r^T r} & \text{(case (i))} \\
    \vspace{-1.5ex} \\
    \col{1.1}{s^T s}{0} & \text{(case (ii))}\,.
  \end{array}
\right.
\end{equation}
Whenever either $b_1\not=0$ or $s^T b_2\not=0$, the vector $[c_0\;d_1]^T$ is nonzero (recall $b_1\in\sspan\{r\}$ by assumption (b)). 
This means that either $t_0=c_0s\not=0$ or $t_1=c_1s+d_1r$ is not in $\sspan\{s\}$.  If the $11$-entry of $\mathcal{M}_0$ vanishes, that is, $s^T M_2\,s=0$, then in case (i) one can solve for $c_0\not=0$ and obtain $t_0$ as in the theorem.  If, in addition, the $22$-entry of $\mathcal{M}_0$ is nonzero, that is, $r^T M_0\,r\not=0$, then in case (ii), again one can solve for $c_0\not=0$. 
\end{proof}

\begin{proof}[Proofs of Theorems~\ref{thm:leftscatteringreg}--\ref{thm:rightscatteringres}]
To produce the vectors $\mathring v_\rone$ and $\mathring v_\rtwo$ in part\,1 of Theorem~\ref{thm:leftscatteringreg}, observe that the nonresonant condition implies, by Theorem~\ref{thm:D}, that $\To^{-1}[\Vo_+]\cap\Vo$ is one-dimensional; let it be spanned by $v=\To^{-1}(w)$, with $0\not=w\in\Vo_+$.  By the characterizations (\ref{Voplus},\ref{Vo}), $v=ae_++be_0+ce_1$ and $w=\alpha e_++\beta e_0$ for some complex constants.  The interaction matrix in (\ref{jordanform}) for the flux form and the flux-unitarity of $\To$ yields
\begin{equation}
  |a|^2-|c|^2 = [v,v] = [\To^{-1}(w),\To^{-1}(w)] = [w,w] = |\alpha|^2.
\end{equation}
If $a=0$, then also $c=\alpha=0$ and one obtains
\begin{equation}
  b\,\To(e_0) = \To(v) = w = \beta e_0,
\end{equation}
which, by Theorem~\ref{thm:D}(b) holds only in the resonant case.  Thus $a\not=0$ and one can set $\mathring v_\rone=a^{-1}v$ and obtain $\varepsilon^{+}(\mathring v_\rone)=1$.  Since $\Vo=\sspan\{e_+,e_0,e_1\}$, one obtains $\varepsilon^2(\mathring v_\rone)=0$.  Because $\To^{-1}[\Vo_+]\cap\Vo\not=\Vo$, there is a vector $\mathring v_\rtwo$ in $\To^{-1}[\Vo_+]$ with a nonzero $e_2$ component with respect to the basis $\{e_+,e_0,e_1,e_2\}$, which can be taken to be $e_2$ itself, and one can arrange the $e_+$ component to vanish by adding a suitable multiple of~$\mathring v_\rone$.

Proof of the existence of the vectors $\mathring v_\lone$ and $\mathring v_\ltwo$ in part~1 of Theorem~\ref{thm:rightscatteringreg} is almost identical, only that now $w=\alpha e_1+\beta e_0\in\Vo_-$ and $[w,w]=-|\alpha|^2$.  This leads to $c\not=0$, and the rest follows analogously.

\smallskip
{\itshape Proof of Theorem~\ref{thm:leftscatteringreg}} (nonresonant scattering from the left).
In Theorem~\ref{thm:leftscatteringreg}, the fields $\mathring\psi_\rone$ and $\mathring\psi_\rtwo$ are given to the right of the slab by an expression of the form
\begin{eqnarray}
  \psi(z) = F(z-L) e^{i\mathring K(z-L)} v
  \qquad (z>L)
\end{eqnarray}
for some vector $v\in\Vo_+$.  Since $\Vo_+$ is spanned by two eigenvectors, $e_+$ and $e_0$ of $\mathring K$ with real eigenvalues (see~\ref{jordanform}), $\psi(z)$ is bounded for $z>L$.  To the left of the slab, the fields have the form
\begin{eqnarray}
  \psi(z) = F(z) e^{i\mathring K z} v
  \qquad (z<0).
\end{eqnarray}
For $\psi=\mathring\psi_\rone$, $v=\mathring v_\rone=ae_++be_0+ce_1$, and because $e_1$ is a the second vector in a Jordan chain for $\mathring K$, $\mathring\psi_\rone$ experiences linear growth unless $c=\varepsilon^1(v)=0$ (see \ref{jordanpropagation}).
For $\psi=\mathring\psi_\rtwo$, $v=\mathring v_\rtwo=be_0+ce_1+de_2$, with $d\not=0$, and since $e_2$ is the third vector in a Jordan chain for $\mathring K$, $\mathring\psi_\rtwo$ experiences quadratic growth. 

To prove part 2a, put $j_\rtp=1$ and $j_\rte=0$ in the problem of scattering from the left (\ref{leftscattering}).  The total scattering field at $z=0$ is equal to
\begin{equation}\label{label7374}
   \psi_\rone(0) \,=\, v_\rtp + \rho_\lfl v_\lfl + \rho_\lfe v_\lfe \,=\, T^{-1}\left( t_\rtp v_\rtp + t_\rte v_\rte \right).
\end{equation}
The first equation of (\ref{leftscattering}) gives $t_\rtp$ and $t_\rte$ as Puiseux series (having a limit as $\eta\to0$), since the matrix is invertible at $\eta=0$ by the nonresonance condition.  The second equation of (\ref{leftscattering}) shows that $r_\lfp$ and $r_\lfe$ are Puiseux-Laurent series, and therefore so are $\rho_\lfl$ and $\rho_\lfe$.  The conservation law (\ref{leftconservation}) reduces to
\begin{equation}
  |t_\rtp|^2 + C|\rho_\lfl|^2 = 1\,,
\end{equation}
and, since $C>0$ (see \ref{UniqueConstantInNormalizedFluxFormPerturbed}), $\rho_\lfl$ is bounded.  We have established that right-hand side of (\ref{label7374}) as well as the first two terms of the left-hand side have limits at $\eta=0$, and therefore so does $\rho_\lfe$.  This proves the first equality in item (2a), and (\ref{coordinatechange}) establishes the second equality.
Now taking the limit as $\eta\to0$, the left-hand side of (\ref{label7374}) becomes
\begin{equation}
  \lim_{\eta\to0}\psi_\rone(0) \,=\, e_++\rho_\lfl(0) (1-\zeta^2)w_1 +\rho_\lfe(0) e_0 \,=\, \mathring v_\rone \,=\, \mathring\psi_\rone(0)\,.
\end{equation}
The middle equality follows from the characterizing condition (\ref{vorone}) for $\mathring v_\rone$.
The limit $\psi_\rone(z)\to\mathring\psi_\rone(z)$ that is locally uniform in $z$ is inherited from that of the transfer matrix $T(0,z;\kappa(\eta),\omega(\eta))$ (see \cite[Appendix]{ShipmanWelters2013}):
\begin{equation}
  \psi_\rone(z;\eta) \,=\, T(0,z;\kappa(\eta),\omega(\eta))\psi_\rone(0;\eta)
  \,\xrightarrow{\eta\to0}\, 
 T(0,z;\kappa(0),\omega(0))\mathring\psi_\rone(0) \,=\, \mathring\psi_\rone(z)\,.
\end{equation}

To prove (2b) of Theorem~\ref{thm:leftscatteringreg}, put $j_\rtp=0$ and $j_\rte=\eta^\mthirds$.  The total scattering field at $z=0$ is equal to
\begin{equation}\label{label9392}
   \eta^\mthirds v_\rte + r_\lfp v_\lfp + r_\lfe v_\lfe = T^{-1}\left( t_\rtp v_\rtp + t_\rte v_\rte \right).
\end{equation}
The first equation of (\ref{leftscattering}) shows that $t_\rtp$ and $t_\rte$ are (bounded) Puiseux series, and thus the scattering field (\ref{label9392}), is bounded as $\eta\to0$.  The projections $P_\lfp$ and $P_\lfe$ (see \ref{Plp} and \ref{Ple}) show that $r_\lfp$ and $r_\lfe$ are $O(\eta^\mthirds)$; set
\begin{eqnarray}
  r_\lfp &=& \eta^\mthirds\left( r_\lfp^0 + \eta^\third r_\lfp^1 + O(\eta^\thirds) \right)\,, \\
  r_\lfe &=& \eta^\mthirds\left( r_\lfe^0 + \eta^\third r_\lfe^1 + O(\eta^\thirds) \right)\,.
\end{eqnarray}
By forcing the left-hand side of (\ref{label9392}) to be $O(1)$ and using the expansions (\ref{eigenvectors}) of the eigenvectors $v_{\rte,\lfp,\lfe}$, one obtains the two equations
\begin{eqnarray}
  && \left( 1 + r^0_\lfp + r_\lfe^0 \right)e_0 \,=\,0\,,\\
  && \left( \zeta + r^0_\lfp + \zeta^2 r_\lfe^0 \right)w_1 + \left( r^1_\lfp + r_\lfe^1 \right)e_0 \,=\, 0\,.
\end{eqnarray}
Because $\varepsilon^1(w_1)=k_1\not=0$, $e_0$ and $w_1$ are independent, and one obtains
$r^0_\lfp=\zeta^2$, $r^0_\lfe=\zeta$, and $r^1_\lfp = -r_\lfe^1$, and thus the total field at $z=0$ is given by the right-hand side of (\ref{label8465}).  Since it is equal to the limit $\To^{-1}\left( \mathring t_\rtp v_\rtp + \mathring t_\rte v_\rte \right)$ as $\eta\to0$ in (\ref{label9392}), it lies in $\To^{-1}[\Vo_+]$.

If one now takes $j_\rte = \eta^\mthirds(-\zeta/C)$ (and retains $j_\rtp=0$), the resulting scattering field is $\psi_\rtwo(z)$, as defined in the theorem, and one obtains (\ref{label8465}).  Since $v_\rte=e_0+O(\eta^\third)$, one has
\begin{equation}
  \Pre\psi_\rtwo(0) \,=\, \eta^\mthirds(-\zeta/C)v_\rte \,=\, \eta^\mthirds(-\zeta/C)e_0 + O(\eta^\mthird)\,,
\end{equation}
and comparing with the expression (\ref{Pre}) for $\Pre$ yields
\begin{equation}
  \varepsilon^2\left(\lim_{\eta\to0}\psi_\rtwo(0)\right) \,=\, -[e_0,\mathring\psi_\rtwo(0)] \,=\, 1\,.
\end{equation}
Also, $e_+$ enters the projections $P_{\lfp,\rte,\lfe}$ only at $O(\eta^\third)$, and $\Prp\psi_\rtwo(0)=j_\rtp v_\rtp=0$, so that, according to expression (\ref{Prp}) for $\Prp$,
\begin{equation}
  \varepsilon^+\left(\lim_{\eta\to0}\psi_\rtwo(0)\right) \,=\, [e_+,\mathring\psi_\rtwo(0)]\,=\, 0\,.
\end{equation}
Thus, by the defining properties of $\mathring v_\rtwo$, one has $\lim_{\eta\to0}\psi_\rtwo(0)=\mathring v_\rtwo=\mathring\psi_\rtwo(0)$.
Again, the locally uniform in $z$ limit $\psi_\rtwo(z)\to\mathring\psi_\rtwo(z)$ is inherited from from that of the transfer matrix $T(0,z;\kappa(\eta),\omega(\eta))$ and the limit $\psi_\rtwo(0)\to\mathring\psi_\rtwo(0)$ as $\eta\to0$.

Part\,3 of Theorem~\ref{thm:leftscatteringreg} follows from the independence of $\mathring\psi_\rone(0)$ and $\mathring\psi_\rtwo(0)$ and the continuity of $\psi_\rone(0)$ and $\psi_\rtwo(0)$ as functions of $\eta$ at $\eta=0$.

To prove part\,4, assume first that $\mathring\psi(0):=\lim_{\eta\to0}\psi(0;\eta)\in\To^{-1}[\Vo_+]\cap\Vo$, so that $\mathring\psi(0)$ is a multiple of~$\mathring v_\rone$.  Thus $\psi(0;\eta)$ can be identified with $v_\rone(\eta)$ of equation (\ref{rone}), and the projections (\ref{roneprojections}) onto the eigenmodes establish the form of the incident field in part\,4.  Conversely, assume the given form of the incident field.  Equation (\ref{leftscattering}) yields
\begin{equation}
   \mat{1.2}{\left[T\vrp,\vrp\right]}{\left[T\vrp,\vre\right]}{\left[T\vle,\vrp\right]}{\left[T\vle,\vre\right]}
\col{1.2}{t_\rtp}{t_\rte}
= \col{1.2}{O(1)}{O(\eta^\third)}\,.
\end{equation}
By the nonresonant condition, the determinant of this matrix is nonzero for $\eta$ sufficiently small, and thus the coefficients $t_\rtp$ and $t_\rte$ and therefore also the total scattering field at $z=L$ admit Puiseux series.  Thus $\psi(0;\eta)=T(\eta)^{-1}\psi(L;\eta)$ admits a Puiseux series and is a combination of $v_\rone(\eta)$ and $v_\rtwo(\eta)$ as in (\ref{rcombination}).   If $\psi(0;0)$ has a nonzero $e_2$ component, then (\ref{roneprojections},\ref{rtwoprojections}) show that the projection onto $v_\rte$ has a nonzero $\eta^\mthirds$-term.  But this contradicts the assumption that the incident field is $O(\eta^\mthird)$.  Thus $\psi(0;0)$ has no $e_2$ component and therefore lies in $\To^{-1}[\Vo_+]\cap\Vo$.

\smallskip
{\itshape Proof of Theorem~\ref{thm:rightscatteringreg}} (nonresonant scattering from the right).
The behavior of the functions $\mathring\psi_\lone(z)$ and $\mathring\psi_\ltwo(z)$ in part\,1(a,b) can be seen from the elementary matrix solution (\ref{jordanpropagation}) at $\eta=0$. 
To prove part\,2, consider first the incident field stipulated in part\,2a, which is equal to a constant times $v_\lfl$ at $z=L$.  The solution of the scattering problem exists and is unique for $\eta$ sufficiently small so that $D(\eta)\not=0$; call it $\psi(z)$.  We must show that $\psi$ is equal to $\psi_{-1}$ by showing that it satisfies the same initial-value problem (\ref{IVPm1}).
To obtain the transmitted field on the left of the slab, one puts in (\ref{rightscattering}) $j_\lfp=\eta^\mthird$ and $j_\lfe=-\eta^\mthird$.

We will prove that $\psi(z)$ admits a convergent Puiseux series in $\eta$.  This means that, at the left edge of the slab, $\psi(0)=\tau_\lfl v_\lfl + \tau_\lfe v_\lfe$, where the coefficients $\tau_\lfl(\eta)$ and $\tau_\lfe(\eta)$ are convergent Puiseux series since the basis $\{v_\lfl, v_\lfe\}$ for $V_+(\eta)$ is nondegenerate as $\eta\to0$ (it converges to the basis $\{e_1,e_0\}$ for $\Vo_+$).  In the degenerate eigenmode basis $\{v_\lfp,v_\lfe\}$, the field $\psi(0) = t_\lfp v_\lfp + t_\lfe v_\lfe$, has coefficients that blow up like $\eta^\mthird$ since $t_\lfp=\eta^\mthird\tau_\lfl$ and $t_\lfe=-\eta^\mthird\tau_\lfl+\tau_\lfe v_\lfe$.  Thus we will demonstrate a convergent series
\begin{equation}
  \col{1.1}{t_\lfp}{t_\lfe} \,=\, \eta^\mthird \vec t(\eta) \,=\, \eta^\mthird \left( \vec t_0 + \eta^\third \vec t_1 + \eta^\thirds \vec t_2 + \dots \right)\,
\end{equation}
with $\vec t_0$ a multiple of $[1,\;-1]^T$.

The coefficients $t_\lfp$ and $t_\lfe$ satisfy the first equation in~(\ref{rightscattering}), which, with $(j_\lfp,j_\lfe)=\eta^\mthird(1,-1)$, becomes
\begin{equation}\label{label4493}
  M(\eta) \vec t(\eta) = \eta^\thirds \vec\xi\,,
\end{equation}
where $\vec\xi=C[-1,\zeta]^T$.  The matrix $M$ admits a Puiseux expansion in $\eta$\,,
\begin{equation}
    M(\eta) \,=\, M_0 + \eta^\third M_1 + \eta^\thirds M_2 + \dots\,
\end{equation}
that is singular at $\eta=0$, that is, $M_0$ is noninvertible.  Recursive calculation of the coefficients $\vec t_n$ is facilitated by defining the vectors
\begin{equation}
  \xi_{-1} = \col{1.1}{1}{-1},\quad
  \xi_0 = \col{1.1}{1}{1},\quad
  \xi_1 = \col{1.1}{1}{\zeta^2},\quad
  \xi_2 = \col{1.1}{1}{\zeta},\quad
  \xi = C \col{1.1}{-1}{\zeta}.
\end{equation}
%
%
Using the expansions (\ref{eigenvectors}), one calculates the first three matrix coefficients of $M(\eta)$,
\begin{eqnarray}
  M_0 &=& [e_0,\To e_0] \xi_0 \xi_0^T \,\\
  M_1 &=& [e_0,\To w_1] \xi_0 \xi_1^T + [w_1,\To e_0] \xi_1 \xi_0^T\,\\
  M_2 &=& [e_0,\To w_2] \xi_0 \xi_2^T + [w_1,\To w_1] \xi_1 \xi_1^T + [w_2,\To e_0] \xi_2 \xi_0^T\,.
\end{eqnarray}
Notice that range $\rm{Ran}(M_0)=\sspan\{\xi_0\}$ and $\rm{Null}(M_0)=\sspan\{\xi_{-1}\}$, that $\xi_{-1}^T\xi_0=0$, and that
\begin{equation}\label{label8856}
  (\xi_{-1}, M_1\xi_{-1})=0\,.
\end{equation}

Problem (\ref{label4493}) is treated in Lemma~\ref{lemma:Mtb}, in which it is shown that $t(\eta)$ admits a convergent Puiseux series by identifying $z$ in the lemma with $\eta^\third$.  Indeed, the condition (\ref{label8856}) coincides with assumption (c) of the lemma, and assumption (d) is, in the present context,
\begin{multline}
  \det\mathcal{M}_0
  = \det \mat{1.2}{(\xi_{-1},M_{2}\xi_{-1})}{(\xi_{-1},M_{1}\xi_0)}{(\xi_0,M_{1}\xi_{-1})}{(\xi_0,M_0 \xi_0)}
  = \det \mat{1.2}{(1-\zeta^2)^2[w_1,\To w_1]}{2(1-\zeta^2)[w_1,\To e_0]}{2(1-\zeta^2)[e_0,\To w_1]}{4[e_0,\To e_0]} \\
  = 4(1-\zeta^2)^2 \left( [e_0,\To e_0][w_1,\To w_1]-[e_0,\To w_1][w_1,\To e_0] \right) = 0\,,
\end{multline}
which is taken as an assumption in Theorem~\ref{thm:rightscatteringreg}.  That $\xi_{-1}^T\xi=C\zeta^2\not=0$ corresponds to the condition $s^Tb_2\not=0$ in the lemma and therefore (\ref{label9965}) and (\ref{label3524}) hold.  This means that $t_0\in\sspan\{\xi_{-1}\}$, that is
\begin{equation}
  t_0 = c_0\col{1}{1}{-1}
\end{equation}
for some $c_0$ and that if $c_0\not=0$, then necessarily $t_1\not\in\sspan\{\xi_{-1}\}$, that is
\begin{equation}
  t_1 = c_1\col{1}{1}{-1} + d_1\col{1}{1}{1}
  \qquad
  \text{($d_1\not=0$ if $c_0=0$)}.
\end{equation}
From this, one infers that $\psi(0)$ has a nonzero limit as $\eta\to0$:
\begin{multline}
  \psi(0) \,=\, \eta^\mthird \sum_{n=0}^\infty c_n \eta^{n/3}(v_\lfp-v_\lfe) \,+\, \eta^\mthird \sum_{n=1}^\infty d_n \eta^{n/3}(v_\lfp+v_\lfe) \\
  = \, v_\lfl \sum_{n=0}^\infty c_n\eta^{n/3} + (v_\lfp + v_\lfe) \sum_{n=0}^\infty d_{n+1} \eta^{n/3}
  \,\to\, c_0 e_1 + 2d_1 e_0
  \qquad \text{as }\, \eta\to0\,.
\end{multline}
This limit is nonzero since either $c_0$ or $d_1$ is nonzero.  Therefore
\begin{equation}
  \psi(L) = T\psi(0) = v_\lfl + \tilde r_\rtp v_\rtp + \tilde r_\rte v_\rte
\end{equation}
is also a convergent Puiseux series in $\eta$, and since $\varepsilon^1(v_\lfl(0))=k_1(1-\zeta^2)$, $v_\rtp(0)=e_+$, and $v_\rte(0)=e_0$, the reflection coefficients $\tilde r_\rtp(\eta)$ and $\tilde r_\rte(\eta)$ are also convergent Puiseux series.

The incident field in part\,2a is $(k_1(1-\zeta^2))^{-1}$ times the incident field we assumed to obtain the scattering field $\psi$.  Thus value of the total scattering field at $z=L$ for part\,2a is $(k_1(1-\zeta^2))^{-1}\psi(L)$, whose limit as $\eta\to0$ lies in $\Vo$ and has $e_1$-component equal to $(k_1(1-\zeta^2))^{-1}\varepsilon^1(v_\lfl(0))=1$.  

Now consider the incident field in part~2b of Theorem~\ref{thm:rightscatteringreg}, for which $j_\lfp=\eta^\mthirds$ and $j_\lfe=\zeta^2\eta^\mthirds$ so that the first equation in (\ref{rightscattering}) becomes
\begin{equation}\label{label3487}
  M(\eta)t(\eta) = -C\eta^\third\xi_0
\end{equation}
with $[t_\lfp\; t_\lfe]^T=\eta^\mthird t(\eta)$.  Since $\mathrm{Ran}(M_0)=\sspan\{\xi_0\}$, assumption (b) of Lemma~\ref{lemma:Mtb} as well as the condition $b_1\not=0$ are satisfied, and one again obtains $\psi(0)$ as a Puiseux series that converges to a nonzero vector in~$\Vo_-$.  At $z=L$, one has
\begin{equation}\label{label2250}
  \psi(L) = T\psi(0) \,=\, \eta^\mthirds v_\lfp + \eta^\mthirds\zeta^2 v_\lfe + r_\rte(\eta) v_\rte + r_\rtp(\eta) v_\rtp\,,
\end{equation}
which is a (bounded) Puiseux series whose limit as $\eta\to0$ is $\To[\psi(0)|_{\eta=0}]\in\To[\Vo_-]$  The form of the projections (\ref{Pre},\ref{Prp}) shows that $r_\rte(\eta)=O(\eta^\mthirds)$ and $r_\rtp(\eta)=O(1)$.  By using the expansions (\ref{eigenvectors}) of the eigenvectors $v_{\rtp,\rte,\lfp,\lfe}$ and enforcing $\psi(L)=O(1)$, one finds that the $\eta^\mthird$ term of $r_\rte(\eta)$ vanishes and obtains $r_\rte(\eta)=\zeta\eta^\mthirds+c+O(\eta)$ for some number $c$.  One then divides by $C$ to obtain the stated value of $\psi_\ltwo(L)$ in part\,2b for the given incident field.

One computes from (\ref{label2250}) that
\begin{equation}
  \psi(L) = 3w_2 + ce_0 + r_\rtp(0) e_+ + O(\eta),
\end{equation}
so that
\begin{equation}
  \varepsilon^2(\psi_\ltwo(L)) = (1/C)3\hspace{1pt}\varepsilon^2(w_2) = 3k_1^2/C = 1,
\end{equation}
where we have used that $C=3k_1^2$ (see (\ref{UniqueConstantInNormalizedFluxFormPerturbed})) and that $\varepsilon^2(w_2)=-[e_0,w_2]=k_1^2$ (see \ref{LowOrderExpanEigenvecSpanRel}).  This proves that $\psi_\ltwo(L)=\mathring v_\ltwo$ and thus that $\mathring\psi_\ltwo(z)$ and $\lim_{\eta\to0}\psi_\ltwo(z)$ satisfy the same initial-value problem and are thus equal to each other.

Part\,3 follows from the linear independence of $\lim_{\eta\to0}\psi_\lone=\mathring v_\lone$ and $\lim_{\eta\to0}\psi_\ltwo=\mathring v_\ltwo$, so that $\psi_\lone$ and $\psi_\ltwo$ span the linear space of right-scattering fields, and the fact that $\psi_\lone^\inc$ and $\psi_\ltwo^\inc$ form a basis for $V_-(\eta)$.

To prove part\,4, assume first that $\mathring\psi(L):=\lim_{\eta\to0}\psi(L;\eta)\in\To[\Vo_-]\cap\Vo$ so that $\mathring\psi(L)$ is a multiple of $\mathring v_\lone$.  Thus $\psi(L;\eta)$ can be identified with $v_\lone(\eta)$ of equation (\ref{lone}) and the projections (\ref{loneprojections}) establish the claimed form of the incident field.  Conversely, assuming the given form of the incident field and $[t_\lfp, t_\lfe]=\eta^\mthird t(\eta)$, the first equation of (\ref{rightscattering}) is
\begin{equation}
  M(\eta)t(\eta) = \eta^\thirds\xi(\eta),
\end{equation}
with $\xi=[\alpha(\eta),\beta(\eta)]^T$.  Lemma~\ref{lemma:Mtb} guarantees that $t(\eta)$ is a Puiseux series with $t_0=c_0[1,-1]^T$.  Thus, since $v_\lfl = v_\lfp-v_\lfe$, one obtains
\begin{equation}
  \psi(0) = c_0 v_\lfl + \sum_{n=0}^\infty \big( (c_{n+1}+d_{n+1}) v_\lfp + (-c_{n+1}+d_{n+1}) v_\lfe \big)\eta^{-n/3},
\end{equation}
which shows that $\psi(0)$ admits a (bounded) Puiseux series with limit in $\Vo_-$, and therefore $\psi(L)=T\psi(0)$ admits a Puiseux series with $\lim_{\eta\to0}\psi(L)\in\To[\Vo_-]$.  If $\lim_{\eta\to0}\psi(L)$ has a nonzero $e_2$ component, then $[e_0,\psi(L)]\not=0$, so that the projections $\Plp\psi(L)$ and $\Ple\psi(L)$ are of strict order $\eta^\mthirds$, contradictory to the assumption that the incident field is $O(\eta^\mthird)$.  Thus $\varepsilon^2(\psi(L)|_{\eta=0})=0$ and therefore $\lim_{\eta\to0}\psi(L)\in\Vo$.  Since $\lim_{\eta\to0}\psi(L)$ also lies in $\To[\Vo_-]$, it must be a multiple of $\mathring v_\lone$.

\smallskip
{\itshape Proof of Theorem~\ref{thm:leftscatteringres}} (resonant scattering from the left).
We begin by proving that the condition (\ref{generic2}) {\itshape i.e.},
\begin{equation}
[e_+, \To e_1] \not= 0,
\end{equation}
is equivalent in the resonant case $\mathring D = 0$ to any one of the conditions $[e_+, \To^{-1} e_1] \not= 0$, $[e_+, \To w_1][e_+, \To^{-1} w_1]\not= 0$, $\To[\mathring V_+]\not= \mathring V_+$, or $\To[\mathring V_-]\not= \mathring V_-$. Assume that $\mathring D = 0$. To prove that statement we shall prove the contrapositive. Suppose $[e_+, \To e_1] = 0$. Then this implies $\To e_1\in \sspan\{e_0,e_1,e_2\}$ and since $\mathring D = 0$ then it follows by Theorem \ref{thm:D} that we must have $\To[\mathring V_-]=V_-$ so that $[e_+, \To^{-1} e_1] = 0$. As we must also have $[\To^{-1}e_+, e_1] = 0$ then this implies $\To^{-1} e_+\in \sspan\{e_+,e_0,e_2\}$ so again by Theorem \ref{thm:D} that we must have $\To^{-1}[\mathring V_+]=V_+$. Hence by (\ref{LowOrderExpanEigenvecSpanRel}) it follows that $[e_+, \To w_1]=[e_+, \To^{-1} w_1]= 0$.  Conversely, suppose any one of those conditions is not true. Then it follows trivially from Theorem \ref{thm:D} and (\ref{LowOrderExpanEigenvecSpanRel}) that $[e_+, \To e_1] \not= 0$. This proves the equivalence of those conditions in the resonant case.

For part~1, since $e_0\in\Vo_+$ and $\To^{-1}e_0=\ell^{-1}e_0$, we have $e_0\in\To^{-1}[\Vo_+]$.  Also, $e_+\in\Vo_+$, and by Theorem~\ref{thm:D} $\To^{-1} e_+\in\Vo$, so that $\To^{-1}e_+ = ae_++be_0+ce_1$.  Since $1=[e_+,e_+]=[\To^{-1}e_+,\To^{-1}e_+]=|a|^2-|c|^2$, we have $a\not=0$, and thus the vector $a^{-1}\To^{-1}e_+$ has $e_+$ component equal to $1$.  By adding the appropriate multiple of $e_0$, one obtains $\mathring w_+$. That $\varepsilon^{1}(\mathring w_+)\not = 0$ follows from the fact that the condition (\ref{generic2}) implies that $\Vo_+\not=\To^{-1}[\Vo_+]=\sspan\{e_0,\mathring w_+\}$.

Consider the scattering problem from the left with incident field at $z=0$ equal to the vector $\eta^\third v_\rtp$.  The first equation of (\ref{leftscattering}) becomes
\begin{equation}\label{label2278}
  M(\eta)t(\eta) = \eta^\third \epsilon_1\,,
\end{equation}
in which $\epsilon_1=[1\;0]^T$ and $\epsilon_2=[0\;1]^T$.
In this case, we have
\begin{equation}
  M_0 = [\To e_+, e_+]\mat{1}{1}{0}{0}{0},
\end{equation}
which is necessarily nonzero by Theorem~\ref{thm:D}(f), and
\begin{equation}
  \mathrm{Ran}(M_0)= \sspan\{\epsilon_1\}\,,
  \qquad
  \mathrm{Null}(M_0)= \sspan\{\epsilon_2\}\,.
\end{equation}
One computes that the $22$-entry of $M_1$ vanishes, as required by Lemma~\ref{lemma:Mtb} vanishes:
\begin{equation}
  (\epsilon_2,M_1\epsilon_2) = [\To e_0, \zeta w_1]+[\zeta^2\To w_1,e_0] = [\ell e_0,\zeta w_1]+[\zeta^2 w_1,\ell^{-1}e_0]=0
\end{equation}
since $w_1\in\sspan\{e_0,e_1\}$ and $[e_0,e_0]=0=[e_0,e_1]$.  The matrix $\mathcal{M}_0$ in the lemma is
\begin{equation}
  \mathcal{M}_0 = \mat{1.1}{\epsilon_2^T M_2\epsilon_2}{\epsilon_2^T M_1\epsilon_1}{\epsilon_1^T M_1\epsilon_2}{\epsilon_1^T M_0\epsilon_1}
  =\mat{1.1}{0}{\zeta[\To w_1,e_+]}{\zeta[\To e_+,w_1]}{[\To e_+,e_+]} \,,
\end{equation}
and the condition $\det\mathcal{M}_0\not=0$ in the lemma is
\begin{equation}
  [\To e_+, w_1][\To w_1, e_+] \not=0\,,
\end{equation}
which is the generic assumption in Theorem~\ref{thm:leftscatteringres}.

Problem \ref{label2278} is handled by problem (i) in the lemma, which gives $t(\eta)$ as a convergent Puiseux series with $t_0$ a nonzero multiple of $\epsilon_2$; one computes that this multiple is equal to $(\zeta[\To e_+,w_1])^{-1}$.  Thus $\psi(L) = t_\rtp v_\rtp + t_\rte v_\rte$ is a convergent Puiseux series whose value at $\eta=0$ is equal to $(\zeta[\To e_+,w_1])^{-1} e_0$.  Since $\To e_0 = \ell e_0\not=0$, $\lim_{\eta\to0}\psi(0) = c_\rone e_0$, with
\begin{equation}
  c_\rone = \frac{1}{\ell\,\zeta\,[\To e_+,w_1]}\,.
\end{equation}

We have proved that the incident field $F(z)\eta^\third v_\rtp e^{ik_\rtp z}$ results in a total field, say $\phi(z)$ that is a Puiseux series in $\eta$, which, at $z=0$ is equal to
\begin{equation}
  \phi(0) = c_\rone e_0 + \eta^\third v_\rone + O(\eta^\thirds)
\end{equation}
for some vector $v_\rone$, and whose projections onto the rightward modes are
\begin{eqnarray*}
  && \Prp\phi(0) \,=\, \eta^\third v_\rtp = \eta^\third e_+ + O(\eta^\thirds)\,, \\
  && \Pre\phi(0) \,=\, 0\,.
\end{eqnarray*}
Comparing this to (\ref{resproj1},\ref{resproj2}) yields
\begin{equation}\label{label4747}
  [e_+,v_\rone]=1\,,
  \qquad
  [e_0,v_\rone]=0\,,
\end{equation}
that is, $\varepsilon^+(v_\rone)=1$ and $\varepsilon^2(v_\rone)=0$, so that $v_\rone\in\Vo$. 

In part 2a of Theorem~\ref{thm:leftscatteringres}, the incident field is $\eta^\third(c_\rone)^{-1}F(z)v_\rtp e^{ik_\rtp z}$, and therefore the scattering field $\phi_\rone(z)$ is equal to $(c_\rone)^{-1} \phi(z)$.  One thus obtains
\begin{equation}
  \lim_{\eta\to0} \phi_\rone(0) = \lim_{\eta\to0} (c_\rone)^{-1} \phi(0) = e_0\,,
\end{equation}
so that $\lim_{\eta\to0}\phi_\rone(z)$ satisfies the initial-value problem for the field $\phi_\rzero(z)$ given in part~1 of the theorem.

\smallskip
Now consider the incident field $F(z)v_\rte e^{ik_\rte z}$, which results in the equation
\begin{equation}
  M(\eta)t(\eta) = -C\zeta^2\eta^\thirds\epsilon_2\,.
\end{equation}
This problem is handled by problem (ii) of Lemma~\ref{lemma:Mtb}.  The conditions there are satisfied because the $11$-entry of $\mathcal{M}_0$ vanishes, while the $22$-entry does not by Theorem~\ref{thm:D}(f), where $[\To e_+,e_+]=t_{++}$.  Again, $\psi(0)$ is found to be a convergent Puiseux series in $\eta$ whose value at $\eta=0$ is equal to $c_\rtwo e_0$, with
\begin{equation}
  c_\rtwo = -\frac{\zeta\,[\To e_+,e_+]}{\ell\,[\To e_+,w_1][\To w_1,e_+]}\,.
\end{equation}
We have proved that the incident field $F(z)v_\rte e^{ik_\rte z}$ results in a total field $\phi(z)$ that is a Puiseux series in $\eta$ such that
\begin{equation}
  \phi(0) = c_\rtwo e_0 + \eta^\third v_\rtwo + O(\eta^\thirds)
\end{equation}
for some vector $v_\rtwo$, and whose projections onto the rightward modes are
\begin{eqnarray*}
  && \Prp\phi(0) \,=\, v_\rte = e_0 + O(\eta^\third)\,, \\
  && \Pre\phi(0) \,=\, 0\,.
\end{eqnarray*}
Comparing this to (\ref{resproj1},\ref{resproj2}) yields
\begin{equation}\label{label4748}
  [e_+,v_\rtwo]=0\,,
  \qquad
  [e_0,v_\rtwo]=0\,,
\end{equation}
that is, $\varepsilon^+(v_\rtwo)=0$ and $\varepsilon^2(v_\rtwo)=0$, so that $v_\rtwo\in\Vo_-$.
In part 2b of Theorem~\ref{thm:leftscatteringres}, the incident field is $(c_\rtwo)^{-1}F(z)v_\rte e^{ik_\rte z}$, and therefore the scattering field $\phi_\rtwo(z)$ is equal to $(c_\rtwo)^{-1} \phi(z)$.  One thus obtains
\begin{equation}
  \lim_{\eta\to0} \phi_\rtwo(0) = \lim_{\eta\to0} (c_\rtwo)^{-1} \phi(0) = e_0\,,
\end{equation}
so that $\lim_{\eta\to0}\phi_\rone(z)$ satisfies the initial-value problem for $\phi_\rzero(z)$.

To prove part~3 of the theorem, observe that
\begin{equation}
  \phi_+(z) \,=\, \eta^\mthird c_\rone\left( \phi_\rone(z) - \phi_\rtwo(z) \right)
\end{equation}
is a Puiseux series since the leading-order terms cancel, and one obtains
\begin{equation}
  \phi_+(0) \,=\, v_\rone - \frac{c_\rone}{c_\rtwo}\, v_\rtwo\,.
\end{equation}
From (\ref{label4747},\ref{label4748}), one obtains
\begin{equation*}
  [e_+,\phi_+(0)]=1\,,
  \qquad
  [e_0,\phi_+(0)]=0\,,
\end{equation*}
so that $\phi_+(0) = \mathring w_+ + c e_0$ for some constant $c$, which is a combination of the initial values of the problems stipulated in part~1 of the theorem that define $\mathring\phi_+$ and $\mathring\phi_\rzero$.  Therefore $\phi_+|_{\eta=0} = \mathring\phi_+ + c\mathring\phi_\rzero$.  The incident field that results in the total scattering field $\phi_+$, evaluated at $z=0$, is
\begin{equation}
  \eta^\mthird\, c_\rone \left( (c_\rone)^{-1} \eta^\third v_\rtp - (c_\rtwo)^{-1} v_\rte \right)
  \,=\, v_\rtp + \eta^\mthird \zeta \frac{[\To w_1,e_+]}{[\To e_+, e_+]} v_\rte\,.
\end{equation}
Part~4 of Theorem~\ref{thm:leftscatteringres} comes from the linear independence of $\phi_\rone$ and $\phi_\rtwo$ and the definition of $\phi_+$.

\smallskip
{\itshape Proof of Theorem~\ref{thm:rightscatteringres}} (resonant scattering from the right).
Part\,1 is similar to part\,1 of Theorem~\ref{thm:leftscatteringres}.  We have $e_0\in\To[\Vo_-]$ and $\To e_1 = ae_++be_0+ce_1$ and $-1=[e_1,e_1]=[\To e_1,\To e_1]=|a|^2-|c|^2$, so that $c\not=0$ and we can define $\mathring w_-$ to be $c^{-1}\To e_0$ plus an appropriate multiple of $e_0$.  Notice that
\begin{equation}\label{}
  c = -[e_1,\To e_1].
\end{equation}
%

%
%
%
%

The first equation in (\ref{rightscattering}) is
\begin{equation}\label{label7644}
  M(\eta) \col{1}{t_\lfp}{t_\lfe} \,=\, -C\eta^\mthirds\col{1}{j_\lfp}{\zeta j_\lfe}
\end{equation}

The first two coefficients in the Puiseux expansions of $M(\eta)$ vanish when $\To e_0=\ell e_0$.  The $11$-entry of $M(\eta)$, for example,~is
\begin{multline}
  [v_\lfp, T v_\lfp] \,=\, [e_0,\To e_0] + \eta^\third \left( [e_0,\To w_1] + [w_1,\To e_0] \right) + O(\eta^\thirds)\\
   \,=\, \ell[e_0,e_0] + \eta^\third\left( \bar\ell^{-1} [e_0,w_1] + \ell[w_1,e_0] \right) + O(\eta^\thirds) \,=\, O(\eta^\thirds)
\end{multline}
since $[e_0,e_0]=0$ and $[e_0,w_1]=0$.  Set $N(\eta)=\eta^\mthirds M(\eta)=\sum_{\ell=0}^\infty M_{\ell+2}\,\eta^{\ell/3}$, so that (\ref{label7644}) becomes
\begin{equation}
  N(\eta) \col{1}{t_\lfp}{t_\lfe}  = -C\col{1}{j_\lfp}{\zeta j_\lfe}\,.
\end{equation}
One computes that the matrix $N(0)$ is
\begin{equation}\label{M2}
  N(0) = M_2 \,=\, -k_1^2
          \mat{1.5}{\hat \ell+c+\ell}{\zeta\hat\ell+\zeta^2c+\ell}{\hat\ell+\zeta^2c+\zeta\ell}{\zeta\hat\ell+\zeta c+\zeta\ell}
  \qquad
  (\hat\ell := \bar\ell^{-1},\;\,c:=-[e_1,\To e_1])
\end{equation}
by using
\begin{equation}
  [e_0,\To w_2] = -k_1^2\hat\ell,
  \qquad
  [w_1,\To w_1] = k_1^2[e_1,\To e_1],
  \qquad
  [w_2,\To e_0] = -k_1^2\ell\,.
\end{equation}
The determinant is
\begin{equation}
  \det M_2 \,=\, \frac{3k_1^4\zeta}{\bar\ell} \left( \ell - [e_1,\To e_1] \left( |\ell|^2+1 \right)\right),
\end{equation}
which, by assumption, does not vanish.
Thus, one obtains $t_\lfp$ and $t_\lfe$ and therefore also the total field $\phi(0)$ at $z=0$ as convergent Puiseux series in $\eta^\third$,
\begin{equation*}
   \col{1}{t_\lfp}{t_\lfe} = \col{1}{t_\lfp^0 + O(\eta^\third)}{t_\lfe^0 + O(\eta^\third)}  \,,
\end{equation*}
\begin{equation}\label{label4639}
  \phi(0) = t_\lfp v_\lfp + t_\lfe v_\lfe = (t_\lfp^0+t_\lfe^0)e_0 + O(\eta^\third).
\end{equation}
If $j_\lfp$ and $j_\lfe$ do not both vanish, then at least one of $t_\lfp$ and $t_\lfe$ is nonzero at $\eta=0$.

Let the incident field at $z=L$ be the scaled mode $\eta^\third v_\lfl = v_\lfp-v_\lfe$, then $j_\lfp=1$ and $j_\lfe=-1$.  One computes that
\begin{equation}\label{label9922}
   [1 \;\;1] \col{1}{t_\lfp^0}{t_\lfe^0}  =  -C [1 \;\;1]\, M_2 \col{1}{1}{-\zeta} \,=\, C k_1^2 \left(1-\zeta^2\right) \hat\ell \left(2+|\ell|^2\right) \,\not=\, 0 \,,
\end{equation}
which implies that $[t_\lfp^0\;\;t_\lfe^0]$ is not a multiple of $[1\;-\!1]$.  Thus by (\ref{label4639}),
\begin{equation}\label{}
  0\not= \lim_{\eta\to0}\phi(0) \in \sspan\{e_0\} \,.  
\end{equation}

Now let $[j_1\;\,j_2]$ be the (nonzero) vector such that
\begin{equation}\label{j1j2}
  -C M_2 \col{1}{j_1}{\zeta\, j_2} = \col{1}{1}{-1}.
\end{equation}
By choosing the incident field such that
\begin{eqnarray*}
 \tilde j_\lfp &=& \eta^\mthird j_1 \\
 \tilde j_\lfe &=& \eta^\mthird j_2,
\end{eqnarray*}
one obtains
\begin{eqnarray*}
  t_\lfp &=& \eta^\mthird \left( 1+a\,\eta^\third+O(\eta^\thirds) \right) \\
  t_\lfe &=& \eta^\mthird \left( -1+b\,\eta^\third+O(\eta^\thirds) \right)
\end{eqnarray*}
for some numbers $a$ and $b$.  Thus,
\begin{multline*}
  \phi(0) \,=\, \eta^\mthird \left( 1+a\,\eta^\third+\dots \right) \left( e_0 + \eta^\third w_1 + \dots \right)
                   + \eta^\mthird \left( -1+b\,\eta^\third+\dots \right) \left( e_0 + \eta^\third \zeta^2 w_1 + \dots \right) \\
     = (a+b) e_0 + (1-\zeta^2)w_1 + O(\eta^\third) = c' e_0 + (1-\zeta^2)k_1\, e_1 + O(\eta^\third),
\end{multline*}
for some constant $c'$.  Importantly, the $e_1$ component of $\phi(0)$ is nonzero, so that
\begin{equation}\label{}
  \phi(L) \,=\, (1-\zeta^2)k_1\, \mathring w_- + c'e_0
\end{equation}
for some constant $c'$.  By the initial-value problems defining $\mathring w_-$ and $\mathring w_\lzero$, one has
\begin{equation*}
  \frac{1}{(1-\zeta^2)k_1} \phi(z) = \phi_-(z) + c'\phi_\lzero(z).
\end{equation*}
for some constant $c$.
By solving for $j_1$ and $j_2$ and dividing $\tilde j_\lfp$ and $\tilde j_\lfe$ by $(1-\zeta^2)k_1$ (and recalling $C=3k_1^2$), one obtains the coefficients $j_\lfp$ and $j_\lfe$ of the incident field $\phi_-^\inc$ in the theorem.  One calculates that the sum of the components of the vector on the right-hand side of (\ref{jpje}) is equal to $(\zeta-\zeta^2)\hat\ell(2+|\ell|^2)\not=0$.  Thus $[j_\lfp\;j_\lfe]$ is not a multiple of $[1\;-\!1]$, and this establishes part~4 of the theorem.
\end{proof}

\section{Analytic perturbation for scattering in the slow-light regime}
\label{sec:analyticperturbationtheory}

This section establishes several results in this paper, Theorems \ref{thm:localbandstrucandjordanstructure} and \ref{thm:noguidemodesundernondegeneracycondition} and Propositions \ref{prop:PuiseuxSeriesEigenvAsympt} and \ref{prop:PuiseuxSeriesProjAsympt}, by applying analytic perturbation theory for the scattering problem to the slow-light regime.

\begin{proof}[Proof of Theorem \ref{thm:localbandstrucandjordanstructure} (sec.~\ref{sec:EMBlochwaves})]
Define the analytic matrix-valued function (suppressing dependence on $\kk=(k_1,k_2)$)
\begin{align}
L(k,\omega)=\frac{c}{16\pi}J(K(\omega)-k I),\;\;\;|\omega-\omega^0|<r,\;k\in\mathbb{C}\label{DefTheHermitianAnalyticMatrixFunc}
\end{align}
for some $r>0$ (see (\ref{J}) and sec.~\ref{sec:periodic} for the definitions of $J$ and $K$), which satisfies $L(\overline{k},\overline{\omega})^\ast=L(k,\omega)$ on its domain (such an analytic function is said to be Hermitian as in \cite[p.~330]{HrynivLancaster1999}). At $k=k^0$, the Hermitian analytic matrix-valued function $L(k^0,\omega)$ has an eigenvalue $\omega=\omega^0$ ({\itshape i.e.}, $\det L(k^0,\omega^0)=0$) and this eigenvalue is of positive definite type \cite[p.~333]{HrynivLancaster1999}, which means that the following quadratic form on $\ker L(k^0,\omega^0)=\ker(K(\omega^0)-k^0 I)$ is positive-definite:
\begin{align}
\psi\mapsto \left(\psi, \frac{\partial L}{\partial\omega}(k^0,\omega^0)\psi\right)=[\psi,K^{\prime}(\omega^0)\psi],\;\;\psi\in \ker L(k^0,\omega^0). \label{PositiveDefiniteType_kzeroEigenvalue}
\end{align}
This follows from positivity of the energy density described in \cite[sec.~III.B]{ShipmanWelters2013} for layered media and the results \cite[Theorems 3.2 and A.1 and proof of Theorem 3.3]{ShipmanWelters2013}, which connect the frequency derivative $\frac{dK}{d\omega}$ of the indicator matrix $K(\omega)$ to the energy density of electromagnetic fields of the form
$(\EE(z),\HH(z))e^{i(k_1x+k_2y-\omega t)}$.

Let $g$ be the geometric multiplicity of the eigenvalue $\omega=\omega^0$ of $L(k^0,\omega)$ (as defined in \cite[p.~325]{HrynivLancaster1999} or \cite[sec.~3.1, Def.~6, 7]{Welters2011a} by $g=\dim \ker L(k^0,\omega^0)$), and let $M$ denote its algebraic multiplicity (as defined in \cite[p.~328]{HrynivLancaster1999} or \cite[sec.~3.1, Def.~6]{Welters2011a}), which by \cite[Lemma 2.4]{HrynivLancaster1999} and \cite[sec.~3.1, Def.~6]{Welters2011a} is given by the order of the zero of $\det L(k^0,\omega)$ at $\omega=\omega_0$.  We will show that $M$ is finite by establishing $M=g$.

To prove that $M=g$, we must show that $\omega=\omega^0$ is an eigenvalue of $L(k^0,\omega)$ of finite algebraic multiplicity which is also a semisimple eigenvalue of $L(k^0,\omega)$ (as defined in \cite[p.~327]{HrynivLancaster1999} or \cite[sec.~3.1, Def.~7]{Welters2011a}). This will follow from \cite[Proposition 15, p.~59]{Welters2011} if we can show that $\frac{\partial L}{\partial \omega}(k^0,\omega^0)\psi_0=-L(k^0,\omega^0)\psi_1$ has no solutions for $\psi_0\in\ker L(k^0,\omega^0)/\{0\}$ and $\psi_1\in \mathbb{C}^4$. But this follows from the fact that $\omega=\omega^0$ is an eigenvalue of $L(k^0,\omega)$ of positive type; for if $(\psi_1,\psi_2)$ were a solution, then
\begin{equation}
0<\left(\psi_0,\frac{\partial L}{\partial\omega}(k^0,\omega^0)\psi_0\right)=-(\psi_0,L(k^0,\omega^0)\psi_1)\\
=-(L(k^0,\omega^0)\psi_0,\psi_1)=0,
\end{equation}
a contradiction.

Because $L(k,\omega)$ is a Hermitian analytic matrix-valued function and $\omega=\omega^0$ is an eigenvalue of $L(k^0,\omega)$ of finite algebraic multiplicity and of positive type, it follows from \cite[Corollary 3.8, 3.7; Theorem 3.6]{HrynivLancaster1999} that there exists a complex neighborhood $\mathcal{W}_0$ of $\omega=\omega^0$ such that, for all $k$ in a complex neighborhood $\mathcal{K}_0$ of $k^0$, the spectrum of $L(k,\omega)$ in $\mathcal{W}_0$ consists of $g$ eigenvalues $\omega_1(k),\ldots,\omega_g(k)$ (counting multiplicities) that are analytic in $k$ and real-valued for real $k$, with $\omega_j(k^0)=\omega^0$ for all $j$.

These eigenvalue functions are nonconstant functions of $k$.  For suppose this were not the case. Then for some $j$ we would have $\omega_j(k)=\omega_0$ and $0=\det L(k,\omega_j(k))=\det L(k,\omega_0)$ for all $k\in \mathcal{K}_0$, which is not possible since $\det L(k,\omega_0)=\det(\frac{c}{16\pi}J)\det(K(\omega^0)-kI)$ is a nonconstant polynomial in $k$.

Now $k^0$ is an eigenvalue of the matrix $K(\omega^0)$ with geometric multiplicity $\dim\ker (K(\omega^0)-k^0I)=\dim\ker L(k^0,\omega^0)=g$.  Let $m$ denote the algebraic multiplicity of $k_0$ and $m_1\geq \ldots \geq m_g$ its partial multiplicities. This means that $g$ is the number of Jordan blocks in the Jordan normal of $K(\omega^0)$ corresponding to the eigenvalue $k^0$, the dimensions of the $g$ Jordan blocks are $m_1,\ldots, m_g$, and the sum of these dimensions is $m$.  According to the spectral theory of holomorphic matrix functions (see \cite[sec.~3.1.1]{Welters2011}), it follows from \cite[Example 1, p.~49; Proposition 13, pp.~55, 56]{Welters2011} that $k=k^0$ is an eigenvalue of the analytic matrix-valued function $L(k,\omega^0)$ with geometric multiplicity $g$, algebraic multiplicity $m$, and partial multiplicities $m_1, \ldots, m_g$.
From this, \cite[Corollary 4.3]{HrynivLancaster1999}, and the proof of \cite[Theorem 4.2]{HrynivLancaster1999}, we infer that the functions $\omega_j(k)$, $j=1,\ldots,g$ can be reordered such that the order of the zero of $\omega_j(k)-\omega^0$ at $k=k^0$ is $m_j$ for $j=1,\ldots,g$.

We will now prove statement (iv) of Theorem \ref{thm:localbandstrucandjordanstructure}. By definition (see \cite[Theorem 5.1.1]{GohbergLancasterRodman2005}) the sign characteristic of the pair $(K(\omega^0), \frac{c}{16\pi}J)$ corresponding to the eigenvalue $k^0$ of $K(\omega^0)$ is an ordered set $\{\sigma_1,\ldots,\sigma_g\}$ of signs $\pm 1$, which is unique up to permutation of signs corresponding to Jordan blocks for $k^0$ of equal dimensions.  The sign characteristic are the signs in the Jordan-energy-flux canonical form of $K(\omega^0)$ with respect to the indefinite inner-product $[\cdot,\cdot]$ corresponding to the eigenvalue $k^0$ (see \cite[Theorem 5.1.1]{GohbergLancasterRodman2005}).

There is an alternative description of the sign characteristic of the pair $(K(\omega^0), \frac{c}{16\pi}J)$ corresponding to the eigenvalue $k^0$ of $K(\omega^0)$ using \cite[Theorem 5.11.2]{GohbergLancasterRodman2005} which is quite different in character from its definition it in terms of \cite[Theorem 5.1.1]{GohbergLancasterRodman2005}.  Consider the Hermitian matrix-valued function $M(k):=-L(k,\omega^0)$.  From analytic perturbation theory for Hermitian matrices, $M(k)$ has exactly $\zeta_1(k),\cdots, \zeta_g(k)$ eigenvalues (counting multiplicities) in a neighborhood of $k=k_0$, which are analytic in a neighborhood of $k=k_0$ and real for real $k$ and satisfy $\zeta_j(k_0)=0$ for all $j$.  After reordering the eigenvalues, the order of the zero of $\zeta_j(k)$ is $m_j$ for all $j$ and the sign characteristic can be arranged such that $\sigma_j=\frac{\zeta_j^{(m_j)}(k^0)}{|\zeta_j^{(m_j)}(k^0)|}$ for all $j$ (see~\cite[Theorem 5.11.2]{GohbergLancasterRodman2005}).

If $k=k^0$ is an eigenvalue of $L(k,\omega^0)$ with partial multiplicities $m_1,\ldots,m_g$, it follows from the arguments above that the analytic eigenvalues $\mu_1(k),\ldots, \mu_g(k)$ of $L(k,\omega^0)-\mu I$ with $\mu_j(k^0)=0$ are $\mu_j(k)=-\zeta_j(k)$, $j=1,\ldots, g$. Then as described in \cite[p.~337]{HrynivLancaster1999}, the sign characteristic of the eigenvalue $k=k^0$ is the set of numbers $\varepsilon^j:=\frac{\mu_j^{(m_j)}(k^0)}{|\mu_j^{(m_j)}(k^0)|}$.  From this, one obtains $-\varepsilon^j=\frac{\zeta_j^{(m_j)}(k^0)}{|\zeta_j^{(m_j)}(k^0)|}=\sigma_j$ for all $j$.  To complete the proof it suffices to show that $\frac{\omega_j^{(m_j)}(k^0)}{|\omega_j^{(m_j)}(k^0)|}=-\varepsilon^j$ for all $j$. But this follows now from the proof of \cite[Theorem 5.1]{HrynivLancaster1999}.

\end{proof}

\begin{proof}[Proof of Proposition \ref{prop:PuiseuxSeriesEigenvAsympt} (sec.~\ref{sec:UnidirAmbMed})]
Define the analytic matrix-valued function
\begin{align}
L(k,\eta)=\frac{c}{16\pi}J(K(\eta)-k I),\;\;\;|\eta|<r,\;k\in\mathbb{C}\label{DefTheHermitianAnalyticMatrixFunc2}
\end{align}
for some $r>0$, which satisfies $L(\overline{k},\overline{\eta})^\ast=L(k,\eta)$ on its domain (such an analytic function is said to be Hermitian as in \cite[p.~330]{HrynivLancaster1999}). At $k=\kzero$, the Hermitian analytic matrix-valued function $L(\kzero,\eta)$ has an eigenvalue $\eta=0$ ({\itshape i.e.}, $\det L(\kzero,0)=0$) and this eigenvalue is of definite type \cite[p.~333]{HrynivLancaster1999}, that is, the quadratic form on $\ker L(\kzero,0)=\ker(K(0)-\kzero I)$, 
\begin{align}
\psi\mapsto \left(\psi, \frac{\partial L}{\partial\eta}(\kzero,0)\psi\right)=[\psi,K^{\prime}(0)\psi],\;\;\psi\in \ker L(\kzero,0),\label{DefiniteType_kzeroEigenvalue}
\end{align}
is of definite type (positive or negative definite), which follows from hypotheses (\ref{GenericCondition}) and the facts that $\ker(K(0)-\kzero I)=\operatorname{span}\{e_0\}$ and $K^{\prime}(0)^{[\ast]}=K^{\prime}(0)$. Moreover, by definition \cite[p.~325]{HrynivLancaster1999}, the geometric multiplicity of the eigenvalue $\eta=0$ for $L(\kzero,\eta)$ is $\dim \ker L(\kzero,0)$ and is therefore equal to $1$ and, by definition \cite[p.~328]{HrynivLancaster1999} and \cite[Lemma 2.4]{HrynivLancaster1999}, the algebraic multiplicity $M$ of the eigenvalue $\eta=0$ for $L(\kzero,\eta)$ is the order of the zero of $\det L(\kzero,\eta)$ at $\eta=0$.  Since the eigenvalue $\eta=0$ of $L(\kzero,\eta)$ is of definite type, by \cite[Corollary 3.8 \& Theorem 3.6]{HrynivLancaster1999} it is also is semisimple, that is, its geometric and algebraic multiplicities are the same, so that $M=1$ and all the eigenvalues of $L(k,\eta)$ in a sufficiently small complex neighborhood of $(\kzero,0)$ are given by a single analytic function $\eta=\eta(k)$ of $k$ which is real-valued for real $k$.

Now $k=\kzero$ is an eigenvalue of $L(k,0)$ with geometric multiplicity $1$ (since $\dim \ker L(\kzero,0)=1$) and algebraic multiplicity $3$ (since order of zero of $\det L(k,0)$ at $k=\kzero$ is $3$). In fact, it follows from (\ref{DefTheHermitianAnalyticMatrixFunc2}) that the definitions of Jordan chain, geometric multiplicity, partial multiplicities, and algebraic multiplicity for the matrix function $L(k,0)$ in \cite{HrynivLancaster1999} correspond to the standard matrix definitions for the matrix $K(0)$ and its eigenvalues. Thus, from \cite[Corollary 4.3]{HrynivLancaster1999} and the proof of \cite[Theorem 4.2]{HrynivLancaster1999} one finds that the eigenvalue $k=\kzero$ of $L(k,\eta)$ at $\eta=0$ has the completely regular splitting property (CRS) (as defined in \cite[p.~327]{HrynivLancaster1999}), meaning that the real-analytic function $\eta=\eta(k)$ has the series expansion about $k=\kzero$
\begin{align}
\eta(k)=c_3(k-\kzero)^3+\cdots,\;\;\;c_3\not=0,
\label{Proof_PropPuiseuxSeriesEigenvAsympt_EtaExpans}
\end{align}
which can be inverted to obtain a convergent Puiseux series expansion in $\eta^{\third}$ whose branches are of the form
\begin{align}
k_{h}(\eta)=\kzero+\eta^{\third}\kone\zeta^h+\eta^{\thirds}k_2\zeta^{2h}+\eta k_3+\ldots,\;\;\;\kone=\frac{1}{\sqrt[3]{c_3}}\not=0,\;\;\;\eta(k_{h}(\eta))=\eta,\;\;\;h=0,1,2,\;\;\;|\eta|\ll 1.
\label{Proof_PropPuiseuxSeriesEigenvAsympt_EtaInvExpans}
\end{align}
Here, $\zeta=e^{i2\pi/3}$ and $\sqrt[3]{c_3}$ is the real cubic root of $c_3$ such that for real $\eta^{\third}$ near $\eta=0$ the function $k_0(\eta)$ is real, and these are all the eigenvalues of $L(k,\eta)$ ({\itshape i.e.} solutions of $\det L(k,\eta)=0$) in a sufficiently small complex neighborhood of $(\kzero,0)$.

To the eigenvalue Puiseux series $k(\eta)$ of $L(k,\eta)$ whose branches are $k_h(\eta)$ ($h=0,1,2$) there corresponds an eigenvector Puiseux series $\psi(\eta)$ that also has a convergent Puiseux series expansion in $\eta^{\third}$ whose branches $\psi_h(\eta)$ satisfy
\begin{align}
L(k_h(\eta),\eta)\psi_h(\eta)=0,\;\;\;\psi_h(\eta)=\psi_0+\eta^{\third}\zeta^h\psi_1+\eta^{\thirds}\zeta^{2h}\psi_2+\eta\psi_3+\ldots,\;\;\;\psi_0=e_0,\;\;\;h=0,1,2,\;\;\;|\eta|\ll 1.\label{Proof_PropPuiseuxSeriesEigenvAsympt_EigenvectorEtaExpans}
\end{align}
(see \cite[Lemma 2]{LancasterMarkusZhou2003}, \cite[\S 3.3.2]{Welters2011a} or \cite{Welters2011}).
It follows from the analyticity of $K(\eta)$ at $\eta=0$, the definition (\ref{DefTheHermitianAnalyticMatrixFunc2}), the expansion (\ref{Proof_PropPuiseuxSeriesEigenvAsympt_EtaInvExpans}), and the identity (\ref{Proof_PropPuiseuxSeriesEigenvAsympt_EigenvectorEtaExpans}) that
\begin{eqnarray}
&& (K(0)-\kzero I)\psi_0=0,\notag\\
&&(K(0)-\kzero I)\psi_1-\kone\psi_0=0,\notag\\
&& (K(0)-\kzero I)\psi_2-\kone\psi_1-k_2\psi_0=0,\notag\\
&&(K(0)-\kzero I)\psi_3-\kone\psi_2-k_2\psi_1+(K^{\prime}(0)-k_3I)\psi_0=0.\notag
\end{eqnarray}
From this, the fact that the sequence $e_0$, $e_1$, $e_2$ is a Jordan chain for $K(0)=K(0)^{[\ast]}$ corresponding to the real eigenvalue $\kzero$, and the Jordan normal form and flux-form (\ref{jordanform}) with respect to the Jordan basis $\{e_+,e_0,e_1,e_2\}$ it follows that
\begin{gather}
\psi_2=\kone^2e_2+a_1e_1+a_0e_0,\;\;\;\psi_1=\kone e_1+a_+e_0,\;\;\;\text{for some }a_+,a_0,a_1\in\mathbb{C}\,\text{ with }a_1=\kone a_++k_2,\label{Proof_PropPuiseuxSeriesEigenvAsympt_LowOrderExpanEigenvecSpanRel}\\
-\kone=[e_1,\psi_1],\;\;\;-\kone^2=[e_0,\psi_2],\;\;\;-\kone^3=[e_0,K^{\prime}(0)e_0].\nonumber
\end{gather}
Now define for $|\eta|\ll 1$,
\begin{align}
\klp(\eta)=k_0(\eta),\; \vlp(\eta)=\psi_0(\eta),\\
\kre(\eta)=k_1(\eta),\;\vre(\eta)=\psi_1(\eta),\\
\kle(\eta)=k_2(\eta),\;\vle(\eta)=\psi_2(\eta).\label{Proof_PropPuiseuxSeriesEigenvAsympt_DefinitionOfEigenpairsTokzeroGroup}
\end{align}
This establishes the statements (\ref{eigenvalues})--(\ref{LowOrderExpanEigenvecSpanRel}) of this proposition for the eigenpairs $(k_j(\eta),v_j(\eta))$ of $K(\eta)$ for $j\in\{-p,+e,-e\}$ when $|\eta|\ll 1$.

Next, there exists a vector-valued function $\nu(\delta)$ which is analytic in a complex neighborhood of $\delta=0$ such that $\psi_h(\eta)=\nu(\zeta^{h}\eta^{\third})$, $h=0,1,2$, for any branch of $\eta^{\third}$ near $\eta=0$.  By (\ref{Proof_PropPuiseuxSeriesEigenvAsympt_EigenvectorEtaExpans}), its power series expansion is
\begin{align}
\nu(\delta)=\psi_0+\delta\psi_1+\delta^2\psi_2+\ldots,\;\;\;\psi_0=e_0,\;\;\;|\delta|\ll 1.
\end{align}
From this, the identities (\ref{Proof_PropPuiseuxSeriesEigenvAsympt_LowOrderExpanEigenvecSpanRel}), and the flux-form (\ref{jordanform}), one finds that the analytic function $[\nu(\overline{\delta}),\nu(\delta)]$ has the power series expansion about $\delta=0$ given by
\begin{align}
[\nu(\overline{\delta}),\nu(\delta)]=([\psi_0,\psi_2]+[\psi_1,\psi_1]+[\psi_2,\psi_0])\delta^2+\ldots=-3\kone^2\delta^2+\ldots=-3\kone^2\delta^2g(\delta),\;\;\;|\delta|\ll 1,
\end{align}
where $g(\delta)$ is an analytic function in a complex neighborhood of $\delta=0$ which is real-valued for real $\delta$ with $g(0)=1$. Thus there exists an analytic function $f(\delta)$ in a complex neighborhood of $\delta=0$ which is real-valued for real $\delta$ such that $f(0)=1$ and $f(\delta)^2=g(\delta)$ for $|\delta|\ll 1$. Consider now the analytic vector-valued function
\begin{align}
\tilde{\nu}(\delta)=f(\delta)^{-1}\nu(\delta),\;\;\;|\delta|\ll 1.
\end{align}
Then we can redefine the eigenvector Puiseux series $\psi(\eta)$ in our proof above to be $\psi(\eta)=\tilde{\nu}(\eta^{\third})$ instead.  All the properties (\ref{Proof_PropPuiseuxSeriesEigenvAsympt_EigenvectorEtaExpans})--(\ref{Proof_PropPuiseuxSeriesEigenvAsympt_LowOrderExpanEigenvecSpanRel}) continue to hold for this Puiseux series, where $\psi_h(\eta)=\tilde{\nu}(\zeta^h\eta^{\third})$, $h=0,1,2$, for $|\eta|\ll 1$. Moreover, with the definition (\ref{Proof_PropPuiseuxSeriesEigenvAsympt_DefinitionOfEigenpairsTokzeroGroup}) we have, for real $\eta^\third$ near $\eta=0$,
\begin{gather}
[\vlp(\eta),\vlp(\eta)]=[\tilde{\nu}(\eta^{\third}),\tilde{\nu}(\eta^{\third})]=-3\kone^2\eta^{\thirds},\nonumber\\
[\vre(\eta),\vle(\eta)]=[\tilde{\nu}(\zeta\eta^{\third}),\tilde{\nu}(\zeta^2\eta^{\third})]=[\tilde{\nu}(\overline{\zeta^2\eta^{\third}}),\tilde{\nu}(\zeta^2\eta^{\third})]=-3\kone^2(\zeta^2\eta^{\third})^2=-3\kone^2\zeta\eta^{\thirds},\label{Proof_PropPuiseuxSeriesEigenvAsympt_kzeroNormalizedEigenvectors}\\
[\vle(\eta),\vre(\eta)]=\overline{[\vre(\eta),\vle(\eta)]}=-3\kone^2\zeta^2\eta^{\thirds}.\nonumber
\end{gather}

We will now prove the statements (\ref{eigenvalues}) and (\ref{eigenvectors}) of this proposition for some eigenpair $\krp(\eta),\vrp(\eta)$ of $K(\eta)$ when $|\eta|\ll 1$. At $\eta=0$, the Hermitian analytic matrix-valued function $L(k,0)$ has an eigenvalue $k=\kplus$ ({\itshape i.e.}, $\det L(\kplus,0)$) and this eigenvalue is of negative type \cite[p.~333]{HrynivLancaster1999}, {\itshape i.e.}, the quadratic form on $\ker L(\kplus,0)=\ker(K(0)-\kplus I)$,
\begin{align}
\psi\mapsto \left(\psi, \frac{\partial L}{\partial k}(\kplus,0)\psi\right)=-[\psi,\psi],\;\;\psi\in \ker L(\kplus,0) ,\label{DefiniteType_kplusEigenvalue}
\end{align}
is negative definite, which follows from hypotheses (\ref{jordanform}) and the fact $\ker(K(0)-\kplus I)=\operatorname{span}\{e_+\}$.  By definition \cite[p.~325]{HrynivLancaster1999}, the geometric multiplicity of the eigenvalue $k=\kplus$ for $L(k,0)$ is $\dim \ker L(\kplus,0)$ and is therefore equal to $1$; and by definition \cite[p.~328]{HrynivLancaster1999} and \cite[Lemma 2.4]{HrynivLancaster1999}, the algebraic multiplicity $M$ of the eigenvalue $k=\kplus$ for $L(k,0)$ is the order of the zero of $\det L(k,0)$ at $k=\kplus$ which implies $M=1$.  This fact and the fact the eigenvalue $k=\kplus$ of $L(k,0)$ is of definite type, together with the result \cite[Corollary 3.8]{HrynivLancaster1999}, imply that all the eigenvalues of $L(k,\eta)$ in a sufficiently small complex neighborhood of $(\kplus,0)$ are given by a single analytic function $k=\krp(\eta)$ of $\eta$ which is real-valued for real $\eta$.

To the eigenvalue $\krp(\eta)$ there corresponds an analytic eigenvector $\vrp(\eta)$ which satisfies
\begin{align}
L(\krp(\eta),\eta)\vrp(\eta)=0,\;\;\;\vrp(\eta)=\phi_0+\eta\phi_1+\ldots,\;\;\;\phi_0=e_+,\;\;\;|\eta|\ll 1.\label{Proof_PropPuiseuxSeriesEigenvAsympt_kplusEigenvectorEtaExpans}
\end{align}
(see \cite[Lemma 2]{LancasterMarkusZhou2003}, \cite[\S 3.3.2]{Welters2011a} or \cite{Welters2011}).  Since the function $g_{+}(\eta)=[\vrp(\overline{\eta}),\vrp(\eta)]$ is analytic in a complex neighborhood of $\eta=0$ and is real-valued for real $\eta$ with $g_{+}(0)=1$, there exists an analytic function $f_{+}(\eta)$ in a complex neighborhood of $\eta=0$ which is real-valued for real $\eta$ such that $f_{+}(0)=1$ and $f_{+}(\eta)^2=g_{+}(\eta)$ for $|\eta|\ll 1$. Thus, we redefine the analytic eigenvector $\vrp(\eta)$ to be $f_{+}(\eta)^{-1}\vrp(\eta)$ instead which has the same properties as in (\ref{Proof_PropPuiseuxSeriesEigenvAsympt_kplusEigenvectorEtaExpans}) but also satisfies the flux-normalization condition
\begin{align}
[\vrp(\eta),\vrp(\eta)]=1,\;\;\;\eta\in\mathbb{R},\;\;\;|\eta|\ll1.
\label{Proof_PropPuiseuxSeriesEigenvAsympt_kplusNormalizedEigenvector}
\end{align}

Finally, by (\ref{Proof_PropPuiseuxSeriesEigenvAsympt_kzeroNormalizedEigenvectors}) and (\ref{Proof_PropPuiseuxSeriesEigenvAsympt_kplusNormalizedEigenvector}) our proof has constructed eigenvectors $\{\vrp,\vre,\vlp,\vle\}$ having all the properties in the proposition, including the self-flux interactions given by (\ref{fluxmatrix2}) for real $\eta^{\third}$ near $\eta=0$ for the constant $C=3\kone^2$ in (\ref{UniqueConstantInNormalizedFluxFormPerturbed}), except we still need to prove that the flux-orthogonality conditions in (\ref{fluxmatrix2}) hold for real $\eta^{\third}$ near $\eta=0$ and the uniqueness statement for the constant $C$.  Flux orthogonality follows from the eigenvalue expansions (\ref{eigenvalues}, \ref{k1}) for which $\krp$, $\klp$ are real for real  $\eta^{\third}$ near $\eta=0$ with $\kplus\not=\kzero$, together with \cite[Theorem 4.2.4]{GohbergLancasterRodman2005}, which implies that for real $\eta$, the eigenspaces for any two eigenvalues $\lambda,\mu$ of $K(\eta)$ with $\lambda\not=\overline{\mu}$ are orthogonal with respect to $[\cdot,\cdot]$ (since $K(\eta)$ is self-adjoint with respect to the indefinite inner-product $[\cdot,\cdot]$). 

We complete the proof now by proving the uniqueness statement for the constant $C$ in this proposition. Suppose that (\ref{eigenvectors})--(\ref{fluxmatrix2}) hold for some vectors $\{\vrp,\vre,\vlp,\vle\}$.  From the flux form (\ref{jordanform}), one obtains, for real $\eta^{\third}$ near $\eta=0$, the asymptotic expansion as $\eta\rightarrow 0$
\begin{align}
-C\eta^{\thirds}=[\vlp(\eta),\vlp(\eta)]=([\psi_0,\psi_2]+[\psi_1,\psi_1]+[\psi_2,\psi_0])\eta^{\thirds}+o(\eta^{\thirds})=-3\kone^{2}\eta^{\thirds}+o(\eta^{\thirds}).
\end{align}
This implies $C=3\kone^{2}$.
\end{proof}

\begin{proof}[Proof of Proposition \ref{prop:PuiseuxSeriesProjAsympt} (sec.~\ref{sec:UnidirAmbMed})]
Proposition (\ref{prop:PuiseuxSeriesEigenvAsympt}) and standard perturbation results \cite[pp.~233-234, Theorem 2]{Baum85}, \cite[p.~70, Theorem 1.8]{Kato95} on eigenprojections imply that in a complex neighborhood of $\eta=0$, the eigenprojection $\Prp=\Prp(\eta)$ of $K(\eta)$ corresponding to the eigenvalue $\krp(\eta)$ is analytic whereas the eigenprojections $\Plp=\Plp(\eta), \Pre=\Pre(\eta), \Ple=\Ple(\eta)$ of $K(\eta)$ corresponding to the eigenvalues $\klp(\eta),\kre(\eta),\kle(\eta)$, with the asymptotic expansions (\ref{eigenvalues}) in $\eta^{\third}$, are the branches of a Laurent-Puiseux series $P(\eta)$ having the series expansion in $\eta^{\third}$ of the form
\begin{align}
P_h(\eta)=\sum_{j=-m}^{\infty}B_j(\zeta^h\eta^{\third})^j,\;\;\;h=0,1,2,\;\;\;0<|\eta|\ll 1,
\end{align}
where $\Plp(\eta)=P_0(\eta),\Pre(\eta)=P_1(\eta), \Ple(\eta)=P_2(\eta)$. The proposition now follows from this expansion, the asymptotic expansions (\ref{eigenvectors}), the flux interactions (\ref{fluxmatrix2}) for real $\eta^{\third}$ near $\eta=0$, and the identities $\zeta^3=1, \zeta^{-1}=\zeta^2$. 
\end{proof}

\begin{proof}[Proof of Theorem \ref{thm:noguidemodesundernondegeneracycondition} (sec.~\ref{sec:scatteringdefectlayer})]

Let us proceed with a frequency perturbation, so that $\kk(\eta)=\kk^0$ for all $\eta$ near $\eta=0$ with $\frac{d\omega}{d\eta}(0)\not=0$. 
We begin with the generic condition (\ref{GenericCondition}).  From the proof of Theorem \ref{thm:localbandstrucandjordanstructure}, specifically (\ref{PositiveDefiniteType_kzeroEigenvalue}), one obtains $[e_0,\frac{\partial K}{\partial \omega}(\kk^0\omega^0)e_0]>0$.  This, together with the chain rule $\frac{dK}{d\eta}(0)=\frac{d\omega}{d\eta}(0)\frac{\partial K}{\partial \omega}(\kk^0\omega^0)$ yields $[e_0,\frac{dK}{d\eta}(0)e_0]\not=0$.
By Proposition \ref{prop:PuiseuxSeriesProjAsympt}, $T(\eta)\Pl(\eta)-\Pr(\eta)$ has a Laurent-Puiseux series expansion in $\eta^{\third}$ and therefore so does $\det(T(\eta)\Pl(\eta)-\Pr(\eta))$, which converges in an open complex neighborhood of $\eta=0$.

Suppose the Theorem were false.  Then $\det(T(\eta)\Pl(\eta)-\Pr(\eta))=0$ for all $\eta$ in a punctured complex open neighborhood of $\eta=0$.  We will show this leads to a contradiction.
Let $\eta^\third$ denote any branch of the cube root of $\eta$, giving a single-valued and analytic function whose domain is the complex plane with a cut along the negative imaginary axis. By Proposition \ref{prop:PuiseuxSeriesEigenvAsympt} we know that $k_1$ is real and nonzero.  Choose the branch depending on the sign of $k_1$: if $k_1>0$ then demand $\eta^\third>0$ for $\eta>0$, and if $k_1<0$ demand $\eta^\third<0$ for $\eta<0$.  The proof for $k_1<0$ is virtually the same as for $k_1>0$, so we assume $k_1>0$ for simplicity.

By Propositions \ref{prop:PuiseuxSeriesEigenvAsympt} and \ref{prop:PuiseuxSeriesProjAsympt} the matrix $K(\eta)$, which is analytic in $\eta$, has four distinct eigenvalues $\krp(\eta)$, $\klp(\eta)$, $\kre(\eta)$, $\kle(\eta)$, which are also single-valued analytic functions in $\eta$ near but not at $\eta=0$, as are their eigenprojections $\Prp(\eta)$, $\Plp(\eta)$, $\Pre(\eta)$, $\Ple(\eta)$ ({\itshape i.e.}, those convergent Puiseux series in (\ref{eigenvalues}) and corresponding convergent Laurent-Puiseux series in (\ref{Prp})--(\ref{Ple}) using this analytic branch $\eta^{\third}$).  This implies that $T(\eta)\Pl(\eta)-\Pr(\eta)$ is a single-valued analytic function of $\eta$ in this domain. It follows from the asymptotic expansions in (\ref{eigenvalues}) and the choice of the branch of $\eta^\third$, that $\krp(\eta)$ (resp. $\klp(\eta)$) is real for real $0<\eta\ll1$ and that the energy-flux form $[\cdot,\cdot]$ is positive (resp. negative) definite on the corresponding eigenspaces. Moreover, $\Im\kre(\eta)>0$ ($\Im\kle(\eta)<0$) for $\eta>0$ near $\eta=0$ implying $e^{i\kre(\eta)z}\rightarrow 0$ ($e^{i\kle(\eta)z}\rightarrow 0$) as $z\rightarrow \infty$ ($z\rightarrow -\infty$). The theorem now follows from the fact that the zero of  $\det B(\kk^0,\omega(\eta))$ (where $B(\kk^0,\omega(\eta)):=T(\kk^0,\omega(\eta))P_-(\eta)-P_+(\eta)$) at $\eta=\eta^0$ for $0<\eta_0\ll 1$ must be simple by \cite[Theorem 4.4]{ShipmanWelters2013} and the hypothesis $\frac{d\omega}{d\eta}(0)\not=0$. But from our discussion above we must have $\det B(\kk^0,\omega(\eta))=0$ for $0<\eta_0\ll 1$, a contradiction.

\end{proof}

\textbf{Acknowledgment.}
Part of this work was completed while A.\,Welters was a postdoc at Massachusetts Institute of Technology in the Department of Mathematics.


\end{document}